\documentclass[12pt,aps,tightenlines,notitlepage,showpacs,showkeys]{revtex4-1}
\usepackage{amsmath}
\usepackage{amssymb}
\usepackage{bm}
\usepackage[mathscr]{eucal}
\usepackage{theorem}
\usepackage{graphicx}
\usepackage{psfrag}
\usepackage{subfigure}
\usepackage{color}
\usepackage{wasysym}
\include{graphix}
\usepackage{framed}
\usepackage{enumitem}%
\DeclareMathAlphabet{\mathpzc}{OT1}{pzc}{m}{it}

\graphicspath{{figs_pdf/}}

\newtheorem{theorem}{Theorem}

\newenvironment{proof}[1][Proof]{\begin{trivlist}
\item[\hskip \labelsep {\bfseries #1}]}{\end{trivlist}}

\newcommand{\qed}{\nobreak \ifvmode \relax \else
      \ifdim\lastskip<1.5em \hskip-\lastskip
      \hskip1.5em plus0em minus0.5em \fi \nobreak
      \vrule height0.75em width0.5em depth0.25em\fi}

\newcommand{\imag}{\mathrm{i}}
\newcommand{\mathe}{\mathrm{e}}
\newcommand{\mathd}{\mathrm{d}}

\newcommand{\omi}{\omega_{\mathrm{i}}}
\newcommand{\omr}{\omega_{\mathrm{r}}}
\newcommand{\vecx}{\bm{x}}

\newcommand{\vecu}{\bm{u}}

\newcommand{\gmr}{\gamma_{\mathrm{r}}}
\newcommand{\gmi}{\gamma_{\mathrm{i}}}

\newcommand{\radiusone}{\mathpzc{R}_1}
\newcommand{\radiustwo}{\mathpzc{R}_2}

\begin{document}
\title{Global modes  in nonlinear non-normal evolutionary models: exact solutions, perturbation theory, direct numerical simulation, and chaos}
\author{Lennon \'O N\'araigh}
\email{onaraigh@maths.ucd.ie}
\affiliation{School of Mathematical Sciences, University College Dublin, Belfield, Dublin 4}
\affiliation{Complex and Adaptive Systems Laboratory, University College Dublin, Belfield, Dublin 4}


\date{\today}

\begin{abstract}
This paper is concerned with the theory of generic non-normal nonlinear evolutionary equations, with potential applications in Fluid Dynamics and Optics.  Two theoretical models are presented.  The first is a model two-level non-normal nonlinear system  that not only highlights the phenomena of linear transient growth, subcritical transition and global modes, but is also of potential interest in its own right in the field of nonlinear optics.  The second is the fairly familiar inhomogeneous nonlinear complex Ginzburg--Landau (CGL) equation.  The two-level model is exactly solvable for the nonlinear global mode and its stability, while for the spatially-extended CGL equation, perturbative solutions for the global mode and its stability are presented, valid for inhomogeneities with arbitrary scales of spatial variation and global modes of small amplitude, corresponding to a scenario near criticality.  For other scenarios, a numerical iterative nonlinear eigenvalue technique is preferred.  Two global modes of different amplitudes are revealed in the numerical approach.   For both the two-level system and the nonlinear CGL equation, the analytical calculations are supplemented with direct numerical simulation, thus showing the fate of unstable global modes.  For the two-level model this results in unbounded growth of the full nonlinear equations.  For the spatially-extended CGL model in the subcritical regime, the global mode of larger amplitude exhibits a  `one-sided' instability leading to a chaotic dynamics, while the global mode of smaller amplitude is always unstable (theory confirms this).  However,  advection can stabilize the mode of larger amplitude.
\end{abstract}


\maketitle

\section{Introduction}
\label{sec:intro}

Non-normal nonlinear evolutionary equations are of great interest not only in their own right but also in application areas, in particular, in Fluid Dynamics~\cite{eckhardt2007,Grossmann2000}, where the combination of non-normality and non-linearity give rise respectively to (i) transient growth and (ii) coherent states and global modes, and are thereby evoked as a subcritical route to turbulence.  The focus of the present work is on the study of such equations in their own right, although 
potential applications in Optics and Fluid Dynamics are highlighted.  
The present work is first of all placed in the context of the existing literature concerning transient growth, coherent states, and global modes.

\paragraph*{Transient growth in linear systems:} Historically, the linear stability of a system of evolutionary equations was tackled via eigenvalue analysis.  Consider a generic evolutionary system $\partial_t \vecu=\bm{F}(\vecu;\mu)$, where $\vecu(t)$ is a trajectory in $\mathbb{C}^n$ and $\mu$ is a real parameter.  A \textit{base state} $\vecu_0$ of the system is a solution of the equation $\bm{F}(\vecu_0;\mu)=0$.  The \textit{linearization} around the base state refers to the following equation: $\partial_t \vecu=\bm{J}\vecu$, where $\bm{J}$ is the Jacobian matrix, with $J_{ij}=(\partial F_i/\partial u_j)_{(\bm{u}_0,\mu)}$.  Let $\{\lambda_1,\cdots,\lambda_n\}=\mathrm{spec}(\bm{J})$, and let $\lambda_0$ denote the eigenvalue with largest real part.  The full system is said to be linearly stable if $\Re(\lambda_0)<0$.  Suppose for a critical parameter value $\mu=\mu_\mathrm{c}$ we have $\Re(\lambda_0)<0$ for $\mu<\mu_\mathrm{c}$ and $\Re(\lambda_0)>0$ for $\mu>\mu_\mathrm{c}$.  Then the system undergoes a bifurcation from linearly stable to linearly unstable at the critical value $\mu=\mu_\mathrm{c}$.  The same analysis can usually be applied  to spatially extended systems by projecting the evolutionary equation on to a set of basis functions of dimension $n$ and taking $n\rightarrow\infty$, assuming that the weak solutions so constructed tend to a strong solution.

Such systems can exhibit transient growth in a subcritical parameter regime whereby for suitable initial conditions involving a superposition of eigenmodes, the $L^2$ norm (or some other more appropriate measure of the energy) of the solution grows initially before eventually decaying at a rate dictated by $\lambda_0$.  This can happen in two ways.  If the eigenvectors do not form a complete set spanning $\mathbb{C}^n$, then the solution to the initial-value problem will involve contributions such as $t^p\mathe^{\lambda_j t}$, where $p\geq 1$ and $j\in\{1,\cdots,n\}$, such that transient algebraic growth is possible  before the onset of the asymptotic decay of the solution.  We do not consider this case further here (but see Reference~\cite{Schmid2001}).  Instead, we consider situations where the eigenvectors of $\bm{J}$ do indeed span $\mathbb{C}^n$ but are non-orthogonal, corresponding to a situation where $\bm{J}$ is not a normal operator.  In such a scenario, 
the $L^2$ norm of solutions involving certain combinations of eigenvectors can grow transiently before eventually decaying, with transient growth factors $\|\vecu(t)\|_2/\|\vecu(t=0)\|_2$ that can be as high as $O(10^5)$ (e.g. Reference~\cite{Chomaz1997}).  This is the phenomenon of transient growth.  

Starting in the 1990s, transient growth has been identified in fluid dynamics where the base state is a parallel flow and the linearized equations describe the dynamics of the perturbed velocity and pressure fields~\cite{Trefethan1993,Schmid2001}.  It is mooted as a mechanism for the subcritical transition to turbulence~\cite{Grossmann2000} and as a route to three-dimensional waves in two-phase parallel flow~\cite{YeckoZaleski2002}.  However, it is not always observed in direct numerical simulation based on the full nonlinear equations (even when the linearized equations of motion predicts that it should play a role)~\cite{Onaraigh2014}, and its importance as a potential mechanism for the subcritical transition to instability (hence turbulence) is debated: in the work by Waleffe~\cite{Waleffe1995,Waleffe1998} the importance of coherent states (i.e. a self-sustained spatial structure; see below) is emphasized in the subcritical transition to instability, while the putative role of transient growth is minimized.  A synthesis arising from this debate is presented in References~\cite{Schmid2001,Grossmann2000,eckhardt2007}, in which a scenario is presented whereby the transient growth provides the initial energy needed for a small-amplitude disturbance to grow subcritically and become self-sustained via the excitation of a coherent state.
Indeed, transient growth has been experimentally observed in at least one context (the optical one), in the propagation of light in unidirectional wave guides~\cite{Kostenbauder1997,SiegmanBook}.  Since the generic properties of the evolutionary equations are similar in both fluids and optics, a suggestion here is to use transient growth on the side of experimental optics (where extremely precise experiments with materials with almost arbitrary optical properties) to illuminate transient growth on the fluids side.
  

\textit{Applications of the linearized Complex Ginzburg--Landau equation:} We review briefly the applications of the complex Ginzburg--Landau equation, focusing in particular on the linearized version, although the extension to nonlinear versions is immediate.  The application of the CGL equation to optical waveguides is immediate (Reference~\cite{SiegmanBook} and Appendix~\ref{app:paraxial}) and needs no further discussion.  We therefore also describe briefly the role of the CGL equation in Fluid Dynamics.   As in the above, for linear stability analysis of parallel flows, one solves the linearized equations of motion transformed into the frequency domain and thereby reduced to eigenvalue form (e.g. Orr--Sommerfeld or Orr--Sommerfeld--Squire equations).  
 In a similar manner, one may solve for the response of the base state to an impulsive disturbance, localized both in space and time.  If the disturbance grows in time but is convected downstream by the base-state parallel flow, the system is said to be convectively unstable.  On the other hand, if the disturbance grows in time at the location of the initial impulsive disturbance, then the flow is said to be absolutely unstable~\cite{Schmid2001,huerre90a}. 

The growth rates obtained from eigenvalue analysis of the linearized Navier--Stokes equations are merely asymptotic (i.e. valid in the limit as $t\rightarrow\infty$).  This is true for both parallel flows and global stability analyses.  In the parallel flow case, for both the Orr--Sommerfeld and Orr--Sommerfeld--Squire equations, effects not captured by a simple eigenvalue analysis can be present at early times and lead to substantial transient growth because these equations are non-normal.  In addition, there are very few closed-form solutions to the Orr--Sommerfeld and Orr--Sommerfeld--Squire equations.  Thus, the complex Ginzburg--Landau (CGL) equation is introduced instead  as a minimal model of absolute instability, transient growth, and global instability in parallel flows.  Although different from the physical linearized fluid equations (e.g. Orr--Sommerfeld or Orr--Sommerfeld--Squire) equations, this model is sufficiently complex that it still retains the same important qualitative behaviours as the former.
  The model equation for the homogeneous case is
\begin{equation}
\frac{\partial u}{\partial t}+U\frac{\partial u}{\partial x}=\mu u+\gamma\frac{\partial^2 u}{\partial x^2};
\label{eq:cgl1}
\end{equation}
here $u(x,t)$ represents some instability in the fluid that is convected downstream by a constant velocity $U$.  The parameters $U$ and $\mu$ are real, while  $\gamma=\gmr+\imag\gmi$ is the complex-valued diffusion coefficient, with $\gmr>0$.  

Equation~\eqref{eq:cgl1}  is amenable to a modal analysis, wherein a normal-mode solution $u(x,t)=\mathe^{-\imag\omega t+\alpha x}$ is possible, provided the following dispersion relation is satisfied:
\begin{equation}
\omega=\imag\left(\mu-U\alpha-\gamma\alpha^2\right)
\label{eq:normalmode}
\end{equation}
In this context, a complex diffusion coefficient with positive imaginary part is necessary to give a growth rate $\omi=\mu-\gmr\alpha^2$ that is negative for short wavenumbers, leading to a well-posed problem.  This is a minimal model of linear instability, and it includes longwave growth and shortwave stabilization.  Clearly, the flow is unstable if $\mu>0$, where instability exists for a band of wavenumbers with $\alpha^2<\gmr/\mu$.  More importantly, we have a minimal model of absolute instability~\citep{huerre90a}: the model is said to be absolutely unstable if (i) it is unstable, with $\mu>0$; (ii) the Green's function $G(x,x_0,t)$ blows up at $x=x_0$ and as $t\rightarrow\infty$.   If condition (i) is satisfied but not (ii), the flow is said to be convectively unstable.   In more detail, the Green's function solves
\[
\frac{\partial G}{\partial t}+U\frac{\partial G}{\partial x}=\mu G+\gamma\frac{\partial^2G}{\partial x^2},\qquad 
G(x,x_0,t=0)=\delta(x-x_0);
\]
by direct computation we have
\begin{subequations}
\begin{equation}
G(x,x_0,t)=\left(2\pi\gamma t\right)^{-1/2}\exp\Bigg\{t\left[\mu-\frac{1}{4\gamma}\left(\frac{x-x_0}{t}-U\right)^2\right]\Bigg\},\qquad \gmr>0.
\label{eq:greens_cgl}
\end{equation}
The flow is absolutely unstable if $\lim_{t\rightarrow\infty}G(x,x_0,t)=\infty$, i.e. if $\mu-\tfrac{1}{4}U^2\Re(1/\gamma)>0$.  However, this condition can be recast as follows:
\begin{equation}
\omi(\alpha_0)>0,\qquad \frac{d\omega}{d\alpha}\bigg|_{\alpha_0}=0,
\label{eq:greens_cgl_cond}
\end{equation}%
\label{eq:greens_cgl_all}%
\end{subequations}%
i.e. $\alpha_0$ is the (complex) saddle point of the dispersion relation~\eqref{eq:normalmode}.
The derivation of Equations~\eqref{eq:greens_cgl_all} involves the computation of a Gaussian integral over all wavenumbers.  Here lies the wide applicability of CGL theory: a wide class of dispersion relations in fluid mechanics is locally quadratic in wavenumber, enabling a saddle-point approximation of the generic Green's function, such that all Green's functions in this class resemble Equation~\eqref{eq:greens_cgl} in the  limit as $t\rightarrow\infty$.  Thus, for this wide class of fluid-mechanical dispersion relations, the condition for absolute instability is simply that the imaginary part of the frequency at the saddle point be positive~\cite{huerre90a}.  Of course, there is a variety of fluid-mechanical problems where naive application of the saddle-point criterion~\eqref{eq:greens_cgl_cond} fails~\cite{Onaraigh2013}.  However, the saddle-point criterion~\eqref{eq:greens_cgl_cond} supplemented with the requirement that the spatial branches emanating from the saddle point should extend into opposite half-planes in the complex $\alpha$-plane, is sufficient to guarantee absolute instability~\cite{Briggs1964,Schmid2001}.

\paragraph*{Coherent states:}  A coherent state refers to a self-sustained nonlinear disturbance.  In work on parallel flow instability, the coherent state typically takes the form of a stationary state or a travelling wave (e.g. References~\cite{Waleffe1998,eckhardt2007}).  However, so-called nonlinear global modes in open shear flows
can be classified in the same way, since such global modes  comprise a coherent disturbance localized in space (but often involving propagating fronts) and oscillating in time, sustained by nonlinear absolute instability (e.g. Reference~\cite{huerre2000}).  A main aim in this and related works is to understand the nonlinear response of parallel-flow systems to impulsive localized forcing, moving beyond the linear theory of absolute and convective instability.  The nonlinear (homogeneous) complex Ginzburg--Landau equation (CGL) is introduced in these papers as a `toy model' whose solution gives insight into the far more complicated (but analogous) fluid-mechanical systems~\cite{Pier2002}. 

\paragraph*{The inhomogeneous complex Ginzburg--Landau equation (linear and nonlinear):} The focus of the present work is in the inhomogeneous Complex Ginzburg--Landau (CGL) equation with quadratic inhomogeneity.   The linearized equation is solved exactly in Reference~\cite{Hunt1991} (but see also Reference~\cite{le1996linear}), and the resulting solutions have been used to demonstrate the phenomenon of transient growth~\cite{Chomaz1997}.  There, a connection is made with the phenomenon of linear absolute instability in the homogeneous case and linear (`global') instability in the inhomogeneous case: linear absolute instability is necessary but not sufficient for the onset of linear global instability.
The nonlinear inhomogeneous  case has been solved asymptotically by Pier and co-workers (References~\cite{pier1996fully,pier1998}; for a review, see Reference~\cite{huerre2000}) in the limit where the spatial scale of the inhomogeneity varies slowly, such that WKB theory applies.  Two kinds of self-sustained oscillation appear, namely \textit{soft global modes} and \textit{steep global modes}.  The soft modes obtain their frequency from wholly nonlinear considerations, and are so called because the nonlinear upstream spatial wavenumber of the global mode switches smoothly to the nonlinear downstream spatial wavenumber across a region of local absolute instability.  In contrast, the frequency selected by the steep modes is obtained from linear theory, and the upstream and downstream spatial wavenumbers of the global mode connect in a non-smooth fashion.  Again, the two kinds of nonlinear global mode retain a ghostly residuum of the linear theory: a necessary condition for the existence of a soft or steep global mode is the presence of a region of local linear absolute instability.  Also, steep global modes appear as soon as there exists a point of local linear absolute instability, even though the  system is still linearly globally stable.

In the present work, the focus is moved to the study of the inhomogeneous CGL equation with inhomogeneities that vary on arbitrary spatial scales.  A small-amplitude perturbation theory is presented corresponding to the situation close to criticality.  The perturbation theory is based on the regular perturbation theory of Quantum Mechanics~\cite{cohen1977quantum}, but the results are similar to a previous work based on Stuart--Landau theory~\cite{Chomaz05}.  The main innovation therefore is to use the perturbative results as a stepping stone to a full analysis using a numerical approach that is valid for arbitrary amplitudes and arbitrary departures from criticality, supplemented by direct numerical simulation to determine whether the global modes are realised in the transient evolution of the model system.

This work is organized as follows.  In Section~\ref{sec:two_level} the model two-level nonlinear system is introduced and a closed-form solution for the self-sustained oscillation (`global mode') is obtained.  The instability of the solution is demonstrated, while in Section~\ref{sec:two_level_dns} the evolution of states close to the global mode is checked via numerical simulation.  In Section~\ref{sec:global_modes} the results concerning the inhomogeneous CGL equation are presented.  To augmet these results, numerical methods and resulting calculations are presented in Section~\ref{sec:cgl_numerical}.  Concluding remarks are made in Section~\ref{sec:conc}.

\section{A reduced model involving a two-level system with nonnormality and nonlinearity}
\label{sec:two_level}

In this section a two-level system is introduced that exhibits the twin effects of non-normality and nonlinearity.
The aim is to start with a set of linear equations and to quantify the influence of the non-normality on the transient growth at subcritical parameter regimes.  Next, a cubic nonlinearity is introduced and an exact solution for the global mode (and its stability) is obtained.  The section concludes with a crucial discussion about the physical relevance of the system in optics. 

\subsection{Non-normal linear model}

We consider the following two-level system
\begin{equation}
\imag\frac{\partial u}{\partial t}=\mathcal{H}u+\imag\left(\mu_0\mathbb{I}+\mathcal{G}\right)u,\qquad u\in\mathbb{C}^2,
\label{eq:qm_model}
\end{equation}
where
\[
\mathcal{H}=\left(\begin{array}{cc}E_0&A\\A&E_0\end{array}\right),\qquad
\mathcal{G}=\left(\begin{array}{cc}-g_1&0\\0&-g_2\end{array}\right),
\]
where $E_0,A,\mu_0,g_1,g_2$ are positive real numbers, and where $\mathcal{H}$ and $\mu_0\mathbb{I}+\mathcal{G}$ are Hermitian matrices, with $[\mathcal{H},\mathcal{G}]\neq 0$, where $[\cdot,\cdot]$ denotes the matrix commutator.  Note that $[\mathcal{H},\mathcal{G}]\neq 0$ implies that the operator on the right-hand side of Equation~\eqref{eq:qm_model} is non-normal.    One writes down the trial solution $u(t)=u_0\mathe^{-\imag\omega t}$ to obtain the following eigenvalues:
\begin{equation}
\omr=E_0\pm\sqrt{4A^2-(g_1-g_2)^2},\qquad \omi=\mu_0-\tfrac{1}{2}(g_1+g_2),\qquad 4A^2>(g_1-g_2)^2,\qquad \text{Case 1},
\label{eq:case1}
\end{equation}
and
\begin{equation}
\omr=E_0,\qquad \omi=\mu_0-\tfrac{1}{2}(g_1+g_2)\pm\sqrt{(g_1-g_2)^2-4A^2},\qquad 4A^2<(g_1-g_2)^2,\qquad \text{Case 2},
\label{eq:case2}
\end{equation}
The corresponding eigenvectors for both Case 1 and Case 2 are
\[
u_{0+}=\left(\begin{array}{c}\mathe^{-\imag\varphi/2}\\\mathe^{\imag\varphi/2}\end{array}\right),\qquad
u_{0-}=\left(\begin{array}{c}-\mathe^{\imag\varphi/2}\\\mathe^{-\imag\varphi/2}\end{array}\right),\qquad
\varphi=\sin^{-1}\left(\frac{g_1-g_2}{2A}\right).
\]
For Case 1, we have $|(g_1-g_2)/2A|<1$, hence $\varphi\in\mathbb{R}$, while for Case 2, $\varphi$ is complex.  For both cases, the eigenvectors are non-orthogonal; the degree of non-orthogonality is expressed by the relation $\langle u_{0+},u_{0-}\rangle=2\imag(g_2-g_1)/2A$.  Here, $\langle\cdot,\cdot\rangle$ denotes the usual scalar product on $\mathbb{C}^2$, with $\langle \bm{u},\bm{v}\rangle=\overline{\bm{u}^{T}}\bm{v}$, for all column vectors $\bm{u},\bm{v}$ in $\mathbb{C}^2$. Orthogonality of eigenvectors is regained for $g_1=g_2$; this corresponds precisely to $\left[\mathcal{H},\mathcal{G}\right]=0$ and the normality of the (symmetric) operator $\mathcal{H}+\imag\mathcal{G}$.
The crossover between Cases 1 and 2, where $(g_1-g_2)^2=4A^2$ is referred to as a \textit{diabolic point} in the literature concerning non-Hermitian quantum mechanics~\cite{graefe2009,rotter2009}; this point truly is fiendish, since the eigenvectors degenerate and fail to span $\mathbb{C}^2$ in this instance.  However, in this work, consideration is given strictly to Case 2, for reasons given in what follows.

For all possible parameter values, the intrinsic non-Hermiticity of the operator in Equation~\eqref{eq:qm_model} implies the following nontrivial evolution equation for the $L^2$ norm of the general solution $u(t)$:
\begin{equation}
\|u\|_2^2:=\langle u,u\rangle,\qquad
\tfrac{1}{2}\frac{d}{dt}\|u\|_2^2=\langle u,\left(\mu_0+\mathcal{G}\right)u\rangle.
\label{eq:evol_l2}
\end{equation}
The quadratic form on the right-hand side can be evaluated in the eigenbasis of the operator $\mathcal{G}$ and a bound on the growth of the $L^2$ norm is obtained:
\begin{equation}
\tfrac{1}{2}\frac{d}{dt}\|u\|_2^2\leq \left[\mu_0-\min(g_1,g_2)\right]\|u\|_2^2.
\label{eq:evol_l2_bound}
\end{equation}
In this work, we focus on models that exhibit transient growth.  Since this relies on the `mixing' of eigenstates with eigenfrequencies whose imaginary parts are distinct, we are forced by this constraint to work in Case 2 (this rules out of consideration the diabolic point).  With reference to this case, and to Equation~\eqref{eq:case2}, transient non-asymptotic growth is possible in the paramter range
\begin{equation}
\min(g_1,g_2)<\mu_0<\tfrac{1}{2}(g_1+g_2)-\sqrt{(g_1-g_2)^2-4A^2}.
\label{eq:transient_growth_sweet}
\end{equation}
This is the `sweet' operational range where the system is subcritical, in the sense that $\omi<0$ for both eigenvalues, but where the balance between forcing and dissipation is ambiguous, so that the sign of the upper bound of the growth rate in Equation~\eqref{eq:evol_l2_bound} is not definite.

\subsection{Introduction of nonlinear terms}
\label{sec:two_level_nonlin}

We add non-linearity to the problem by modifying Equation~\eqref{eq:qm_model} as follows:
\begin{equation}
\imag \frac{\partial u}{\partial t}=\mathcal{H}u+\imag\left(\mu_0\mathbb{I}+\mathcal{G}\right)u+
a\left(\begin{array}{cc}|u_1|^2&0\\0&|u_2|^2\end{array}\right)u,
\label{eq:qm_model_nonlin_local}
\end{equation}
where $a$ is real. 
The evolution of the $L^2$ norm $\|u\|_2^2$ is again unchanged under the addition of the nonlinear term: the quantity $\|u\|_2^2$ still evolves according to the norm-evolution first written down for the linear problem (i.e. Equation~\eqref{eq:evol_l2}).
However, in contrast to the linear model, Equation~\eqref{eq:qm_model_nonlin_local} has a nonlinear periodic solution, which we find by making the trial solution
\begin{equation}
u=R\mathe^{-\imag\Omega t}u_0,\qquad \|u_0\|_2^2=1,\qquad \Omega\in\mathbb{R}.
\label{eq:u_nonlin_ansatz}
\end{equation}
Substitution of Equation~\eqref{eq:u_nonlin_ansatz} into Equation~\eqref{eq:qm_model_nonlin_local} yields the following eigenvalue problem for $u_0$:
\begin{equation}
\Omega u_0=\mathcal{L}u_0+aR^2\left(\begin{array}{cc}|u_{01}|^2&0\\0&|u_{02}|^2\end{array}\right)u_0,\qquad
\mathcal{L}=\mathcal{H}+\imag\left(\mu_0\mathbb{I}+\mathcal{G}\right).
\label{eq:nonlin_eigs}
\end{equation}
In general, Equation~\eqref{eq:nonlin_eigs} has a family of solutions labelled by the continuous parameter $R$, with corresponding eigenvalues $\Omega(R)$.  However, there is  a finite number of $R$-values consistent with the requirement $\Im[\Omega(R)]=0$; all other $R$-values correspond to eigensolutions of Equation~\eqref{eq:nonlin_eigs} that are nonetheless inconsistent with Equation~\eqref{eq:u_nonlin_ansatz}.  It is as if we are solving a nonlinear double eigenvalue problem in the parameters $\Omega$ and $R$ (the analogous double eigenvalue problem of linear algebra is addressed elsewhere~\cite{Blum1978}).  Solutions of this double eigenvalue problem are called \textit{self-consistent} in the remainder of this work, and can be found analytically for the simple two-level system considered in this section.  Indeed, we have the following theorem:
\begin{theorem}
\label{thm:exact1}
Let $g_2<\mu_0<g_1$, and let
\begin{subequations}
\begin{equation}
X^2=-\frac{(\mu_0-g_1)(\mu_0-g_2)}{A^2}.
\end{equation}
Assume that $X^2<1$.   Then Equation~\eqref{eq:nonlin_eigs} has the following self-consistent solution:
\begin{equation}
\Omega=E_0+aR^2,\qquad R^2=\frac{g_1-g_2}{a}\sqrt{\frac{1}{X^2}-1},\qquad u_0=\left(\begin{array}{c}r\mathe^{\imag\varphi}\\x\end{array}\right),
\end{equation}
where
\begin{equation}
r=\sqrt{\frac{\mu_0-g_2}{g_1-g_2}},\qquad x=\sqrt{\frac{g_1-\mu_0}{g_1-g_2}},\qquad \varphi=-\sin^{-1}(X).
\end{equation}
\label{eq:lovely_theorem}
\end{subequations}
\end{theorem}
\begin{proof}
The so-called self-consistent nonlinear eigenvalue problem refers to the solution $u_0$ of Equation~\eqref{eq:nonlin_eigs} with the constraint that 
\begin{equation}
\langle u_0,(\mu_0+\mathcal{G})u_0\rangle=0.
\label{eq:nonlin_eigs_marginal}
\end{equation}
To obtain such a solution, we  take
\begin{equation}
u_0=\left(\begin{array}{c}r\mathe^{\imag\varphi(R)}\\a\end{array}\right),\qquad r,\varphi,x\in\mathbb{R},\qquad r^2+x^2=1;
\label{eq:u0_trial}
\end{equation}
this amounts to a fixed choice for the (arbitrary) global phase of $u_0$.  Substitution into Equation~\eqref{eq:nonlin_eigs_marginal} yields
\begin{equation}
r^2\left(\mu_0-g_1\right)+x^2\left(\mu_0-g_2\right)=0,\qquad r^2+x^2=1,
\label{eq:rsq1}
\end{equation}
hence
\begin{equation}
r=\sqrt{\frac{\mu_0-g_2}{g_1-g_2}},\qquad x=\sqrt{\frac{g_1-\mu_0}{g_1-g_2}},
\label{eq:rsq2}
\end{equation}
where $r$ and $x$ are both real because $g_2<\mu_0<g_1$.
The phase $\varphi$ is not determined by this analysis; this is obtained  by consideration of the full nonlinear eigenvalue problem, which is written out in full as follows:
\begin{subequations}
\begin{eqnarray}
\Omega r\mathe^{\imag\varphi}&=&\mathcal{L}_{11}r\mathe^{\imag\varphi}+\mathcal{L}_{12}x+aR^2r^3\mathe^{\imag\varphi},\label{eq:om1}\\
\Omega x&=&\mathcal{L}_{21}r\mathe^{\imag\varphi}+\mathcal{L}_{22}x+aR^2x^3\label{eq:om2}.
\end{eqnarray}%
\label{eq:om_all}%
\end{subequations}%
The imaginary part of the second equation is set to zero to yield a root-finding condition for $\varphi$:
\begin{equation}
\mathcal{L}_{21}^{(\mathrm{r})}\sin\varphi+\mathcal{L}_{21}^{(\mathrm{i})}\cos\varphi+\mathcal{L}_{22}^{(\mathrm{i})}\left(x/r\right)=0.
\label{eq:root_phi}
\end{equation}
Because the operator $\mathcal{L}$ is symmetric, setting the imaginary part of the first equation to zero yields precisely the same condition for $\varphi$. 
Equation~\eqref{eq:root_phi} simplifies drastically, once the coefficients $\mathcal{L}_{21}^{(\mathrm{r})}$ etc. are filled in:
\[
A\sin\varphi+(\mu_0-g_2)(x/r)=0.
\]
Using Equations~\eqref{eq:rsq1}--\eqref{eq:rsq2}, this expression can be written as an explicit function of the problem parameters alone:
\[
A\sin\varphi\pm (\mu_0-g_2)\sqrt{-\frac{\mu_0-g_1}{\mu_0-g_2}}=0,
\]
where the two branches come from taking $x=\pm \sqrt{1-r^2}$ in Equation~\eqref{eq:rsq2}, and where the radicand is positive because of the constraints on the parameters in Theorem~\ref{thm:exact1}.  Hence,
\begin{equation}
\varphi=\mp \sin^{-1}\left(\sqrt{-\frac{(\mu_0-g_1)(\mu_0-g_2)}{A^2}}\right)=\mp \sin^{-1}(X),
\label{eq:root_phi_expl}
\end{equation}
with $\varphi\in [-\pi/2,\pi/2]$.
Upon satisfaction of the condition~\eqref{eq:root_phi_expl}, Equations~\eqref{eq:om1} and~\eqref{eq:om2} are consistent, but only in the sense that both equations now imply that $\Omega$ is real.  The equations~\eqref{eq:om1} and~\eqref{eq:om2} are made totally consistent by choosing a value of $R$ such that the value of $\Omega$ in both these equations is the same.  In terms of the explicit values of the $\mathcal{L}$-matrix, we have
\begin{eqnarray*}
\Omega&=&E_0+A\cos\varphi(x/r)+aR^2r^2,\\
&=&E_0+A\cos\varphi(r/x)+aR^2x^2.
\end{eqnarray*}
This is a simple quadratic equation in $R$, for which $R$ has a single real positive root:
\[
R^2=\frac{g_1-g_2}{a}\sqrt{\frac{1}{X^2}-1}.
\]
This procedure therefore picks out two isolated pairs $(\mp\varphi,R)$ that leads to a self-consistent solution of the nonlinear double eigenvalue problem.\qed
\end{proof}
Further analytical progress is possible with respect to Equation~\eqref{eq:nonlin_eigs}: the  family of complex-valued eigenvalues $\Omega(R)$ can be controlled in the following manner:
\begin{theorem}
Take as given the non-linear eigenvalue problem~\eqref{eq:nonlin_eigs}, with eigenvalues $\Omega(R)\in\mathbb{C}$ for $R\geq 0$.  Then
\[
(\mu_0-g_2)-(g_1+g_2)\leq \Im(\Omega) \leq \mu_0-g_2.
\]
\end{theorem}
\begin{proof}
Take $r\mathe^{-\imag\varphi} \left[\text{Eq.~\eqref{eq:om1}}\right]+x\left[\text{Eq.~\eqref{eq:om2}}\right]$ to obtain
\[
\Omega=E_0+2Arx\cos\varphi+aR^2(r^4+x^4)+\imag(\mu_0-g_1r^2-g_2x^2).
\]
Since $x^2+r^2=1$, we have
\[
\Im(\Omega)=(\mu_0-g_2)-(g_1+g_2)r^2,
\]
with $0\leq r^2\leq 1$, hence $(\mu_0-g_2)-(g_1+g_2)\leq \Im(\Omega) \leq \mu_0-g_2.$\qed
\end{proof}
Thus, the range of allowed values of $\Im(\Omega)$ is bounded above and below.  In order for $\Im(\Omega)=0$ to be achievable, the upper bound must be non-negative, so that a necessary condition for the existence of a global mode is $\mu_0>g_2$.  By Equation~\eqref{eq:evol_l2_bound}, this is the same necessary condition as the one required for  linear transient growth to occur.

\subsection{Stability of nonlinear oscillatory state}
The stability of the self-consistent oscillatory state is investigated by consideration of a trial solution
\[
u=R\mathe^{-\imag\Omega t}u_0+\delta u,
\]
where $R=R_\mathrm{c}$, and where $u_0$ is the solution of the nonlinear eigenvalue problem.  Linearization around the periodic state yields the following ODE:
\begin{eqnarray}
\imag \frac{d}{dt}\delta u&=&\mathcal{L}\delta u+2R^2a \left(\begin{array}{cc} |u_{01}|^2 & 0\\0& |u_{02}|^2\end{array}\right)\delta u+R^2\mathe^{-2\imag\Omega t}\left(\begin{array}{cc} u_{01}^2 &0\\0&u_{02}^2\end{array}\right)\overline{\delta u},\nonumber\\
&:=&\mathcal{A}\,\delta u+\mathe^{-2\imag\Omega t}\mathcal{B}\,\overline{\delta u}.
\label{eq:periodic_lsa0}
\end{eqnarray}
The manifestly time-dependent term is removed by making the trial solution $\delta u=\delta v \mathe^{-\imag\Omega t}$, such that Equation~\eqref{eq:periodic_lsa0} becomes
\begin{equation}
\imag \frac{d}{dt}\delta v=\left(\mathcal{A}- \Omega\right)\delta v+\mathcal{B}\overline{\delta v}.
\label{eq:periodic_lsa1}
\end{equation}
Equation~\eqref{eq:periodic_lsa1} can be further simplified by breaking it up into real and imaginary parts.  We take $\delta v=(a+\imag b,u+\imag v)^T$, where $a$, $b$, $u$, and $v$ are real-valued functions of $t$.  We also take $\mathcal{B}=\widehat{\beta}+\imag\widehat{\gamma}$, where $\widehat{\beta}$ and $\widehat{\gamma}$ are real matrices.  Similarly, we write $\Re(\mathcal{A}-\Omega):=\widehat{\mathcal{H}}$.  Finally, we note that $\Im(\mathcal{A}-\Omega)=\mu_0\mathbb{I}+\mathcal{G}$.  Using this notation, we obtain the following real system of equations:
\begin{equation}
\frac{d}{dt}\left(\begin{array}{c}a\\b\\u\\v\end{array}\right)=
\left(\begin{array}{cccc}
\mu_0-g_1+\widehat{\gamma}_{11}&\widehat{\mathcal{H}}_{11}-\widehat{\beta}_{11}&0&A\\
-\widehat{\mathcal{H}}_{11}-\widehat{\beta}_{11}&\mu_0-g_1-\widehat{\gamma}_{11}&-A&0\\
0&A&\mu_0-g_2+\widehat{\gamma}_{22}&\widehat{\mathcal{H}}_{22}-\widehat{\beta}_{22}\\
-A&0&-\widehat{\mathcal{H}}_{22}-\widehat{\beta}_{22}&\mu_0-g_2-\widehat{\gamma}_{22}\end{array}\right)
\left(\begin{array}{c}a\\b\\u\\v\end{array}\right).
\label{eq:periodic_lsa2}
\end{equation}
The results of Theorem~\ref{thm:exact1} are now used to fill out the entries in the matrix in Equation~\eqref{eq:periodic_lsa2} with closed-form expressions involving the model parameters:
\begin{equation}
\frac{d}{dt}\left(\begin{array}{c}a\\ \\b\\ \\u\\ \\v\end{array}\right)=
\left(\begin{array}{c|c|c|c}
\mu_0-g_1\phantom{aaaa}&aR^2(2r^2-1)&0&A\\
+aR^2r^2\sin 2\varphi & -aR^2r^2\cos 2\varphi & & \\\hline
-aR^2(2r^2-1)&\mu_0-g_1&-A&0\\
-aR^2r^2\cos 2\varphi & -aR^2r^2\sin 2\varphi & & \\\hline
0&A&\mu_0-g_2&-aR^2(2r^2-1) \\
 & &         & -aR^2(1-r^2) \\\hline
-A&0&aR^2(2r^2-1)&\mu_0-g_2\\
  &  & -aR^2(1-r^2) &      \end{array}\right)
\left(\begin{array}{c}a\\\\b\\\\u\\\\v\end{array}\right)
\label{eq:periodic_lsa3}
\end{equation}
Based on this expression of the stability equations, we prove the following theorem:
\begin{theorem}
The oscillatory state in Equation~\eqref{eq:lovely_theorem} is asymptotically unstable.
\end{theorem}
\begin{proof}We write $(a,b,u,v)^T=\mathe^{\sigma t}(a_0,b_0,u_0,v_0)$, and search for eigensolutions where $a_0,b_0,u_0$ and $v_0$ are real.  Calling the matrix in Equation~\eqref{eq:periodic_lsa3} $\mathcal{M}$, it can be shown by direct computation that $\mathrm{det}(\mathcal{M})=0$ identically (e.g. using a symbolic algebra package), such that the characteristic equation of the matrix $\mathcal{M}$ is
\begin{equation}
\sigma\left[\sigma^3-\mathrm{tr}(\mathcal{M})\sigma^2+m_1\sigma+m_2\right]=0,
\label{eq:charpoly4}
\end{equation}
where $m_1$ and $m_2$ are real constants that depend on the problem parameters.  A further direct computation yields
\begin{equation}
m_2=8A^2(1-X^2)\left[\mu_0-\tfrac{1}{2}(g_1+g_2)\right].
\end{equation}
Under the twin assumptions that we are in the operating range given by Equation~\eqref{eq:transient_growth_sweet} and that $0<X^2<1$, we have that $m_2<0$, such that the cubic polynomial in the square brackets $\left[\cdots\right]$ in Equation~\eqref{eq:charpoly4} has at least one real \textit{positive} root, with a corresponding real eigenvector $(a_0,b_0,u_0,v_0)^T$.  The corresponding real eigenvalue therefore has $\sigma>0$, hence the system~\eqref{eq:periodic_lsa3} is asymptotically unstable.\qed
\end{proof}

\subsection{Physical relevance of the two-level system}

Remarkably, the linear two-level system~\eqref{eq:qm_model} is directly applicable in optics, and describes the interaction between two coupled optical channels.  The starting-point is the complex Ginzburg--Landau equation for propagation along the optical axis ($z$, Equation~\eqref{eq:cgl_optical}).  An optical potential is prescribed with variations in the orthogonal direction ($x$) such that the beam is strongly focused at two $x$-locations (the channels).  Much like the `linear combination of atomic orbitals' approach in Quantum Mechanics~\cite{cohen1977quantum}, the electric field can then written as a superposition of two states, each corresponding to localization at one of the preferred $x$-positions.  The two amplitudes in the superposition vary in time and describe `hopping' as the intensity maximum of the beam switches between the channels.  The evolution equation for the amplitudes is formally identical to Equation~\eqref{eq:qm_model}.  The non-normality of Equation~\eqref{eq:qm_model} is also a feature in optical systems, in particular in systems where spontaneous PT-symmetry breaking occurs~\cite{ruter2010observation}.
Also, concerning the addition of the nonlinear term in Equation~\eqref{eq:qm_model_nonlin_local}, a very similar set of equations holds in a one-dimensional optical array with a Kerr nonlinearity~\cite{lederer2008discrete}.  The proposed model here represents an extension of such coupled optical systems to non-Hermitian (i.e. PT-symmetry-breaking) optical media.  A nice feature therefore of this model is that it opens up the possibility that some fluid-like phenomena (transient growth and subcritical transition) can be demonstrated in an optical context.



\section{Two-level system -- Illustrative numerical studies}
\label{sec:two_level_dns}

Although we have completely characterized the nonlinear oscillatory state and its stability by analytical means, it is useful nonetheless to consider numerical simulations.  First, the numerical method developed in the present section for the computation of the nonlinear eigenstate can be extended to more complicated systems where exact closed-form solutions are not available.  Also, assuming the system is prepared in the nonlinear (unstable) oscillatory base state, it is of interest to know what is the ultimate fate of the system if it experiences small perturbations: whether the system reverts to the linearized dynamics, whether another nonlinear coherent structure appears, or whether nonlinear unbounded growth occurs.  This question is answered via a numerical simulation in the present section.

We first of all solve the eigenvalue problem for Equation~\eqref{eq:nonlin_eigs} numerically.  We use a modified Rayleigh-Quotient iterative (RQI) method, modified to take account both of the non-normality of the operator $\mathcal{L}=\mathcal{H}+\imag\mathcal{G}$, and the nonlinearity in Equation~\eqref{eq:nonlin_eigs}.  Modified (`two-sided') RQI methods exist for linear non-normal problems: the standard Ostrawski method exhibits local cubic convergence; other variants are globally convergent~\cite{parlett1974}.  We extend the standard Ostrawski method here to the nonlinear problem~\eqref{eq:qm_model_nonlin_local}: the initial guess $(u^{(0)},u^{\dagger(0)},\Omega^{(0)})$ for the iterative method is the solution of Equation~\eqref{eq:nonlin_eigs} with $R=0$ (i.e. the (known) solution of the linearized problem).  Here the dagger refers to solutions of the adjoint problem.    Subsequent (and improved) guesses $(\psi^{(n+1)},\psi^{\dagger(n+1)},\Omega^{(n+1)})$ with $n=0,1,2,\cdots$ are found as follows:
\begin{subequations}
\begin{eqnarray}
\mathcal{N}^{(n)}&=&aR^2\left(\begin{array}{cc}|u_1^{(n)}|^2&0\\0&|u_2^{(n)}|^2\end{array}\right),\\
u^{(n+1)}&=&
\frac{\left(\mathcal{L}+\mathcal{N}^{(n)}-\Omega^{(n)}\mathbb{I}\right)^{-1}u^{(n)}}{\|\left(\mathcal{L}+\mathcal{N}^{(n)}-\Omega^{(n)}\mathbb{I}\right)^{-1}u^{(n)}\|_2},
\label{eq:rqi_direct}\\
u^{\dagger(n+1)}&=&\frac{\left[\left(\mathcal{L}+\mathcal{N}^{(n)}-\Omega^{(n)}\mathbb{I}\right)^{\dagger}\right]^{-1}u^{\dagger (n)}}{\|\left[\left(\mathcal{L}+\mathcal{N}^{(n)}-\Omega^{(n)}\mathbb{I}\right)^{\dagger}\right]^{-1}u^{\dagger (n)}\|_2},
\label{eq:rqi_adjoint}\\
\mathcal{N}^{(n+1)}&=&aR^2\left(\begin{array}{cc}|u_1^{(n+1)}|^2&0\\0&|u_2^{(n+1)}|^2\end{array}\right),\\
\Omega^{(n+1)}&=&\langle u^{\dagger (n+1)},\left(\mathcal{L}+\mathcal{N}^{(n+1)}\right)u^{(n+1)}\rangle.
\label{eq:rqi_twolevel}
\end{eqnarray}
\end{subequations}
%
The numerical method is deemed to have converged at iteration level $n_\mathrm{c}$ when
\[
\|\Omega^{(n_\mathrm{c})}\psi^{(n_\mathrm{c})}-\left(\mathcal{L}+\mathcal{N}^{(n_\mathrm{c})}\right)\psi^{(n_\mathrm{c})}\|_2\leq \mathrm{tol},
\]
where $\mathrm{tol}=10^{-16}$ is chosen in the present work.
The rigorous analysis of the convergence of the method~\eqref{eq:rqi_twolevel} is beyond the scope of this work: in practice, the method~\eqref{eq:rqi_twolevel} converges for all of the examples considered in this Report; also, the numerical method agrees exactly (i.e. to within machine precision, using IEEE double precision) with the analytical expressions derived in Theorem~\ref{thm:exact1}.

Results of the nonlinear eigenvalue analysis are presented now.  Results of the numerical and analytical approaches are identical and are referred to interchangeably.  We use the parameter values 
\begin{multline}
E_0=1,\qquad A=0.1,\qquad g_2=0.01,\qquad g_1=g_2+(4A^2+0.01^2)^{1/2},\qquad a=1,\\
\mu_0=0.95\left[\tfrac{1}{2}\left(g_1+g_2\right)-\sqrt{\left(g_1-g_2\right)^2-4A^2}\right].
\label{eq:params_all}
\end{multline}
Also, the initial condition was taken to be $(A_0/\sqrt{2})(\imag,1)^T$, where $A_0$ is taken to be small in an appropriate sense, so that the early-stage dynamics correspond to linear theory.
The parameters~\eqref{eq:params_all} correspond to a subcritical case: both eigenfrequencies of the linearized dynamics possess negative imaginary parts.  However, the same parameters are appropriate for (linear) transient growth, as outlined in Equation~\eqref{eq:transient_growth_sweet}.    The nonlinear eigenvalue analysis is presented in the first instance in Figure~\ref{fig:spec_toy_model}.
\begin{figure}[htb]
\centering
\subfigure[]{\includegraphics[width=0.47\textwidth]{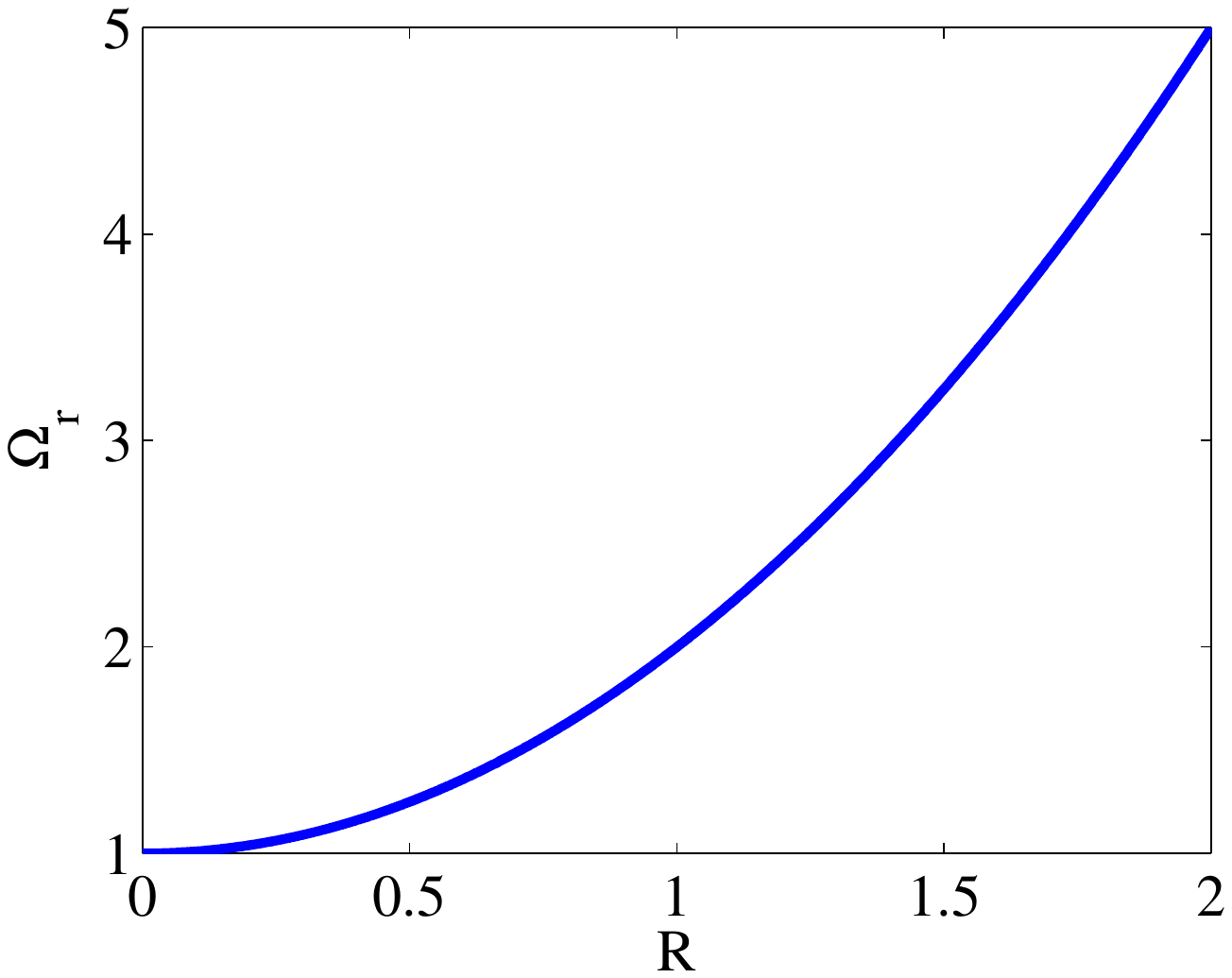}}
\subfigure[]{\includegraphics[width=0.47\textwidth]{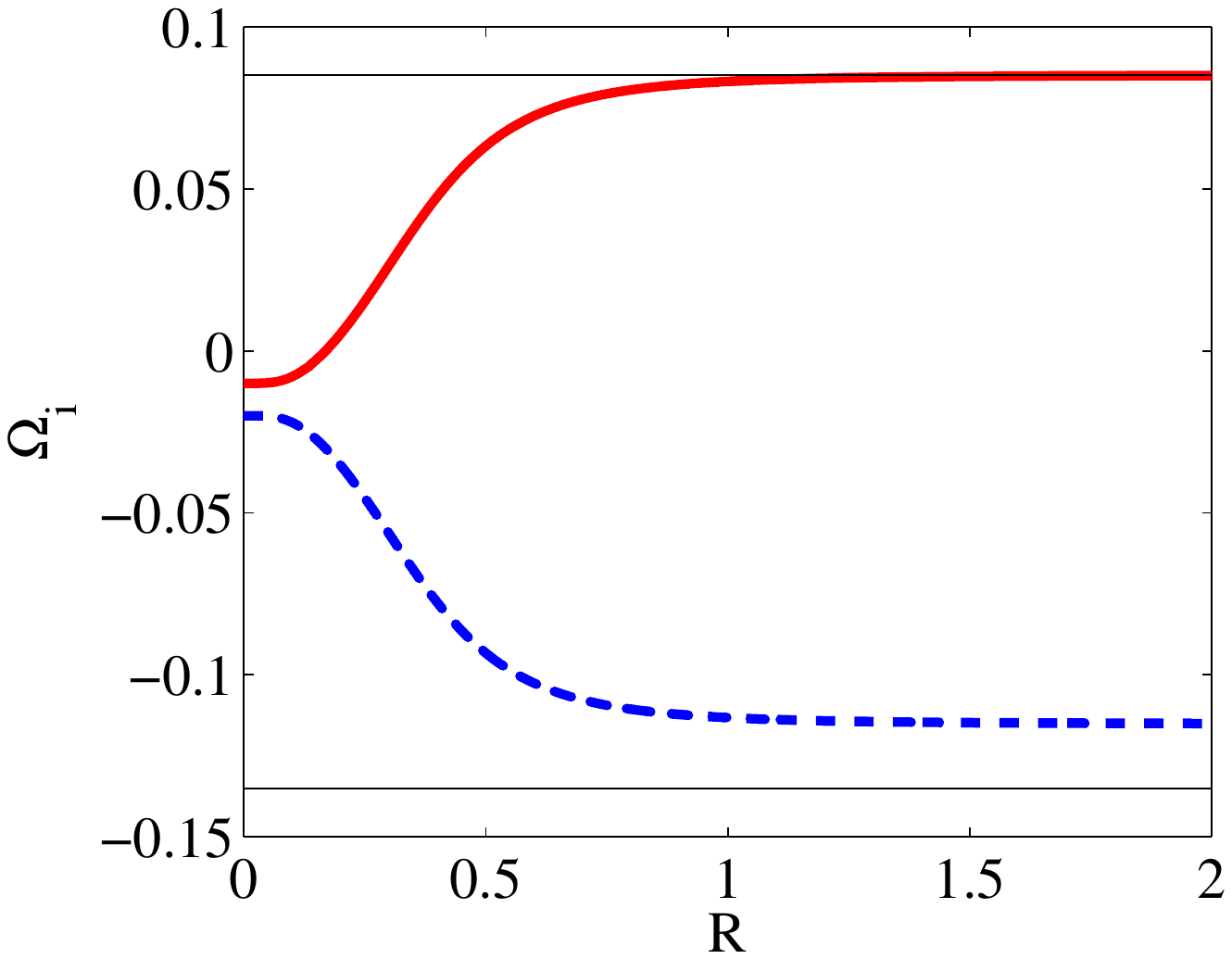}}
\caption{Eigenvalue analysis of Equation~\eqref{eq:nonlin_eigs}, for the parameter values given in the text.  At $R\approx 0.169176816128673$ (analytical calculation, to machine precision), a self-consistent solution is obtained along one of the eigenvalue branches.  }
\label{fig:spec_toy_model}
\end{figure}
The parameter $R$ is increased, starting at $R=0$.  At $R=0$, the linear problem is recovered, with two linearly independent (but nonorthogonal) eigenvalues with identical real parts but distincitive imaginary parts.  This feature is retained at finite $R$: the imaginary parts of the two distinct eigenvalues branch off and become increasingly large and increasingly small.  On the increasing branch, there is a point $R_\mathrm{c}$ at which $\Im[\Omega(R)]=0$; this is the self-consistent nonlinear eigensolution of Equation~\eqref{eq:qm_model_nonlin_local}.  Also, for each finite $R$-value, the cross product of the two eigenvalues is computed: this never vanishes.  Thus, the two eigenvalues remain linearly independent: there is no nonlinear diabolic point.  Note finally that the upper bound $\Im(\Omega)\leq \mu_0-g_2$ is sharp as $R\rightarrow\infty$, corresponding to $r=0$ and $x=1$, hence $\Re(\Omega)\sim E_0+aR^2$ as $R\rightarrow\infty$.

It is of interest to examine whether the periodic unstable state has any bearing on the dynamics of Equation~\eqref{eq:qm_model_nonlin_local}.  Thus, temporal numerical simulations of the same equation are performed, and the results are shown in Figure~\ref{fig:qm_nonlin_modelB}.   The numerical solutions are obtained using an eight-order accurate Runge--Kutta scheme~\cite{Govorukhin2003}.
\begin{figure}[htb]
\centering
\subfigure[]{\includegraphics[width=0.48\textwidth]{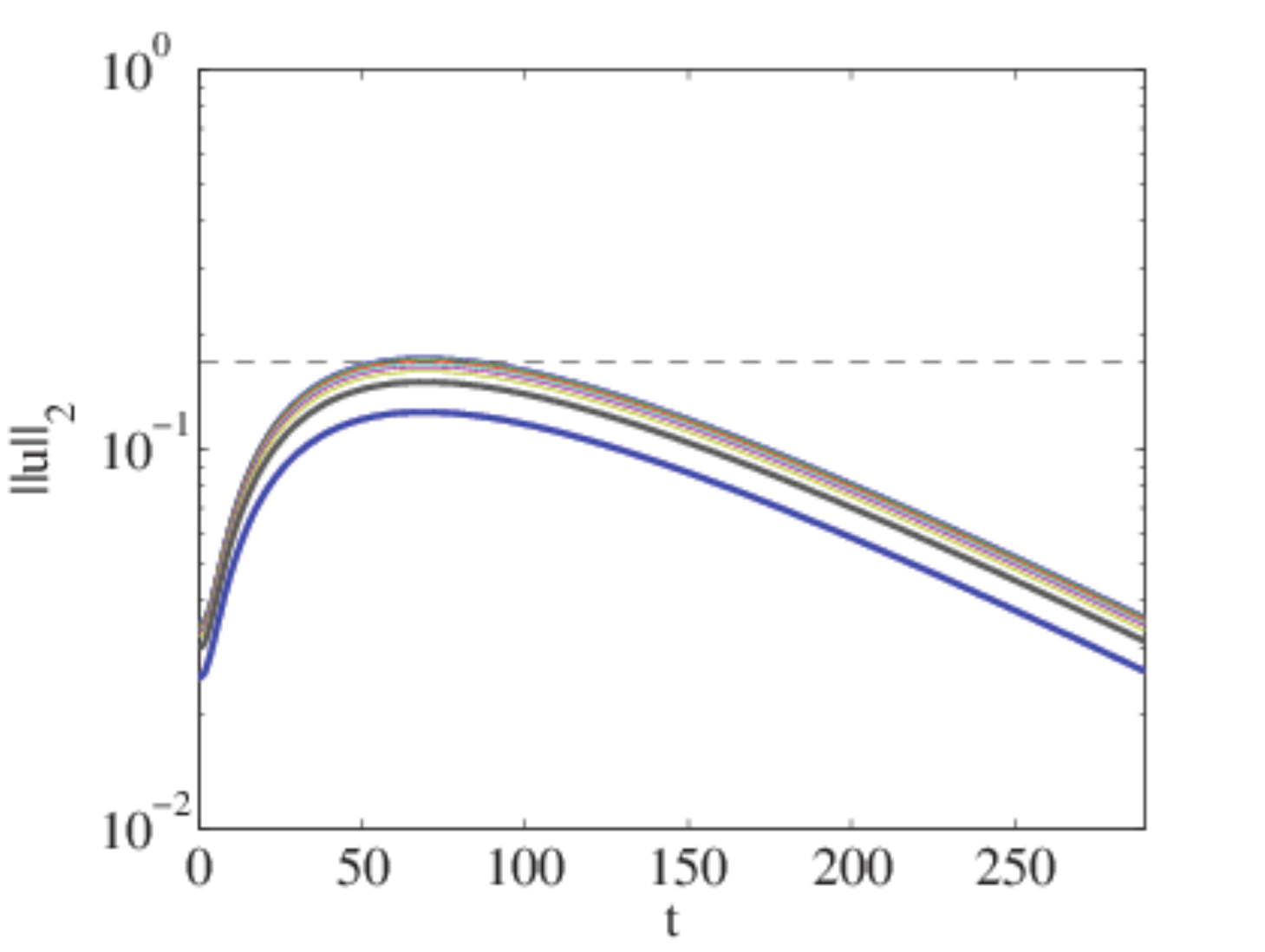}}
\subfigure[]{\includegraphics[width=0.48\textwidth]{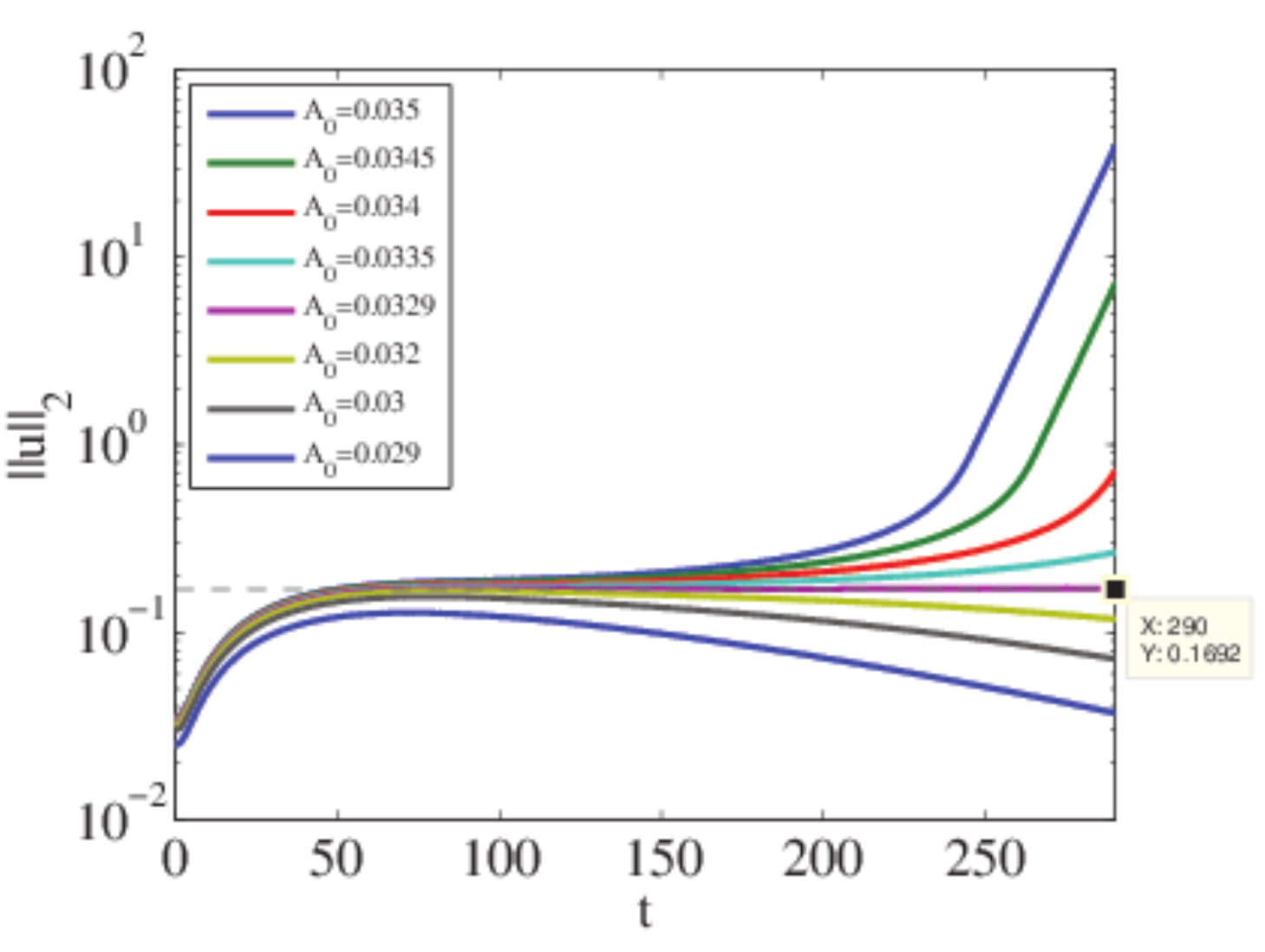}}
\caption{Solutions of (a) the  non-Hermitian linear Schrodinger equation; (b) the non-Hermitian {\textit{nonlinear}} Schr\"odinger equation.  The initial data are the parameters are the same in (a) and (b).}
\label{fig:qm_nonlin_modelB}
\end{figure}
Solutions under the linearized dynamics are shown in Figure~\ref{fig:qm_nonlin_modelB}(a).  Significant transient growth occurs for the initial data and parameter values in Equation~\eqref{eq:params_all}, in spite of the negatativity of the linear growth rates, i.e. $\Im\left[\mathrm{spec}\left(\mathcal{L}\right)\right]\leq 0$: the maximum amplification is over $500\%$ in the figure.  In contrast, solutions under the nonlinear dynamcis are shown in  Figure~\ref{fig:qm_nonlin_modelB}(b).  
%
%
The linear transient growth operates in the same manner as in Figure~\ref{fig:qm_nonlin_modelB}, up to $t\approx 80$.  Thereafter, a range of possibilities exists, depending on the initial amplitude of the disturbance.  First, one notes that there is a critical initial amplitude $A_{0\mathrm{c}}$ such that  $\lim_{t\rightarrow\infty}\|u\|_2=R_\mathrm{c}$;
%
%
this is precisely the critical radius  for which the nonlinear oscillatory state exists, as in the analysis in Equation~\eqref{eq:nonlin_eigs}.  Also, for initial amplitudes above this threshold, the same nonlinear oscillatory state is excited, but is subsequently destabilized, and indefinite exponential growth takes place.  Finally, for initial amplitudes below the same threshold, the nonlinear oscillatory state cannot be maintained, and the linear dissipation eventually causes the disturbance to decay to zero.  For the  case $A_0>A_{0\mathrm{c}}$, one obtains the eery result that the disturbance grows exponentially fast, in spite of the fact that the eigenfrequencies of the linearized problem possess negative imaginary parts.  This is a subcritical transition to (nonlinear) instability.

These results point to the following conclusion: linear transient growth by itself is not sufficient to induce a transition from a regime of small-amplitude disturbances to one of nonlinear instability in the model equation~\eqref{eq:qm_model_nonlin_local}. Rather, the transient growth must have available an excitable nonlinear eigenmode, and moreover,  the operating parameter regime must be such that the nonlinear eigenmode is linearly (asymptotically) unstable or neutral.  These findings are consistent with the literature on a vastly more complicated system, namely the subcritical transition to turbulence in parallel shear flows~\cite{Schmid2001,Grossmann2000}, where a combination of transient growth and the secondary instability of coherent states~\cite{Waleffe1995,Waleffe1998} is required for subcritical transition to turbulence. 

\section{Global mode analysis of the inhomogeneous complex Ginzburg--Landau equation}
\label{sec:global_modes}

The main aim of this section is to construct the small-amplitude global mode (or modes) of the complex Ginzburg--Landau (CGL) equation with a cubic linearity corresponding to a situation near criticality, using perturbation theory.   The same theory can also be used to determine the stability of the global modes.  This will serve as a basis for numerically determining the finite-amplitude global modes in later sections.

\subsection{The linear model with inhomogeneity}
The linearized CGL equation with parallel flow and quadratic inhomogeneity is presented here as follows:
\begin{subequations}
\begin{equation}
\imag\frac{\partial u}{\partial t}=\mathcal{H}+\imag\left(\mu_0+\mathcal{G}\right)u,\qquad (x,t)\in\mathbb{R}\times(0,\infty),\qquad \mu_0\in\mathbb{R},
\label{eq:cgl_linear}
\end{equation}
with initial and boundary conditions
\begin{equation}
u(x,t=0)=u_0(x),\qquad \lim_{|x|\rightarrow\infty}|u(x,t)|=0,\text{ }t\text{ finite},
\end{equation}
where
\begin{equation}
\mathcal{H}=-\imag U\frac{\partial}{\partial x}-\gmi\frac{\partial^2}{\partial x^2}+\tfrac{1}{2}\mu_{2\mathrm{i}} x^2,\qquad
\mathcal{G}=\gmr\frac{\partial^2}{\partial x^2}-\tfrac{1}{2}\mu_{2\mathrm{r}} x^2.
\end{equation}%
\label{eq:cgl_linear_all}%
\end{subequations}%
and $\mu_{2\mathrm{r}}>0$ corresponding to $\mathcal{G}$ giving a dissipative contribution to the dynamics.  Also, we take $\mu_0\in\mathbb{R}$.  A detailed analysis of Equations~\eqref{eq:cgl_linear_all}  reveals the eigenvalues to be 
\begin{equation}
\omega_n=\imag\left[\mu_0-\frac{U^2}{4\gamma}-\gamma\chi^2(2n+1)\right],\qquad n=0,1,\cdots,
\label{eq:cgl_eigs_formula_linear}
\end{equation}
where $\chi^2=\sqrt{\mu_2/2\gamma}$ and $\Re(\chi^2)>0$ 
(see References~\citep{Chomaz1987,Chomaz1997}).  Hence, the growth rate of the global modes is given by
\[
\Im(\omega_n)=\mu_0-\frac{U^2\gmr}{4|\gamma|^2}-(2n+1)|\gamma||\chi^2|\cos\left(\arg(\gamma)+\arg(\chi^2)\right).
\]
We specialize to a situation where $n=0$ (the `ground state') corresponds to the most unstable state.  In other words,
\begin{equation}
\cos\left(\arg(\gamma)+\arg(\chi^2)\right)>0.
\label{eq:cgl_eigs_formula_linear_cond}
\end{equation}
Also, the  eigenfunctions are
\begin{equation}
\psi_n(x)=\zeta_n\mathe^{(U/2\gamma)x-\chi^2x^2/2}H_n(\chi x),
\label{eq:cgl_eigfns_formula_linear}
\end{equation}
and the adjoint eigenfunctions are
\begin{equation}
\psi_n^\dagger(x)=\xi_n\overline{\psi_n(x)}\mathe^{-Ux/\overline{\gamma}},
\label{eq:cgl_eigfns_adj_formula_linear}
\end{equation}
where $H_n$ denotes the $n^{\mathrm{th}}$-order Hermite polynomial, with $n=0,1,\cdots$.  The eigenfunctions and their adjoints are orthogonal in the sense that $\langle \psi_n^\dagger,\psi_m\rangle=\delta_{nm}\propto \delta_{nm}$.  However, the normalizations $(\zeta_n,\xi_n)$ are chosen such that $\|\psi_n\|_2=\|\psi_n^\dagger\|_2=1$.
Here
\[
\langle f,g\rangle:=\int_{-\infty}^\infty \mathd x \overline{f}g
\]
for all complex-valued square-integrable functions $f,g$ of a single real variable $x$, with $x\in\mathbb{R}$,
and
\[
\|f\|_2=\left[\langle f,f\rangle\right]^{1/2}
\]
denotes the usual $L^2$ norm for such functions.

Based on Equations~\eqref{eq:cgl_linear_all}, the following evolution equation for the $L^2$ norm of the general solution $u(t)$ is obtained:
\begin{equation}
\tfrac{1}{2}\frac{d}{dt}\|u\|_2^2=\langle u,\left(\mu_0+\mathcal{G}\right)u\rangle.
\label{eq:cgl_evol_l2}
\end{equation}
The quadratic form on the right-hand side of Equation~\eqref{eq:cgl_evol_l2} can be evaluated in the eigenbasis of the operator $\mathcal{G}$ and a Poincar\'e-type bound on the growth of the $L^2$ norm is obtained:
\begin{equation}
\tfrac{1}{2}\frac{d}{dt}\|u\|_2^2\leq \left(\mu_{0}-\sqrt{\frac{\mu_{2\mathrm{r }}\gmr}{2}} \right)\|u\|_2^2.
\label{eq:evol_l2_bound}
\end{equation}
Thus, in order to have any growth in the $L^2$ norm (whether transient or asymptotic growth), one must have $\mu_{0}>\sqrt{\mu_{2\mathrm{r }}\gmr/2}$. 
Thus, in order to achieve transient but non-asymptotic growth, one must operate in the following parameter regime for the forcing parameter $\mu_0$
\begin{equation}
\sqrt{\frac{\mu_{2\mathrm{r }}\gmr}{2}}<\mu_{0}< 
\mu_{0\mathrm{c}},\qquad
\mu_{0\mathrm{c}}=\frac{U^2\gmr}{4|\gamma|^2}+\sqrt{\frac{|\mu_2||\gamma|}{2}}\cos(\arg(\gamma)+\arg(\chi^2)).
\label{eq:cgl_eigs_stability_range}
\end{equation}
It is this regime that concerns us throughout the rest of the present section.

\subsection{Introduction of nonlinear terms}
\label{sec:global_modes:nonlin}

We  modify Equation~\eqref{eq:cgl_linear} by introducing a nonlinear term as follows:
\begin{equation}
\imag\frac{\partial u}{\partial t}=\mathcal{H}u+\imag\left(\mu_0+\mathcal{G}\right)u+|u|^2 s(x)u,
\label{eq:cgl_nonlinear}
\end{equation}
where the functional form of $s(x)$ will be discussed later on in this section.  As in Section~\ref{sec:two_level}, we seek a solution of Equation~\eqref{eq:cgl_nonlinear} as follows:
\begin{equation}
u(x,t)=Ru_0(x)\mathe^{-\imag\Omega t},\qquad R\in\mathbb{R},\qquad \|u_0\|_2=1,\qquad \Omega\in\mathbb{R},
\label{eq:cgl_nonlinear_ansatz}
\end{equation}
where $R> 0$ is a real parameter.  
Substitution of Equation~\eqref{eq:cgl_nonlinear_ansatz} into Equation~\eqref{eq:cgl_nonlinear_ansatz} yields the following (global, nonlinear) eigenvalue equation:
\begin{equation}
\Omega u_0(x)=\left[\mathcal{H}+\imag\left(\mu_0+\mathcal{G}\right)+R^2|u_0(x)|^2s(x)\right]u_0(x).
\label{eq:cgl_nonlinear_evp}
\end{equation}
As in Section~\ref{sec:two_level}, we seek eigensolutions of Equation~\eqref{eq:cgl_nonlinear_evp} with $\Im(\Omega)=0$.  If such a state exists, then for $u(x,t)=Ru_0(x)\mathe^{-\imag\Omega t}$,
\[
\tfrac{1}{2}\frac{d}{dt}\|u\|_2^2=R^2\langle u_0,\left(\mu_0+\mathcal{G}\right)u_0\rangle=R^2\Im(\Omega)=0.
\]
The aim of the present subsection is to demonstrate the existence of such oscillatory states for $s(x)=1$, and to demonstrate further the inherent instability of such states.  The method to be used in the demonstration is non-singular time-independent perturbation theory.  To make progress, we introduce a small parameter by specifying $\mu_0$ as follows:
\begin{equation}
\mu_{0}=\mu_{0\mathrm{c}}+\epsilon^2\Delta\mu,\qquad 0<\epsilon \ll 1,
\label{eq:cgl_eigs_stability_range1}
\end{equation}
where $\Delta\mu=\pm 1$.  For $\Delta\mu=-1$, the system is just  below criticality wherein the system is asymptotically stable but wherein transient growth can still occur, while for $\Delta\mu=1$ the system is just above criticality and is linearly unstable.  We have the following theorem:
\begin{theorem}
\label{thm:oscillatory_cgl}
Consider the nonlinear eigenvalue problem~\eqref{eq:cgl_nonlinear_evp} with $s(x)=1$ and $a>0$.  Define $\theta=\arg(\chi^2)$ and let
\begin{equation}
f(\theta,k^2)=
\sqrt{\frac{\mathe^{\imag\theta}\cos\theta}{\cos\theta+\mathe^{\imag\theta}}}\exp\left(-\frac{k^2\mathe^{\imag\theta}}{\cos\theta\left(\cos\theta+\mathe^{\imag\theta}\right)}\right),\qquad
k^2=U^2[\Re(1/\gamma)]^2/4|\chi^2|.
\end{equation}
Assume that the following constraints are applied to the system parameters: 
\begin{enumerate}[label=(\roman*)]
\item $\Delta\mu=-1$ and $f_{\mathrm{i}}(\theta,k^2)>0$ (subcritical), or 
$\Delta\mu=1$ and $f_{\mathrm{i}}(\theta,k^2)<0$ (supercritical);
\item Equations~\eqref{eq:cgl_eigs_formula_linear_cond},~\eqref{eq:cgl_eigs_stability_range} and~\eqref{eq:cgl_eigs_stability_range1} hold;
\end{enumerate}
Then, to first order in a perturbation theory in the small parameter $\epsilon$, there exists at least one value of the parameter $R$ such that the nonlinear eigenvalue problem has an oscillatory solution.
\end{theorem}
\begin{proof} We use fairly standard perturbation methods from Quantum Mechanics.  The eigenvalue pair $(\Omega,u_0(x))$ is expanded in powers of $R^2$, with $\Omega=\omega_0+R^2\omega_1+R^4\omega_2+O(R^6)$.   We obtain
\begin{equation}
\omega_1=a\frac{\langle \psi_0^\dagger,|\psi_0|^2\psi_0\rangle}{\langle\psi_0^\dagger,\psi_0\rangle},
\label{eq:perturbation_theory0}
\end{equation}
such that the perturbed ground-state frequency reads
\begin{equation}
\Omega=\omega_0+aR^2\left(\frac{\langle \psi_0^\dagger,|\psi_0|^2\psi_0\rangle}{\langle\psi_0^\dagger,\psi_0\rangle}\right)+O(R^4).
\label{eq:perturbation_theory1}
\end{equation}
Here, the `ground-state' refers to the nonlinear eigenstate that connects continuously to the ground state of the corresponding linear eigenvalue problem (i.e. $R=0$, and the eigenvalue $n=0$ in Equation~\eqref{eq:cgl_eigs_formula_linear}). 
By direct computation, we obtain
\begin{equation}
\Omega=\omega_0
+\tfrac{1}{\sqrt{\pi}}\sqrt{\frac{\chi^2\Re(\chi^2)}{\chi^2+\Re(\chi^2)}}\exp\bigg\{\tfrac{1}{4}U^2\Re(1/\gamma)^2\left[\frac{1}{\chi^2+\Re(\chi^2)}-\frac{1}{\Re(\chi^2)}\right]\bigg\}.
\label{eq:perturbation_theory2}
\end{equation}
We take $\theta=\arg(\chi^2)$, such that Equation~\eqref{eq:perturbation_theory2} simplifies as
\begin{equation}
\Omega=\omega_0+\tfrac{1}{\sqrt{\pi}}a|\chi| R^2\sqrt{\frac{\mathe^{\imag\theta}\cos\theta}{\cos\theta+\mathe^{\imag\theta}}}\exp\left(-\frac{k^2\mathe^{\imag\theta}}{\cos\theta(\cos\theta+\mathe^{\imag\theta})}\right)+O(R^4),\qquad k^2=\frac{U^2[\Re(1/\gamma)]^2}{4|\chi|^2},
\label{eq:perturbation_theory3}
\end{equation}
or 
\begin{equation}
\Omega=\omega_0+\tfrac{1}{\sqrt{\pi}}a|\chi| R^2 f(\theta,k^2)+O(R^4).
\label{eq:perturbation_theory4}
\end{equation}
Taking $\mu_0=\mu_{0\mathrm{c}}$ as in the theorem statement, we have
\[
\Im(\Omega)=\epsilon^2\Delta\mu+\tfrac{1}{\sqrt{\pi}}a|\chi| R^2 f_\mathrm{i}(\theta,k^2)+O(R^4).
\]
Thus, to lowest order in the perturbation theory, an oscillatory solution with  $\Im(\Omega)=0$ is possible at a critical radius $R_\mathrm{c}$, where
\begin{equation}
R_\mathrm{c}^2=-\epsilon^2\frac{\Delta\mu}{\Im(\omega_1)}=-\frac{\epsilon^2\Delta\mu\sqrt{\pi}}{a|\chi|f_\mathrm{i}(\theta,k^2)},
\label{eq:Rc_perturbation}
\end{equation}
with $R_\mathrm{c}^2>0$ for the conditions outlined in Theorem~\ref{thm:oscillatory_cgl}(i).
The critical radius is $O(\epsilon)$, and we are therefore justified in neglecting the higher-order terms in the perturbation theory.  \qed
\end{proof}

The allowed regions of parameter space with $-\pi/2<\theta<\pi/2$ are shown in Figure~\ref{fig:theta_k}.  White region correspond to the subcritical case ($f_{\mathrm{i}}(\theta,k^2)>0$, $\Delta\mu=-1$), and black regions to the supercritical one.
\begin{figure}[htb]
\centering
\includegraphics[width=0.6\textwidth]{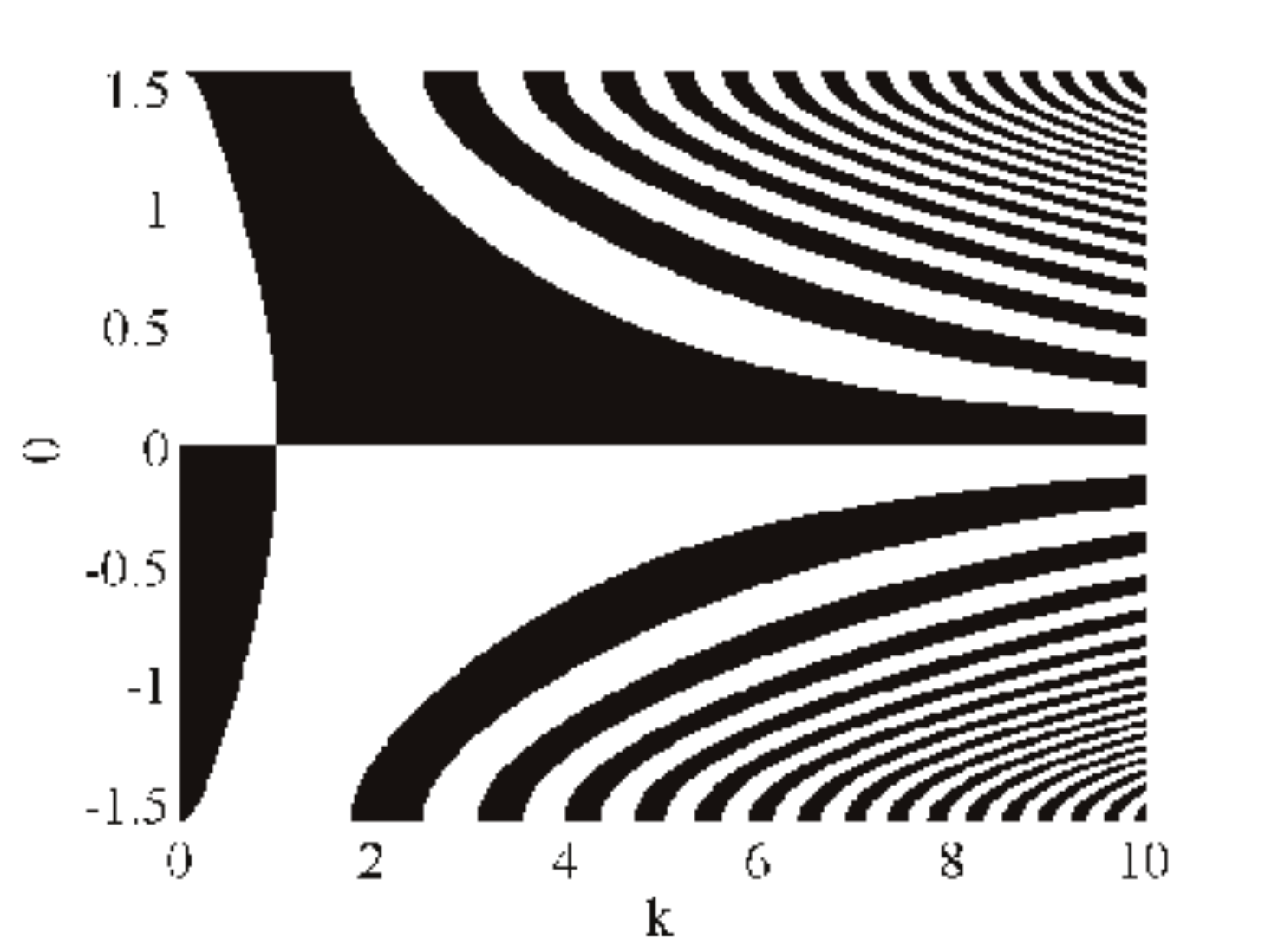}
\caption{Contour plot showing the sign of $f_{\mathrm{i}}(\theta,k^2)$.  Regions in white correspond to the positive sign (subcritical case), regions in black to the negative (supercritical case).}
\label{fig:theta_k}
\end{figure}
Perturbation theory also furnishes the eigenfunction of the nonlinear problem.  The expansion $u_0(x)=\psi_0(x)+R^2 u_1(x)+O(R^4)$ is posed, and $u_1(x)$ is computed in a standard fashion:
\begin{equation}
u_1(x)=-\Re(\langle\psi_0,\widehat{u_1}\rangle)\psi_0(x)+\widehat{u_1}(x),\qquad
\widehat{u_1}=
\sum_{p=1}^\infty \frac{\langle \psi_p^\dagger,|\psi_0|^2\psi_0\rangle}{\langle \psi_p^\dagger,\psi_p\rangle}\frac{\psi_p(x)}{\omega_0-\omega_p},
\label{eq:u1}
\end{equation}
such that
\[
\langle \psi_0+R^2u_1,\psi_0+R^2 u_1\rangle=O(R^4),
\]
i.e. the perturbed solution has unit norm, to order $R^2$ in the perturbation theory.
%
The oscillatory state defined by~\eqref{eq:Rc_perturbation} is intrinsically unstable, as we now demonstrate:
\begin{theorem}
The subcritical case obtained by the method described in Theorem~\ref{thm:oscillatory_cgl} is linearly unstable for $t\ll 1/\epsilon^2$; the supercritical case obtained by the same method is linearly stable, again for $t\ll 1/\epsilon^2$.
\label{thm:cgl_stability}
\end{theorem}
Note that this result holds for all parameter sets consistent with the constraints in Theorem~\ref{thm:oscillatory_cgl}.  
\begin{proof}
We take $u(x,t)=R_c^2 u_0\mathe^{-\imag \Omega t}+\delta v(x,t)\mathe^{-\imag\Omega t}$ in Equation~\eqref{eq:cgl_nonlinear} and linearize, assuming that $\delta v(x,t)$ is small in the appropriate sense.  Here $\Omega=\Omega(R_\mathrm{c})$ is the real frequency of the oscillatory state, with $\Omega=\omega_0+aR_\mathrm{c}^2\omega_a$. The result is the following equation:
\begin{equation}
\imag\frac{\partial}{\partial t}\delta v=\left(\mathcal{L}-\omega_0-a R_\mathrm{c}^2\omega_1\right)\delta v+2R_\mathrm{c}^2|u_0|^2\delta v+R_\mathrm{c}^2u_0^2\overline{\delta v},
\label{eq:stability_cgl}
\end{equation}
where $\mathcal{L}=\mathcal{H}+\imag\left(\mu_0+\mathcal{G}\right)$.   The perturbation $\delta v(x,t)$ is expanded in terms of the eigenfunctions~\eqref{eq:cgl_eigfns_formula_linear}:
\[
\delta v(x,t)=\sum_{n=0}^\infty a_n (t)\psi_n(x)
\]
The following exact (i.e. non-perturbative) equation for the amplitudes $a_n(t)$ is obtained:
\[
\imag \dot a_n=(\omega_n-\omega_0)a_n-a R_\mathrm{c}^2\omega_1 a_n
+2R_\mathrm{c}^2\sum_{j=0}^\infty a_j
\frac{\langle \psi_n^\dagger,|u_0|^2\psi_j\rangle}{\langle\psi_n^\dagger,\psi_n\rangle}
+R_\mathrm{c}^2\sum_{j=0}^\infty \overline{a_j}
\frac{\langle \psi_n^\dagger,u_0^2\overline{\psi_j}\rangle}{\langle \psi_n^\dagger,\overline{\psi_n}\rangle}.
\]
We have $\omega_n-\omega_0=-2\gamma\chi^2n$.  Hence,
\begin{eqnarray*}
\imag \dot a_0&=&O(R_\mathrm{c}^2),\\
\imag \dot a_n&=&-2\imag \gamma\chi^2 a_n+O(R_\mathrm{c}^2),\qquad n\geq 1.
\end{eqnarray*}
In a slow-manifold approximation, we therefore have $a_n\propto \mathe^{-2\gamma\chi^2 t}$, hence $a_n\rightarrow 0$ for $\Re(\gamma\chi^2)\ll t\ll 1/R_\mathrm{c}^2\propto 1/\epsilon^2$.  Consequently, it is valid to assume $a_n=0$ in the $n=0$ equation, at least for in the specified range of times, up to $t\ll 1/\epsilon^2$.  We obtain
\[
\imag \dot a_0=-a\omega_1 R_\mathrm{c}^2a_0+2aR_\mathrm{c}^2 a_0\langle \psi_0^\dagger,|u_0|^2\psi_0\rangle+R_\mathrm{c}^2\overline{a_0}\langle \psi_0^{\dagger},u_0^2\overline{\psi_0}\rangle.
\]
To lowest order in perturbation theory, we have
\[
a R_\mathrm{c}^2\langle \psi_0^{\dagger},|u_0|^2\psi_0\rangle=
a R_\mathrm{c}^2\langle \psi_0^{\dagger},|\psi_0|^2\psi_0\rangle=a R_\mathrm{c}^2\omega_1,
\]
hence
\[
\imag \dot a_0= a\omega_1 R_{\mathrm{c}}^2 a_0+\overline{a_0}aR_\mathrm{c}^2\langle \psi_0^\dagger,\psi_0^2 \overline{\psi_0}\rangle.
\]
For the subcritical case, $\Im(\omega_1)>0$, corresponding to instability.  For the supercritical case, $\Im(\omega_1)<0$, corresponding to stability.\qed
\end{proof}
We have a final result concerning the parameter values necessary to support a global mode:
\begin{theorem}
Suppose that a global mode $u(x,t)$ exists.  Then
\begin{equation}
\sqrt{\frac{\mu_{2\mathrm{r}}\gmr}{2}}<\mu_{0}.
\label{eq:onset_necessary}
\end{equation}
\end{theorem}
\begin{proof}
By the Poincar\'e-type inequality~\eqref{eq:evol_l2_bound}, we have
\[
\tfrac{1}{2}\frac{d}{dt}\|u\|_2^2\leq \left(\mu_{0}-\sqrt{\frac{\mu_{2\mathrm{r}}\gmr}{2}}\right)\|u\|_2^2
\]
for the global mode.  For a global mode, we also have $(d/dt)\|u\|_2^2=0$, hence
\[
0\leq \left(\mu_{0}-\sqrt{\frac{\mu_{2\mathrm{r}}\gmr}{2}}\right)\|u\|_2^2=
\left(\mu_{0}-\sqrt{\frac{\mu_{2\mathrm{r}}\gmr}{2}}\right)R^2
\]
where $R^2>0$ is the radius squared of the global mode. The result now follows immediately.
\qed
\end{proof}
Thus,  the necessary condition for the onset of a global mode coincides with the necessary condition for transient growth (Equation~\eqref{eq:cgl_eigs_stability_range}).
Note that the condition~\eqref{eq:onset_necessary} differs from the criteria implied in the work by Cossu and Chomaz~\cite{Chomaz1997}: the critieria implied there refer instead to the $L^\infty$ norm of the solution.

\subsection{Connection to Stuart--Landau theory}
\label{sec:global_modes_sl}

We discuss briefly the connection between the present work and the computation of the weakly nonlinear global mode via Stuart--Landau theory, performed by Chomaz \textit{et al.}~\cite{chomaz1991effect}.  The aim here is to show the consistency of the present theory to existing theories, and then to transcend all such weakly nonlinear theories by solving the nonlinear eigenvalue problem numerically (Section~\ref{sec:cgl_numerical}), thereby opening up the possibility of having multiple global modes that are triggered by ever-larger disturbance amplitudes.

In the Stuart--Landau theory, consideration is again given to states near criticality, such that $\mu_0=\mu_\mathrm{c}+ \epsilon^2\Delta\mu$ as before, with $\Delta\mu=\pm 1$, and $\epsilon \ll 1$.  A slow timescale $T=\epsilon^2 t$ is introduced, and the following perturbation expansion is prescribed:
\begin{subequations}
\begin{equation}
u(x,t)=\sum_{n=1}^\infty \epsilon^n A_n(x,t,T).
\label{eq:amplsl_expansion}
\end{equation}
At order $\epsilon$, the solution is $A_1=\mathcal{A}(T)\mathe^{-\imag\omega_{0\mathrm{c}}}t \psi_{0\mathrm{c}}(x)$.  A compatibility condition leads to the following evolution equation for $\mathcal{A}(T)$:
\begin{equation}
\frac{d\mathcal{A}}{dt}=\Delta\mu\mathcal{A}+\imag\omega_1 \mathcal{A}|\mathcal{A}|^2,
\label{eq:amplsl}
\end{equation}
where $\omega_1=a\langle \psi_0^\dagger,|\psi_0|^2\psi_0\rangle/\langle\psi_0^\dagger,\psi_0\rangle$ as before.  The equation for $|\mathcal{A}|^2$ reads
\begin{equation}
\frac{d}{dt}|\mathcal{A}|^2=2|\mathcal{A}|^2\left(\Delta\mu+\Im(\omega_1) |\mathcal{A}|^2\right).
\label{eq:amplsq}%
\end{equation}%
\label{eq:my_sl}%
\end{subequations}%
The Stuart--Landau weakly nonlinear theory in Equations~\eqref{eq:my_sl} is valid for times $t\ll 1/\epsilon^2$, which is same as the condition for the stability analysis of Section~\ref{sec:global_modes:nonlin} to be valid.
Based on Equation~\eqref{eq:amplsl}, there are four cases to consider, depending on the respective signs of $\Delta\mu$ and $\Im(\omega_1)$.  Only those cases for which Equation~\eqref{eq:amplsq} admits a positive fixed point can be addressed sensibly within the framework of Stuart--Landau theory, since the existence of such a fixed point is precisely the condition needed to support a global mode.  The parameter values that produce a fixed point are the following:
\begin{itemize}
\item {\textbf{Subcritical bifurcation}}, corresponding to $\Delta\mu=-1$ and $\Im(\omega_1)>0$.  For these parameters, a fixed point $|\mathcal{A}|=1/\sqrt{\Im(\omega_1)}$ exists but is unstable.
\item {\textbf{Supercritical bifurcation}}, corresponding to $\Delta\mu=1$ and $\Im(\omega_1)<0$.  For these parameters, a fixed point $|\mathcal{A}|=1/\sqrt{-\Im(\omega_1)}$ exists and is stable.
\end{itemize}
For both these cases, the lowest-order solution at the fixed point reads
\[
u(x,t)=\epsilon\left(\frac{\Delta\mu}{\Im(\omega_1)}\right)^{1/2}\psi_{0\mathrm{c}}(x)\mathe^{-\imag\omega_{0\mathrm{c}}t}.
\]
and the subcritical case is linearly unstable and the supercritical one is linearly stable.
In comparison, the global-mode solution obtained via regular perturbation theory (Section~\ref{sec:global_modes:nonlin}) reads
\[
u(x,t)=R_\mathrm{c}\psi_0(x)\mathe^{-\imag(\omega_0+a\epsilon^2\omega_1)t}=
\epsilon\left(\frac{\Delta\mu}{\Im(\omega_1)}\right)^{1/2}\left[\psi_{0\mathrm{c}}(x)+\epsilon^2\frac{\partial\psi_0}{\partial\epsilon^2}(x)+\cdots\right]\mathe^{-\imag\omega_0 t},
\]
where again, the subcritical case corresponds to linear instability and the supercritical case to linear stability.  Thus, the two perturbation theories agree at lowest order in their characterization of the global mode.

\subsection{Discussion}

It would seem that the new approach to the construction of the global mode near criticality (i.e. drawing inspriation from the regular perturbation theory of Quantum-Mechancs) yields no more informtion than the existing approach advocated already by Chomaz \textit{et al.}~\cite{chomaz1991effect}.  However, the regular perturbation theory based on Quantum Mechanics does enable a straightforward extension beyond the first order.  Starting with Equation~\eqref{eq:u1}, one readily obtains a further correction for the complex frequency $\Omega=\omega_0+\omega_1 R^2+\omega_2R^4+O(R^6)$, with
\[
\omega_2=\frac{\langle\psi_0^\dagger,\psi_0\left(\psi_0\overline{\widehat{u_1}}+\overline{\psi_0}\widehat{u_1}\right)\psi_0\rangle+\langle\psi_0^\dagger,|\psi_0|^2\widehat{u_1}\rangle}{\langle\psi_0^\dagger,\psi_0\rangle}+2\omega_1\Re\left[-\langle\psi_0,\widehat{u_1}\rangle\right].
\]
Now, high-order perturbation theories can be dubious, since the radius of convergence of the series expansions in powers of $R$ is not known \textit{a priori}.  However, let us momentarily assume that an expansion up to and including the power $R^4$ is valid, such that $\Omega=\omega_0+\omega_1 R^2+\omega_2R^4+O(R^6)$.  We neglect the terms of order $R^6$.  Thus, the condition for the existence of a global mode now reads
\[
\Im\left(\omega_0+\omega_1 R^2+\omega_2 R^4\right)=0.
\]
This is a quadratic equation in $R^2$: previously, the criterion for the existence of a global mode (involving a perturbation theory up to and including powers of $R^2$) produced a single global mode; now the higher-order perturbation theory includes not only that single global mode, but admits also the possibility that further additional global modes can be found.  Even if the validity of these expansions is curtailed by a small radius of convergence, this result is indicative of a mechanism whereby increasing the magnitude of the nonlinearity can lead to an increase in the number of global modes (whether stable or unstable) supported by the system.  Beyond the radius of convergence of the perturbation theory, the existence of multiple global modes can be established via numerical calculations.  This topic is addressed in Section~\ref{sec:cgl_numerical}.

In the present analytical framework, it is possible to demonstrate - albeit heuristically - that the nonlinear CGL equation with focussing nonlinearity (i.e. $a>0$, $s(x)=1$) admits only a finite number of global modes.  This is done by obtaining bounds on $\Im(\Omega)$ taken over the whole range of $R$-values.  In  contrast to the two-level system considered before, where rigorous upper and lower bounds were obtained (Section~\ref{sec:two_level_nonlin}), in the present context only a rigorous upper bound on $\Im(\Omega)$ is available, and for a lower bound we resort to heuristic arguments. Thus, we solve Equation~\eqref{eq:cgl_nonlinear_evp} in the limit $R\rightarrow\infty$, with $\Omega(R)\in\mathbb{C}$.  The dominant balance in Equation~\eqref{eq:cgl_nonlinear_evp} is then between the nonlinear term and the `kinetic' term $\imag\gamma\partial_x^2 u_0$.
  For $a>0$, bound states exist, since the nonlinearity then acts like a `self-focussing' term, thereby inducing a pulse-like solution of width $w$.  The width $w$ is obtained by the dominant balance in Equation~\eqref{eq:cgl_nonlinear_evp}, viz. $|\gamma|/w^2=aR^2/w$, where $|\psi|^2\sim 1/w$ such that $\|\psi\|_2^2=1$.  Thus, $w=|\gamma|/(aR^2)$, as $R\rightarrow\infty$.  This estimate is then substituted into $\Im(\Omega)=-\gmr\|\partial_xu_0\|_2^2-(1/2)\mu_{2\mathrm{r}}\|x u_0\|_2^2+\mu_0$, such that $\Im(\Omega)\sim -\gmr /w^2$, as $R\rightarrow\infty$, hence
\begin{equation}
\Im(\Omega)\sim -\gmr a^2R^4/|\gamma|^2,\qquad R\rightarrow\infty.
\label{eq:omega_heuristic}
\end{equation}
Combining Equations~\eqref{eq:evol_l2_bound} and~\eqref{eq:omega_heuristic}, we have
\[
\mu_{0}-\sqrt{\frac{\mu_{2\mathrm{r}}\gmr}{2}}\geq \Im(\Omega)\sim -\gmr a^2R^4/|\gamma|^2,\qquad R\rightarrow\infty.
\]
The set of zeros $\{R_0|\Omega_{\mathrm{i}}(R_0)=0\}$ is empty if $\mu_{0}<\sqrt{\mu_{2\mathrm{r}}\gmr/2}$.  This case is ruled out because of the operational range of the parameters (Equation~\eqref{eq:cgl_eigs_stability_range}).  Also, the set of zeros is bounded above in view of Equation~\eqref{eq:omega_heuristic}.
Assuming $\Omega_{\mathrm{i}}(R)$ is a non-zero continuous function it follows that the number of zeros is finite,  only a finite number of global modes exists.  These fairly heuristic arguments are confirmed by numerical calculations in the next section.

\section{CGL equation -- Illustrative numerical studies}
\label{sec:cgl_numerical}

In this section, we compute eigensolutions of Equation~\eqref{eq:cgl_nonlinear_ansatz} non-perturbatively.  The aim here is to see if the intuition regarding a second global mode built up in Section~\ref{sec:cgl_numerical} is correct.
The partial differential equation~\eqref{eq:cgl_nonlinear_ansatz} is discretized on a uniform grid and periodic boundary conditions are assumed.  This is a reasonable asssumption: we expect localized eigenfunctions that decay far from the origin, meaning that the details of the solution for large $|x|$ are not important (this is checked \textit{a posteriori}).  In this way, Equation~\eqref{eq:cgl_nonlinear_ansatz} is converted into a nonlinear algebraic problem: the solution $\psi(x)$ is converted into a column vector $\psi(x_i)$, where $x_i$ is a point on the discrete grid, the operators $\partial_x$ and $\partial_{xx}$ become sparse matrices operating on the column vector $\psi(x_i)$, and multiplication of $\psi(x)$ by a function converts into multiplication of $\psi(x_i)$ by a diagonal matrix.   The resulting nonlinear algebraic eigenvalue problem is solved by Rayleigh-Quotient Iteration (RQQI) in a manner that is identical to Equation~\eqref{eq:rqi_twolevel}.  Furthermore, simple numerical continuation is used: the nonlinear algebraic eigenvalue problem with $R=\Delta R$ is solved, using the `ground-state' solution to the corresponding linear eigenvalue problem (i.e. $R=0$, and the eigenvalue $n=0$ in Equation~\eqref{eq:cgl_eigs_formula_linear}) as an initial guess.  Here, $\Delta R$ is assumed to be small in an appropriate sense.  The eigensolution at $R=\Delta R$ is obtained when the RQI method converges.  Then, this eigensolution is used as an initial guess for the nonlinear algebraic eigenvalue problem with $R=2\Delta R$.  The continuation proceeds until the final desired value of $R$ is obtained.  This approach has the additional advantage that it yields the functional dependence of $\Omega$ on $R$ -- this can then be used to obtain self-consistent solutions.

\subsection{Near criticality, with $U=0$}

The correctness of the numerical method is checked by carrying out several tests.  First, the numerical method is tested for convergence: typically, the method converges to four significant figures when the number of gridpoints is $300$.  Secondly, we have checked that the numerical method at $R=0$ produces eigensolutions that are consistent with the known analytical solution in the linear case.  Thirdly, the function $\Omega(R)$ is shown to agree with standard perturbation theory for small $R$.  This third aspect is now described in detail.  The following parameter values are chosen:
\begin{equation}
\gamma=1-\imag,\qquad \mu_2=2\gamma\left(0.1\mathe^{\imag (2\pi\times 0.15)}\right),\qquad U=0,\qquad \mu_0=0.99\mu_{0\mathrm{c}},\qquad a=1.
\label{eq:params_u0}
\end{equation}
%
The function $\Omega(R)$ is computed and the results compared with perturbation theory in Figure~\ref{fig:mypert1}.
\begin{figure}[t]
\centering
\includegraphics[width=0.6\textwidth]{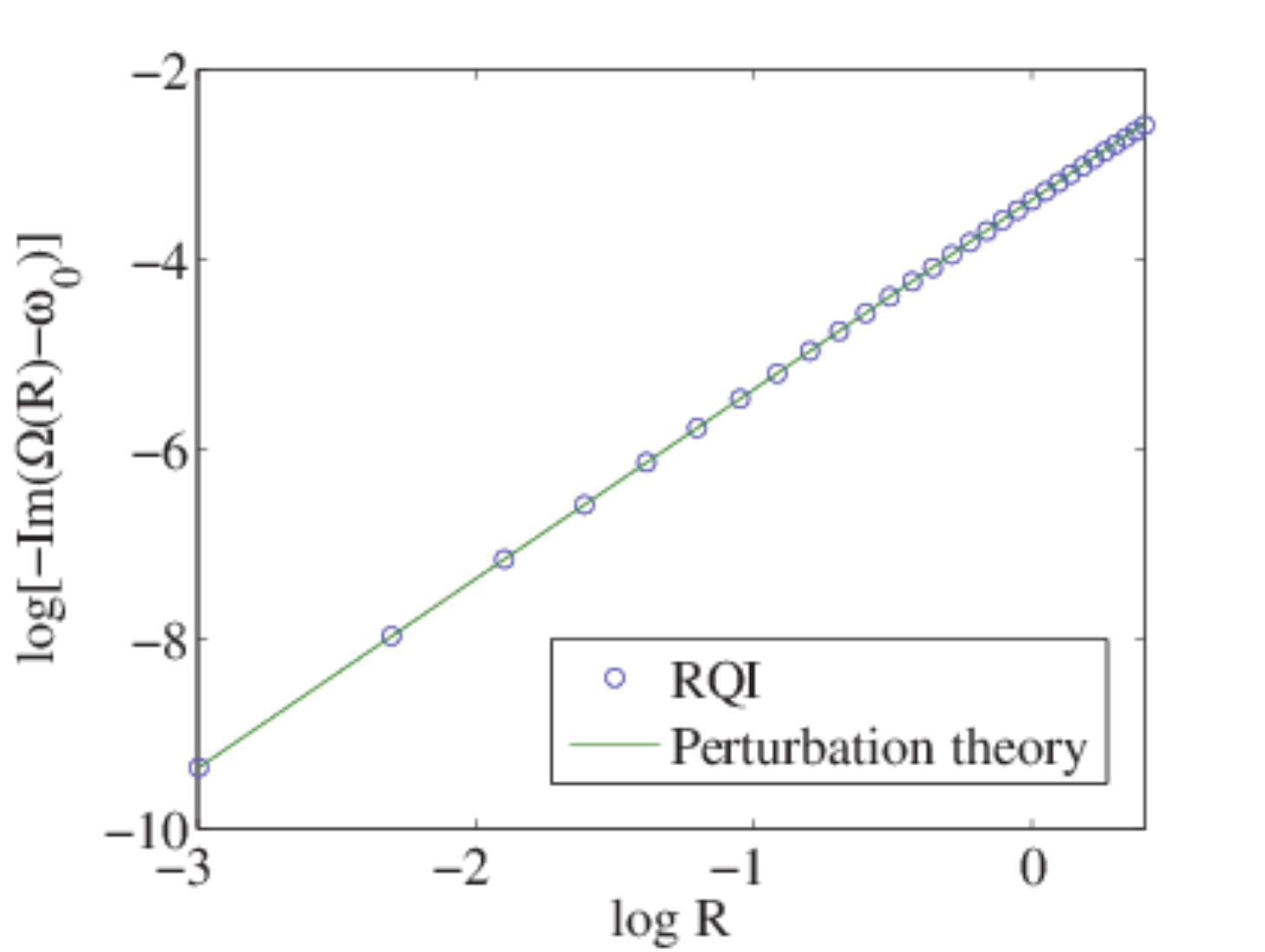}
\caption{Plot of $\Im\left[\Omega(R)-\omega_0\right]$ as a function of $R$ (log-log scale) showing the agreement between numerics and perturbation theory and confirming the correctness of the numerical method.  The parameters are given in Equation~\eqref{eq:params_u0}}
\label{fig:mypert1}
\end{figure}
Excellent agreement between the numerics and the perturbation theory is obtained, confirming the correctness of the former.  A global mode is found at the critical value   $R_\mathrm{c}=0.355$. 
The numerical method is used to continue the plot in Figure~\ref{fig:mypert1} to include a much wider range of $R$-values in Figure~\ref{fig:mypert2}(a).
\begin{figure}[t]
\centering
\includegraphics[width=0.6\textwidth]{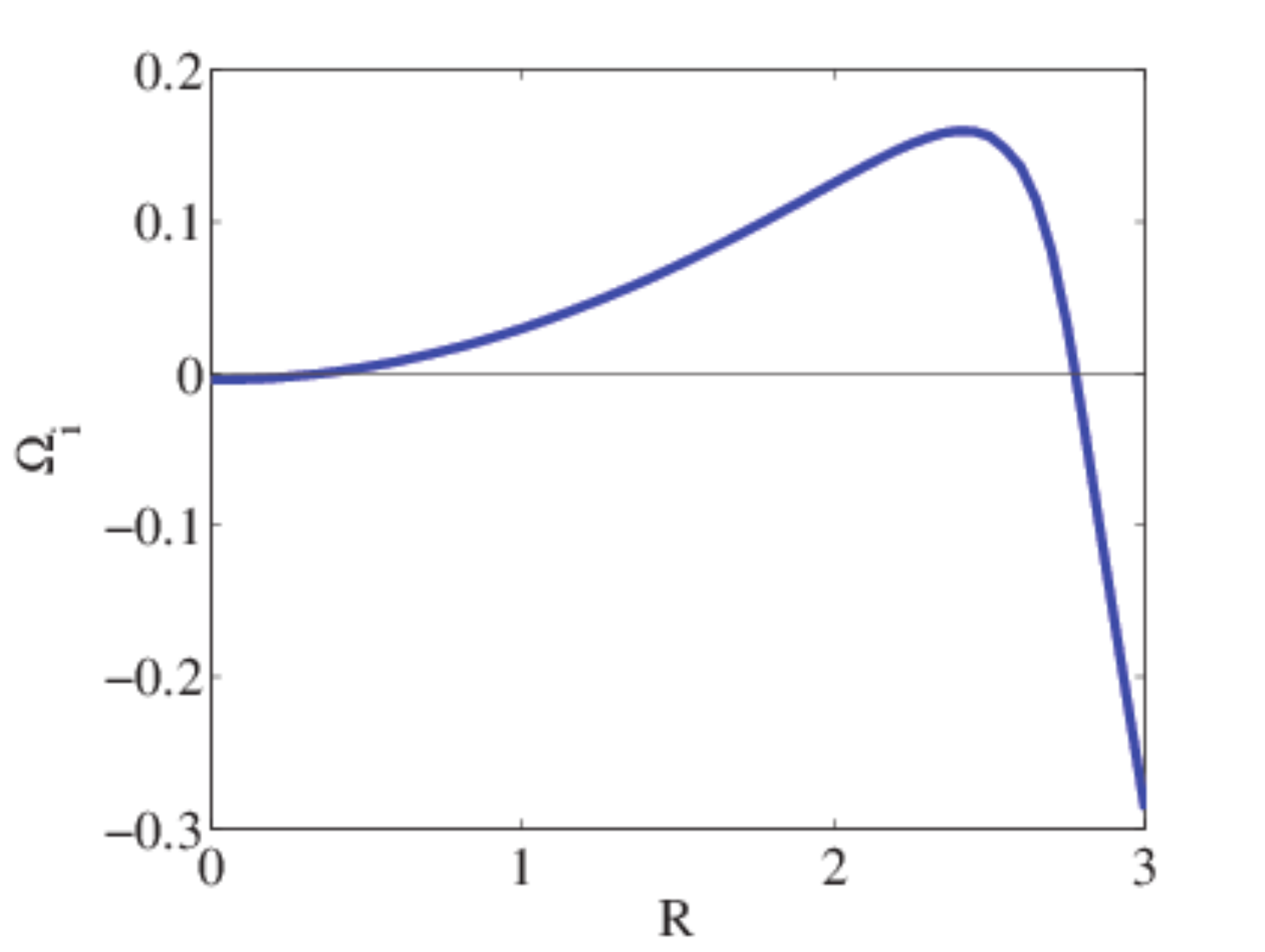}\\
\includegraphics[width=0.48\textwidth]{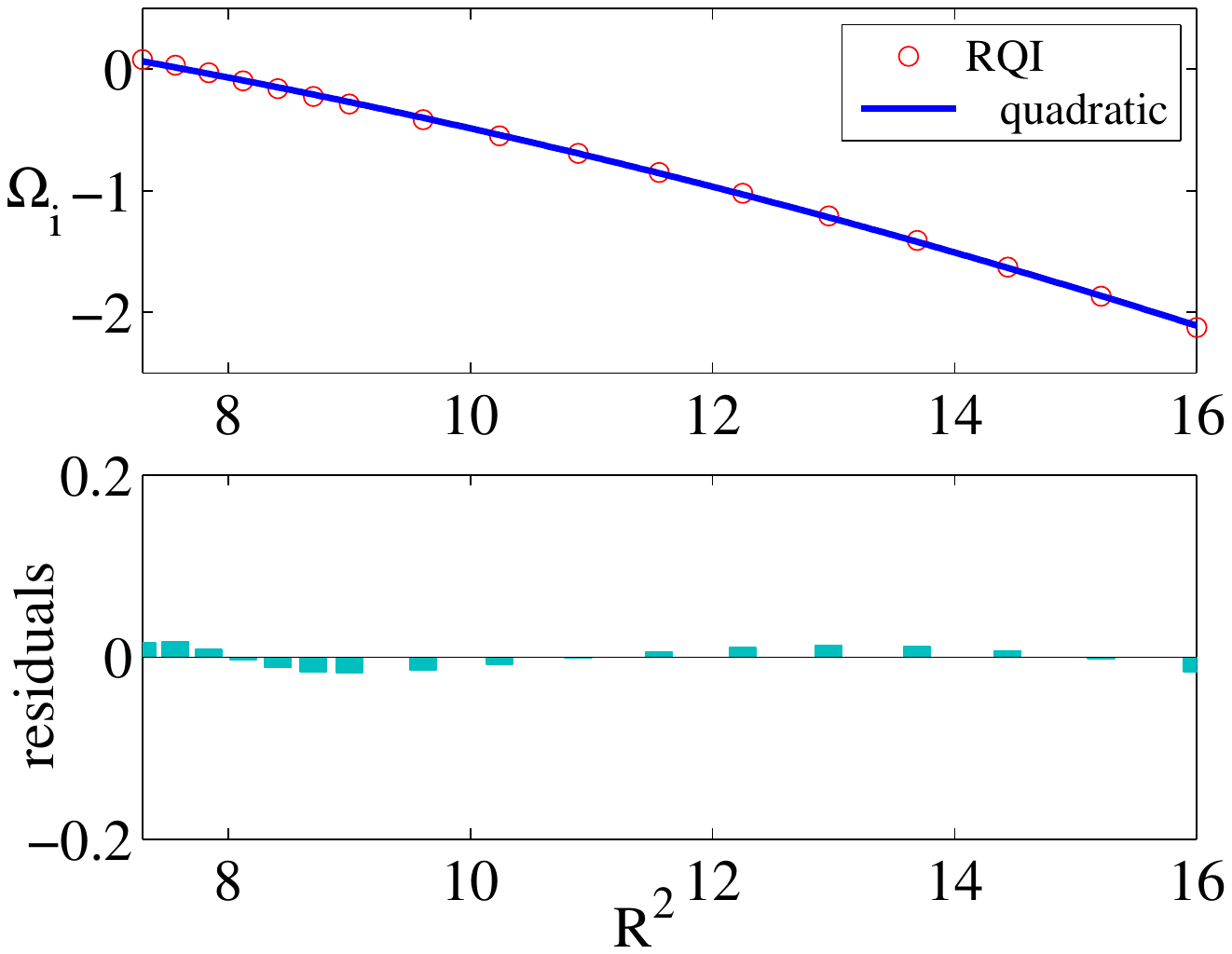}
\includegraphics[width=0.48\textwidth]{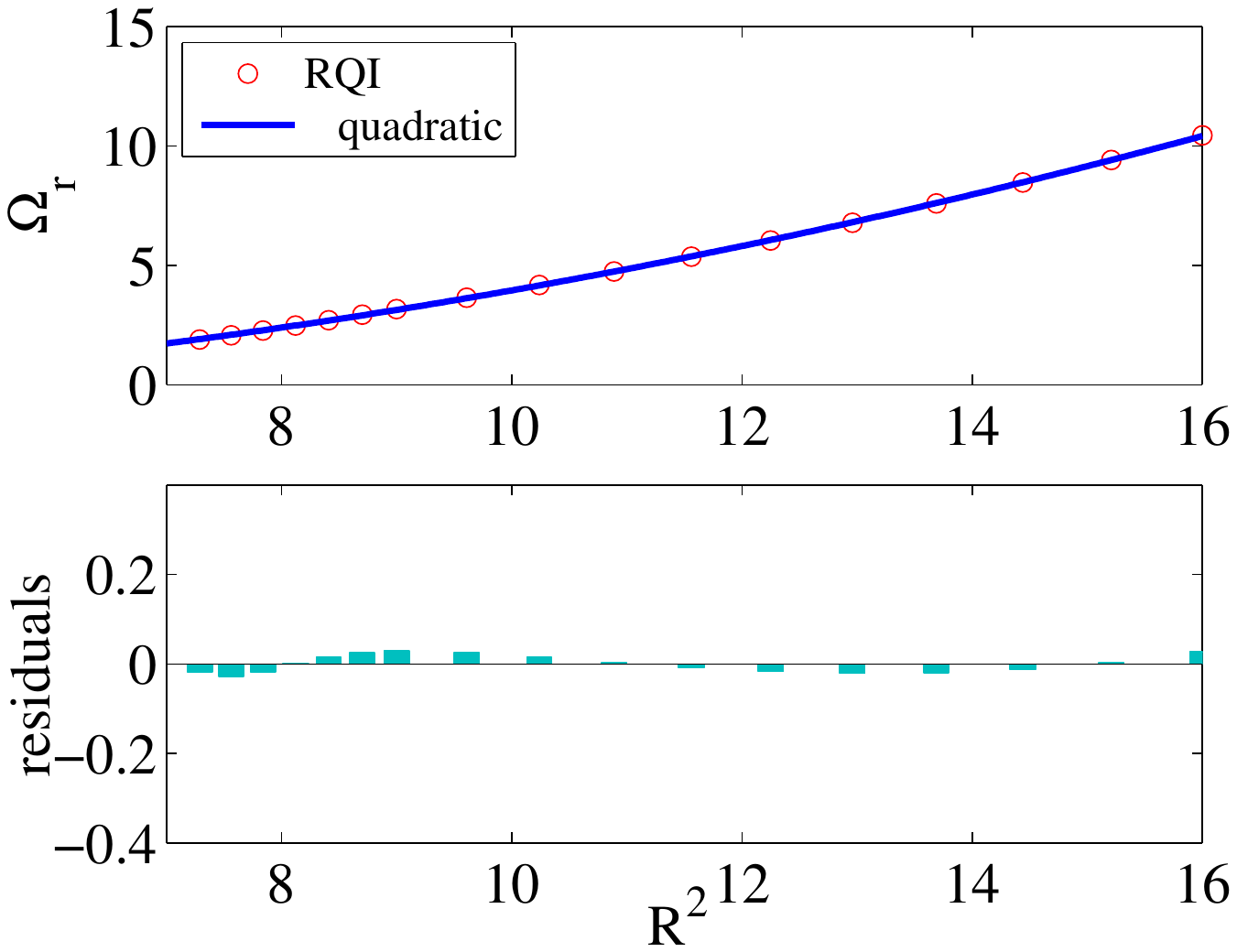}
\caption{(a) Plot of $\Im(\Omega)$ as a function of $R$ for the parameters given in Equation~\eqref{eq:params_u0} (RQI numerical method); (b),(c) Similar to (a), but showing the large-R behaviour of $\Im(\Omega)$ and $\Re(\Omega)$ as a function of $R^2$.}
\label{fig:mypert2}
\end{figure}
The perturbation method breaks down at $R\approx 1$.  At $R_\mathrm{c}=2.77$ a second global mode appears, with (by definition) a large amplitude and hence undetectable using the weakly nonlinear analyses used previously.  Also, the large-$R$ behaviour of $\Omega(R)$ is consistent with the heuristic arguments developed previously, with $\Im(\Omega),\Re(\Omega)\sim R^4$ as $R\rightarrow\infty$: in Figure~\ref{fig:mypert2}(b,c) a fit $\Im(\Omega),\Re(\Omega)=AR^2+BR^4$ is applied to the data and excellent agreement is obtained in the limit of large $R$ (a similar fit provides excellent agreement at small $R$, but with different fitting coefficients, and for reasons not relating to the heuristics of Section~\ref{sec:global_modes_sl}).

The instability of the two global modes is now checked.  Although Theorem~\ref{thm:cgl_stability} guarantees that the global mode at $R_\mathrm{c}=0.355$ is linearly unstable, this is now checked independently via 
Floquet analysis; the same analysis can be used for the large-amplitude global mode at $R_\mathrm{c}=2.77$, for which Theorem~\ref{thm:cgl_stability}  does not apply.  Thus,  the trial solution $\delta v=\alpha(x)\mathe^{\sigma t}+\beta(x)\mathe^{\overline{\sigma}t}$ is substituted into Equation~\eqref{eq:stability_cgl} and the following eigenvalue problem is obtained:
\begin{equation}
\imag \sigma \left(\begin{array}{c}\alpha\\\overline{\beta}\end{array}\right)=
\left(\begin{array}{cc}\mathcal{L}-\Omega + 2R_\mathrm{c}^2|u_0|^2& R_\mathrm{c}^2 u_0^2\\
-R_\mathrm{c}^2 \overline{u_0^2}&-\overline{\mathcal{L}}+\Omega-2R_\mathrm{c}^2|u_0|^2\end{array}\right)
\left(\begin{array}{c}\alpha\\\overline{\beta}\end{array}\right).
\label{eq:cgl_floquet}
\end{equation}
It is clear from Equation~\eqref{eq:cgl_floquet} that the eigenvalues $\sigma$ are purely real, or come in complex-conjugate pairs.  

The case with $R_\mathrm{c}=0.355$ is investigated  in  Figure~\ref{fig:floquet_u0}.  The eigenvalue with maximum real part is itself purely real and positive ($\sigma=0.008396$), thereby confirming the asymptotic result in Theorem~\ref{thm:cgl_stability}.
The corresponding $\alpha(x)$ and $\beta(x)$ are also shown in Figure~\ref{fig:floquet_u0}.
\begin{figure}[t]
\centering
\subfigure[]{\includegraphics[width=0.32\textwidth]{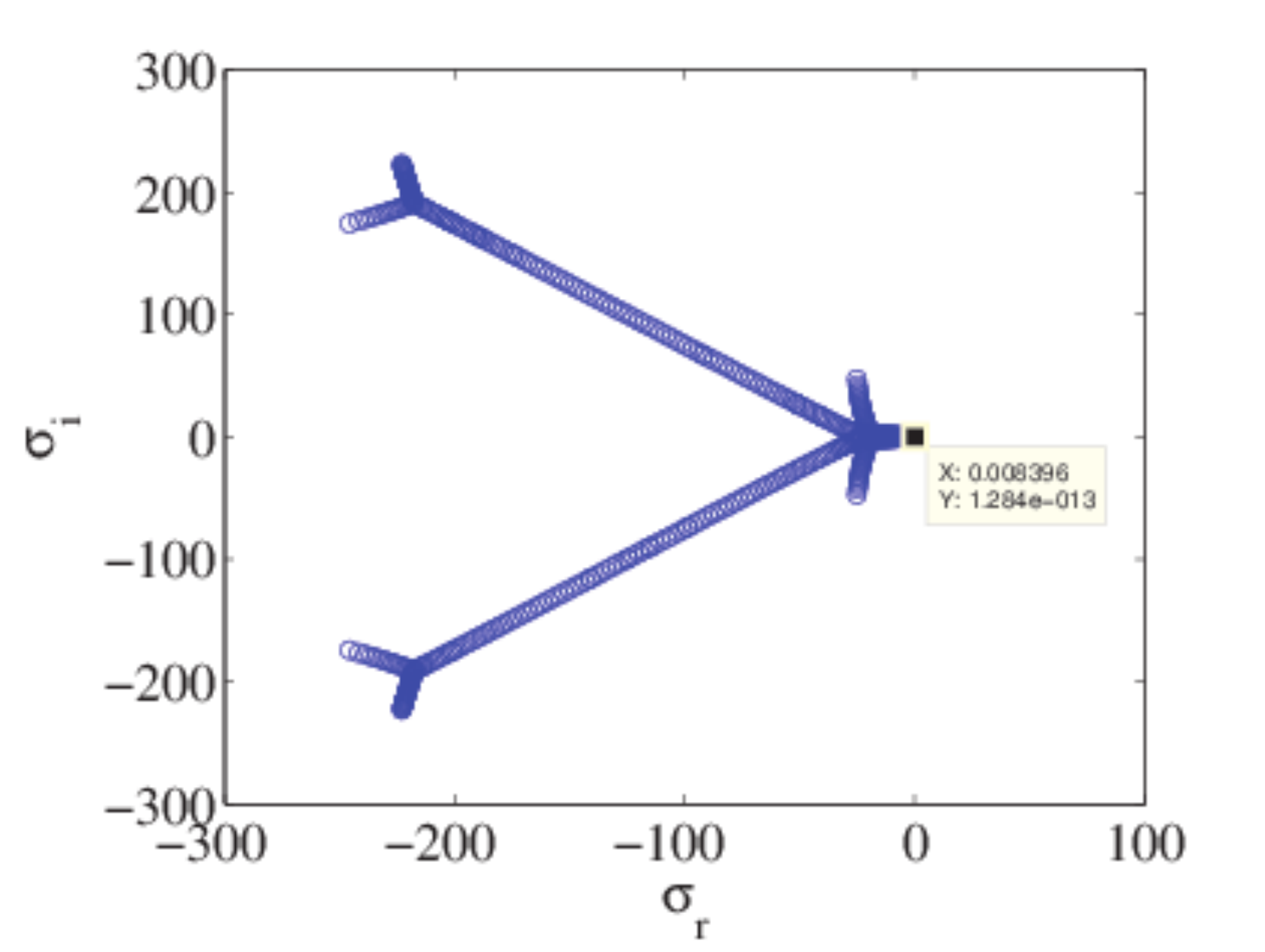}}
\subfigure[]{\includegraphics[width=0.32\textwidth]{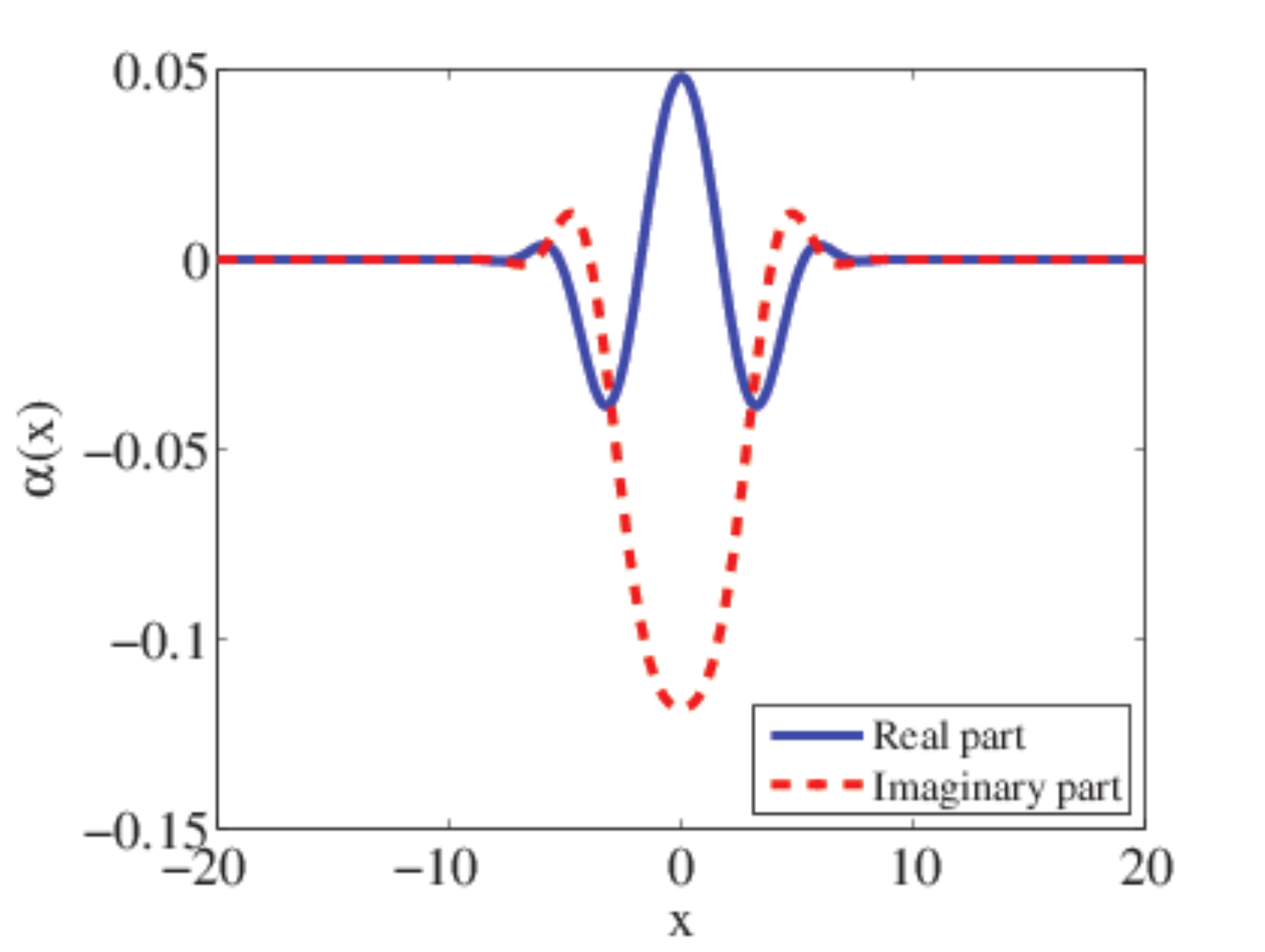}}
\subfigure[]{\includegraphics[width=0.32\textwidth]{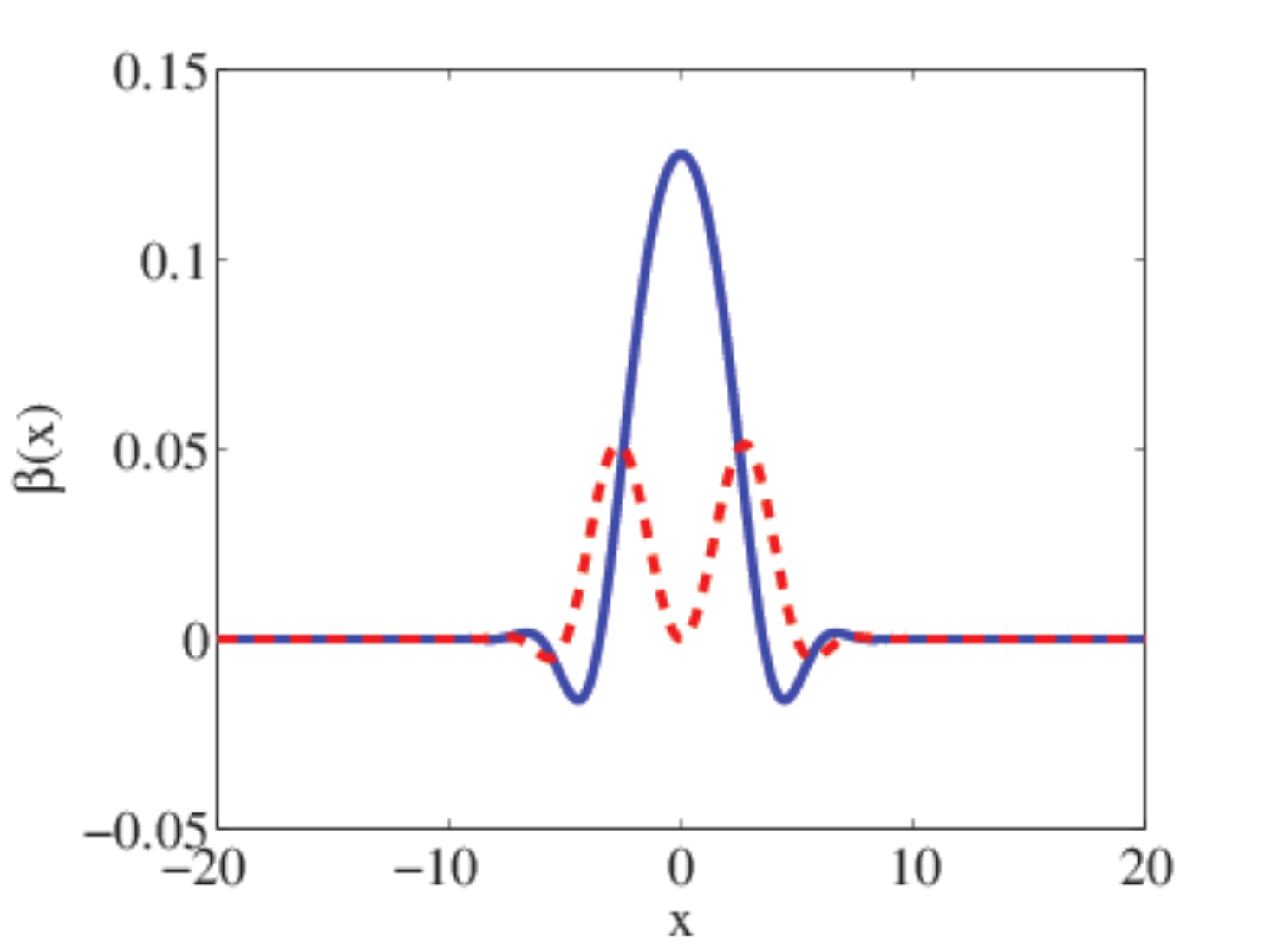}}
\caption{Panel~(a): Spectrum of Floquet exponents corresponding to the linear instability of the global mode with parameters~\eqref{eq:params_u0} and $R_\mathrm{c}=0.355$; Panels (b)-(c): the eigenfunction of the most-dangerous mode arising from the Floquet analysis. }
\label{fig:floquet_u0}
\end{figure}
The projection of $\alpha(x), \beta(x)$ on to the eigenfunctions $\psi_p^{\dagger}$ is shown in Table~\ref{tab:cgl_stability} for the first few $p$-values, confirming that the $p=0$ contribution is dominant, consistent with the analysis in Theorem~\ref{thm:cgl_stability}.
\begin{table}[t]
\centering
\begin{tabular}{|c|c|c|}
\hline
$p$ & $|\langle \psi_p^\dagger,\alpha\rangle|$ & $|\langle \psi_p^\dagger,\beta\rangle|$\\
\hline
\hline
$\phantom{aa}$ 0  $\phantom{aa}$ &  0.5439                 &    0.5439         \\
$\phantom{aa}$ 1  $\phantom{aa}$ &  $9.81\times 10^{-14}$  &    $4.54\times 10^{-14}$   \\
$\phantom{aa}$ 2  $\phantom{aa}$ &  0.0013                 &    0.0013                  \\
$\phantom{aa}$ 3  $\phantom{aa}$ &  $2.81 \times 10^{-14}$ &    $ 1.25\times 10^{-14}$  \\
$\phantom{aa}$ 4  $\phantom{aa}$ &  $8.64 \times 10^{-5}$  &    $8.64\times 10^{-5}$    \\
\hline
\end{tabular}
\caption{Coefficients of the expansion of $\alpha(x)$ and $\beta(x)$ in terms of the eigenbasis $\{\psi_p(x)\}_{p=0}^\infty$ ($R_\mathrm{c}=0.355$).}
\label{tab:cgl_stability}
\end{table}
The unstable mode is an even function of $x$, since the projection of the eigenfunctions onto the odd adjoint eigenfunctions is zero. 
The linear stability of the case with $R_\mathrm{c}=2.77$ is investigated next, in the same manner as before (Figure~\ref{fig:floquet_u1}).  This global mode is also linearly unstable, and the most-dangerous Floquet eigenvector is an odd function of $x$, as evidenced by the figure and by Table~\ref{tab:cgl_stability2}.
\begin{figure}[t]
\centering
\subfigure[]{\includegraphics[width=0.32\textwidth]{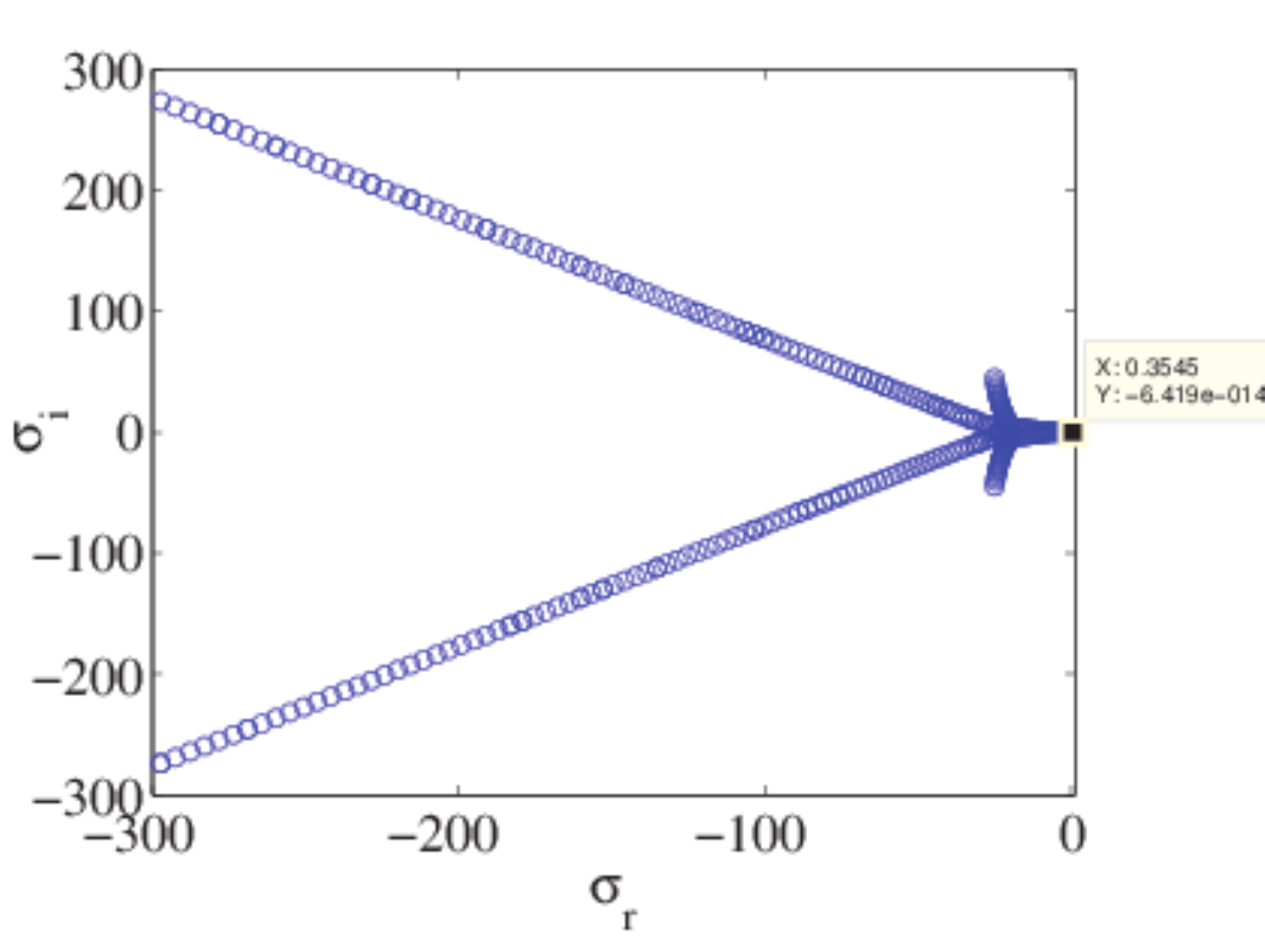}}
\subfigure[]{\includegraphics[width=0.32\textwidth]{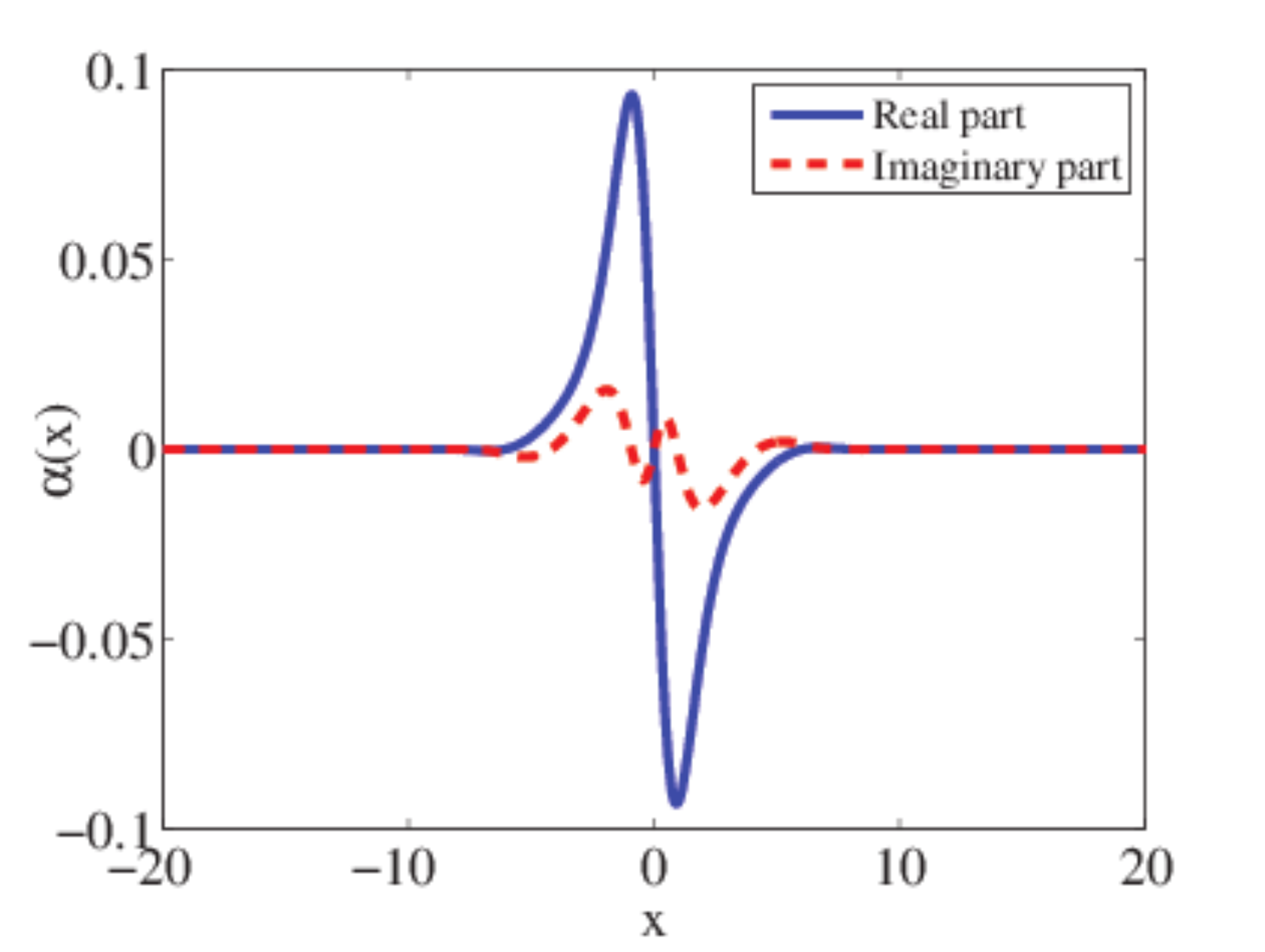}}
\subfigure[]{\includegraphics[width=0.32\textwidth]{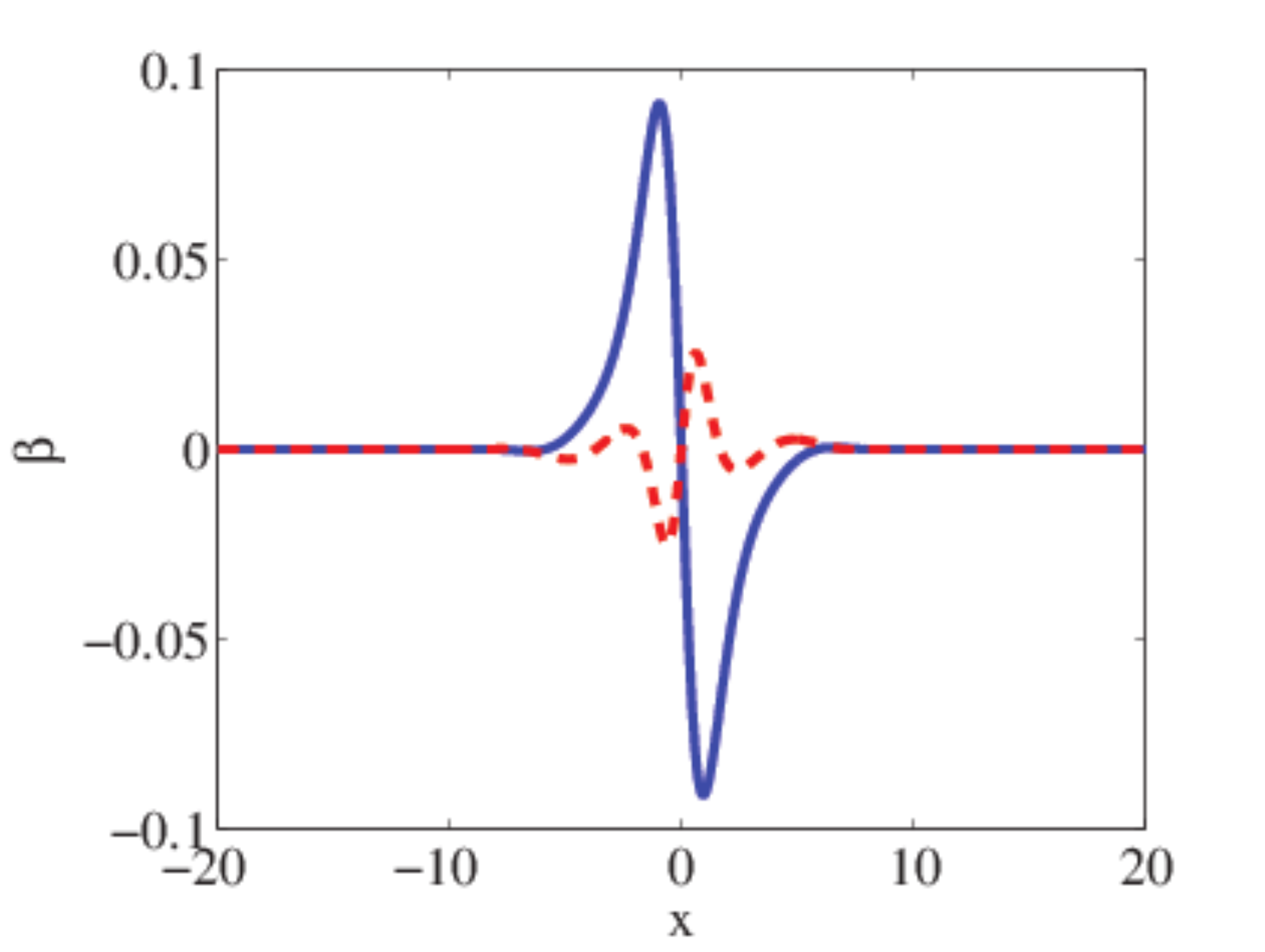}}
\caption{Panel~(a): Spectrum of Floquet exponents corresponding to the linear instability of the global mode with parameters~\eqref{eq:params_u0} and $R_\mathrm{c}=2.77$; Panels (b)-(c): the eigenfunction of the most-dangerous mode arising from the Floquet analysis. }
\label{fig:floquet_u1}
\end{figure}
\begin{table}[t]
\centering
\begin{tabular}{|c|c|c|}
\hline
$p$ & $|\langle \psi_p^\dagger,\alpha\rangle|$ & $|\langle \psi_p^\dagger,\beta\rangle|$\\
\hline
\hline
$\phantom{aa}$ 0  $\phantom{aa}$ &  $1.94\times 10^{-10}$  &    $1.92\times 10^{-10}$         \\
$\phantom{aa}$ 1  $\phantom{aa}$ &  $0.488$                &    $0.488$   \\
$\phantom{aa}$ 2  $\phantom{aa}$ &  $1.64\times 10^{-10}$  &    $1.64\times 10^{-10}$              \\
$\phantom{aa}$ 3  $\phantom{aa}$ &  $0.141$                &    $0.141$  \\
$\phantom{aa}$ 4  $\phantom{aa}$ &  $5.72 \times 10^{-11}$ &    $5.72 \times 10^{-11}$    \\
$\phantom{aa}$ 5  $\phantom{aa}$ &  $0.0338$               &    $0.0338$    \\
$\phantom{aa}$ 6  $\phantom{aa}$ &  $1.66 \times 10^{-11}$ &    $1.66 \times 10^{-11}$    \\
$\phantom{aa}$ 7  $\phantom{aa}$ &  $0.00812$              &    $0.00812$    \\
$\phantom{aa}$ 8  $\phantom{aa}$ &  $4.62 \times 10^{-12}$ &    $4.62\times 10^{-12}$    \\
$\phantom{aa}$ 9  $\phantom{aa}$ &  $0.00200$              &    $0.00200$   \\
\hline
\end{tabular}
\caption{Coefficients of the expansion of $\alpha(x)$ and $\beta(x)$ in terms of the eigenbasis $\{\psi_p(x)\}_{p=0}^\infty$ ($R_\mathrm{c}=2.77$).}
\label{tab:cgl_stability2}
\end{table}
For both global modes, we have $|\langle \psi_p^\dagger,\alpha\rangle|=|\langle\psi_p^\dagger,\beta\rangle|$.  This is a consequence of the structure of the eigenvalue problem~\eqref{eq:cgl_floquet}: not only are the eigenvalues real or complex-conjugate pairs, but also, $\alpha(x)=\beta(x)\mathe^{\imag\varphi}$, where $\varphi$ is a constant phase.

For the particular parameter choice in Equation~\eqref{eq:params_u0}, we have uncovered two unstable global modes.  It is of interest to understand how these global modes manifest themselves in the context of the full time-dependent equation~\eqref{eq:cgl_nonlinear}.  For that purpose, we have carried out direct numerical simulations of Equation~\eqref{eq:cgl_nonlinear} using a pseudospectral code and a fourth-order Adams--Bashforth temporal discretization.  The initial condition
\begin{equation}
u(x,t=0)=k_1\psi_0(x)[1+k_2\psi_3(x)],
\label{eq:ic_cgl}
\end{equation}
where $k_1$ and $k_2=0.00001$ are real constants, and $k_1$ is chosen such that $|u(x,t=0)|=\rho$, where $\rho$ is a parameter to be varied.
The timestep is chosen as $\Delta t=10^{-3}$, the simulation domain has length $L=40$, and $1000$ gridpoints are used.  These numerical parameters are adjusted to check that the simulation has converged numerically.  A spacetime plot of the absolute value of the solution $|u(x,t)|$ is shown in Figure~\ref{fig:cgl_spacetime1}.  The spatial structure is localized in space shows a characteristic oscillation in time.  However, the temporal oscillations are not persistent: there is a quasi-periodic drift away from the oscillatory state, and then a sudden `snapping back' to the oscillatory state. 

To connect these data to the global modes found in the theoretical considerations, consideration is given to the quantities $X(t)=\Re[u(0,t)]$,  $Y(t)=\Im[u(0,t)]$, for various values of $\rho=|u(x,t=0)|$.
\begin{figure}[t]
\centering
\includegraphics[width=0.6\textwidth]{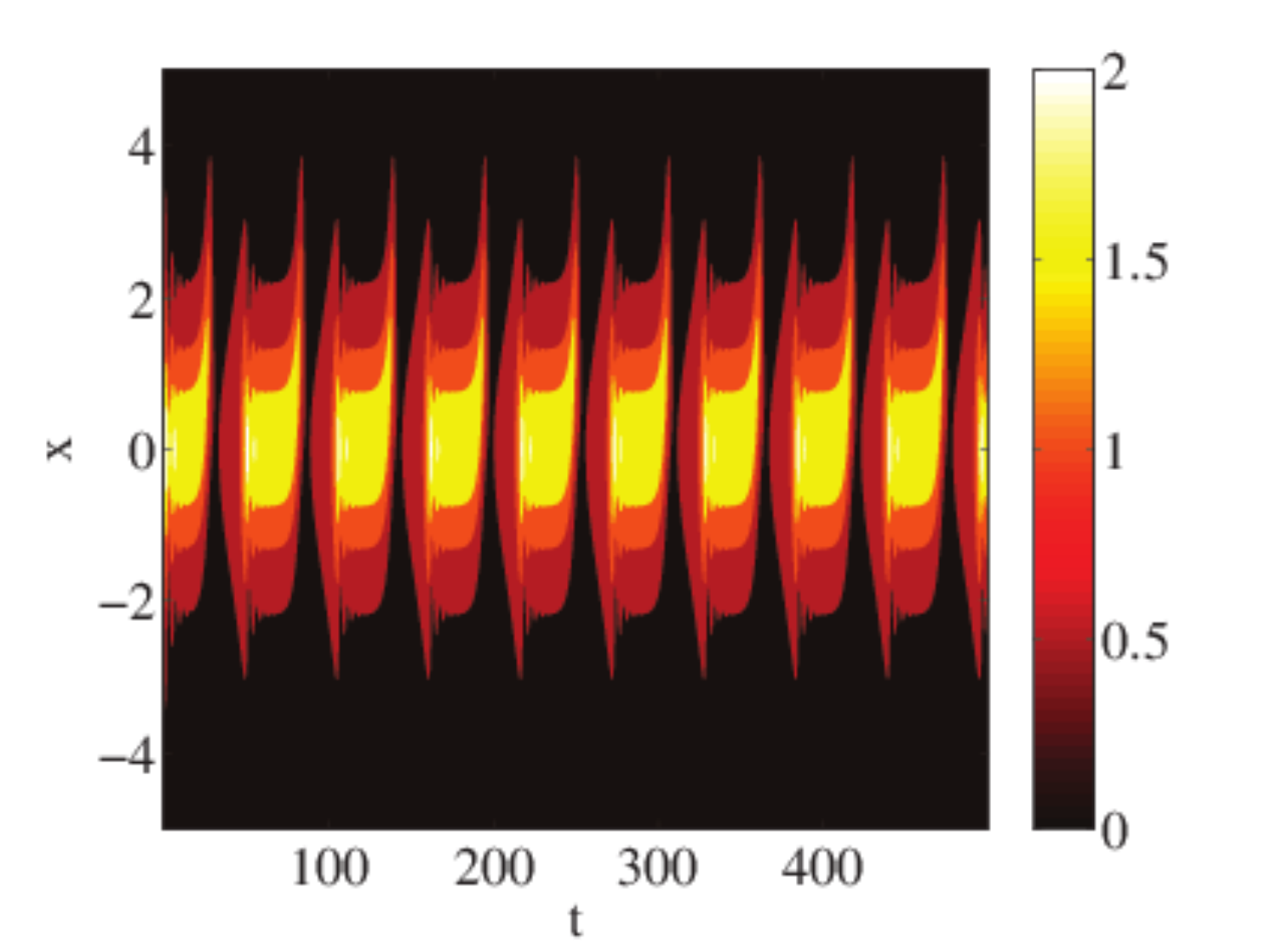}
\caption{Spacetime plot of $|u(x,t)|$ for the direct numerical simulation.  The parameters are as given in Equations~\eqref{eq:params_u0} and~\eqref{eq:ic_cgl}, and $\rho:=|u(x,t=0)|=3$. }
\label{fig:cgl_spacetime1}
\end{figure}
These results are shown in Figure~\ref{fig:cgl_spacetime1}.
\begin{figure}[htb]
\centering
\subfigure[$\,\,\rho=0.2$]{\includegraphics[width=0.45\textwidth]{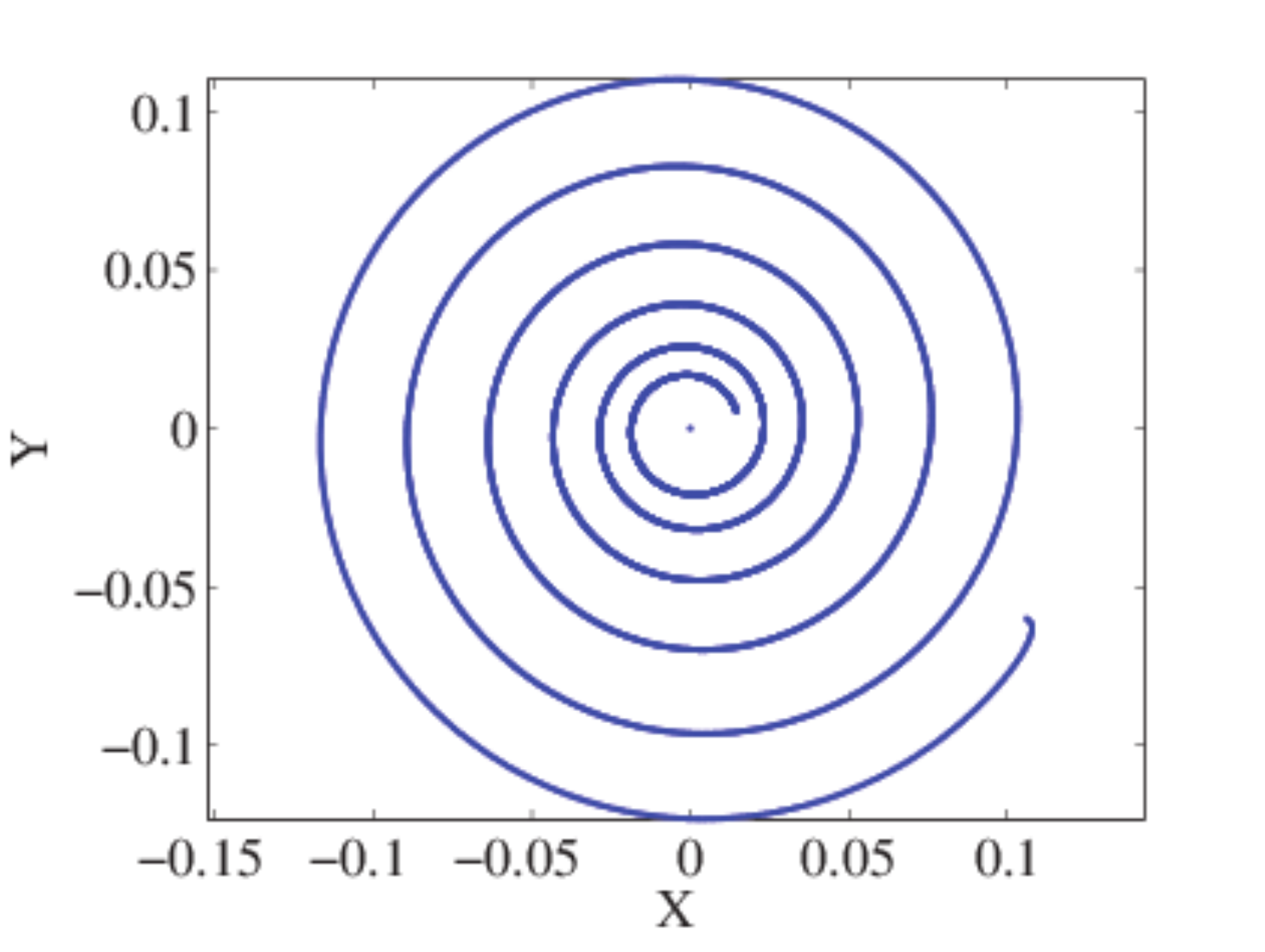}}
\subfigure[$\,\,\rho=0.3$]{\includegraphics[width=0.45\textwidth]{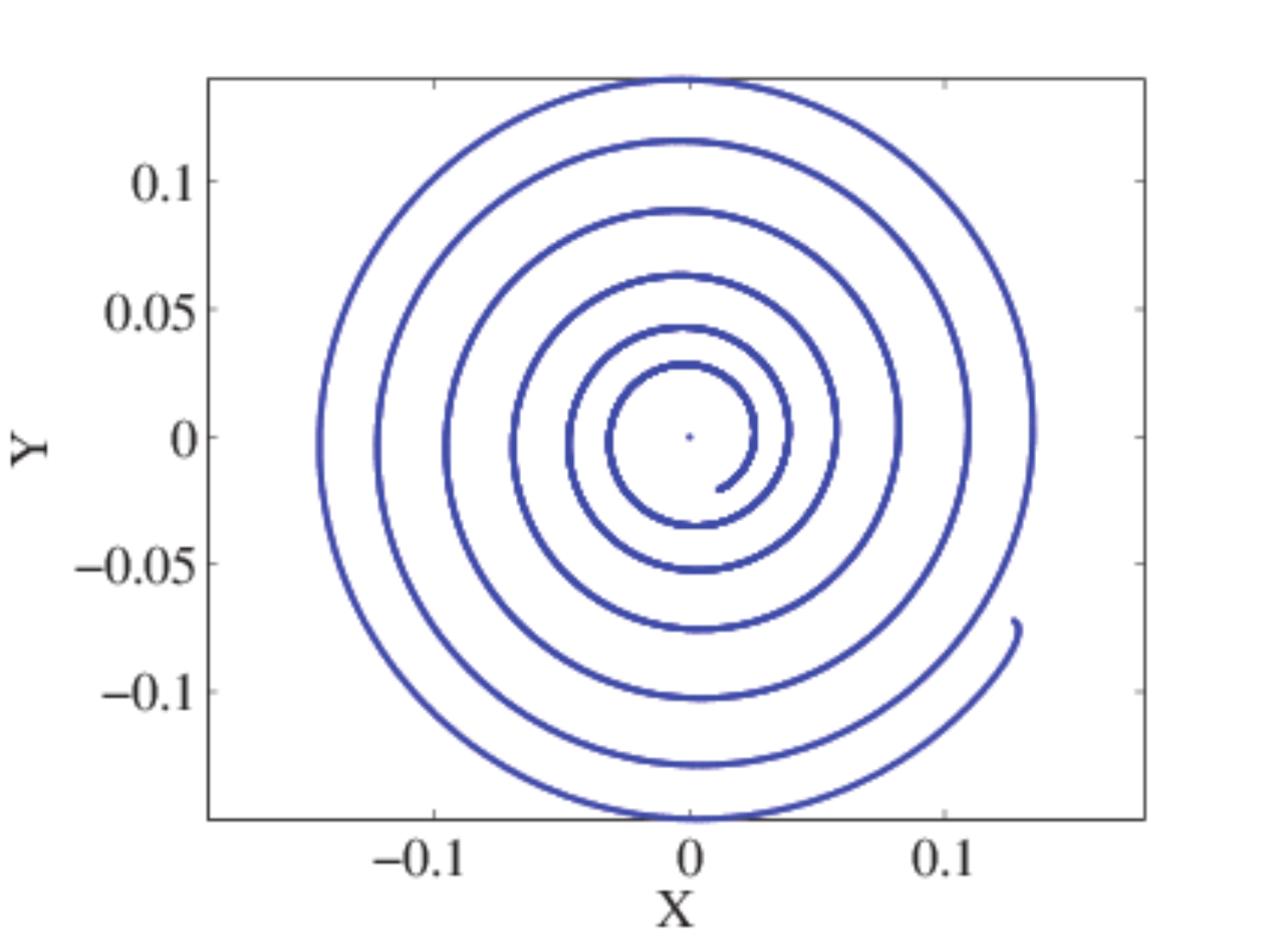}}\\
\subfigure[$\,\,\rho=0.4$]{\includegraphics[width=0.32\textwidth]{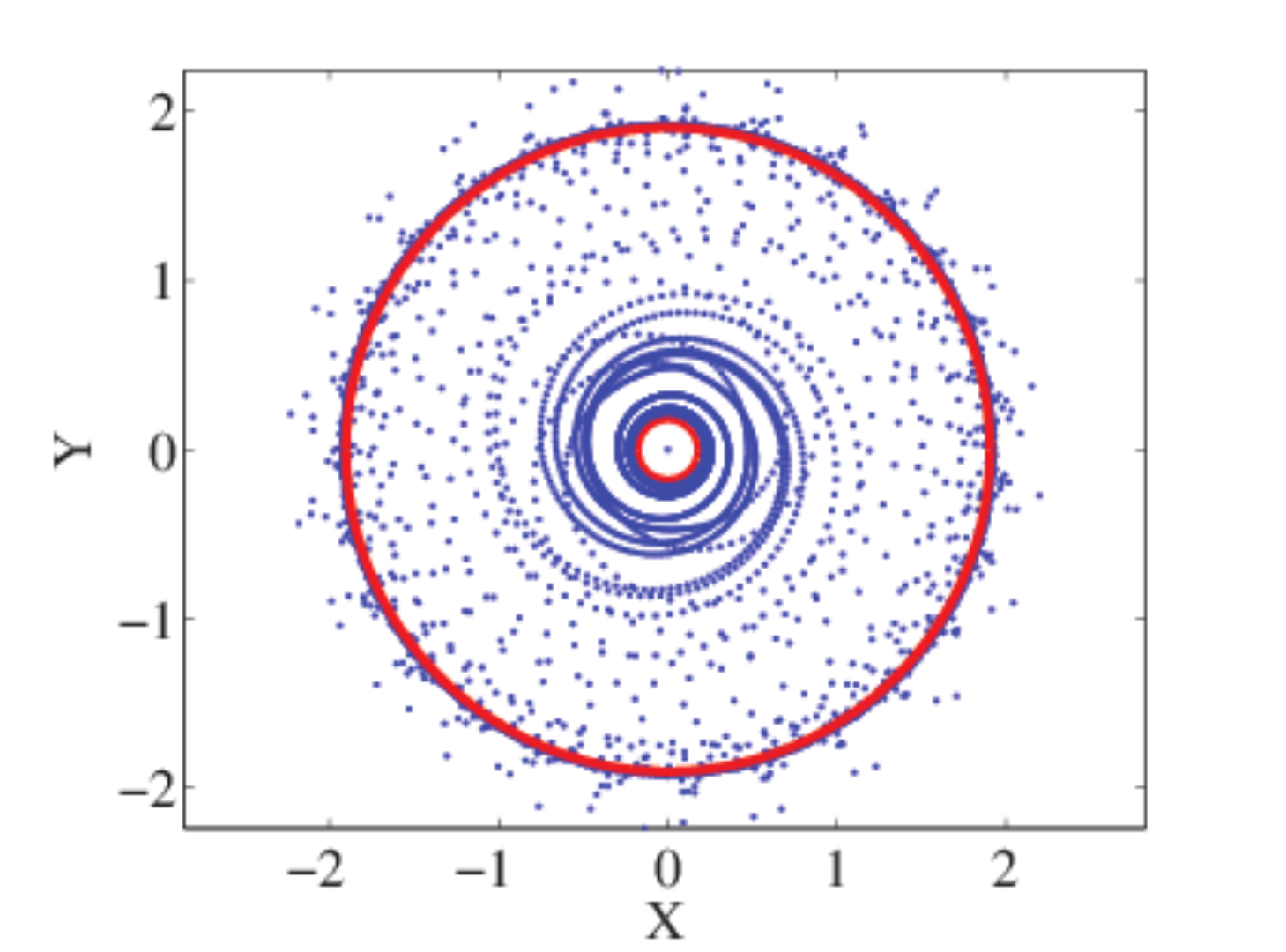}}
\subfigure[$\,\,\rho=3$]{\includegraphics[width=0.32\textwidth]{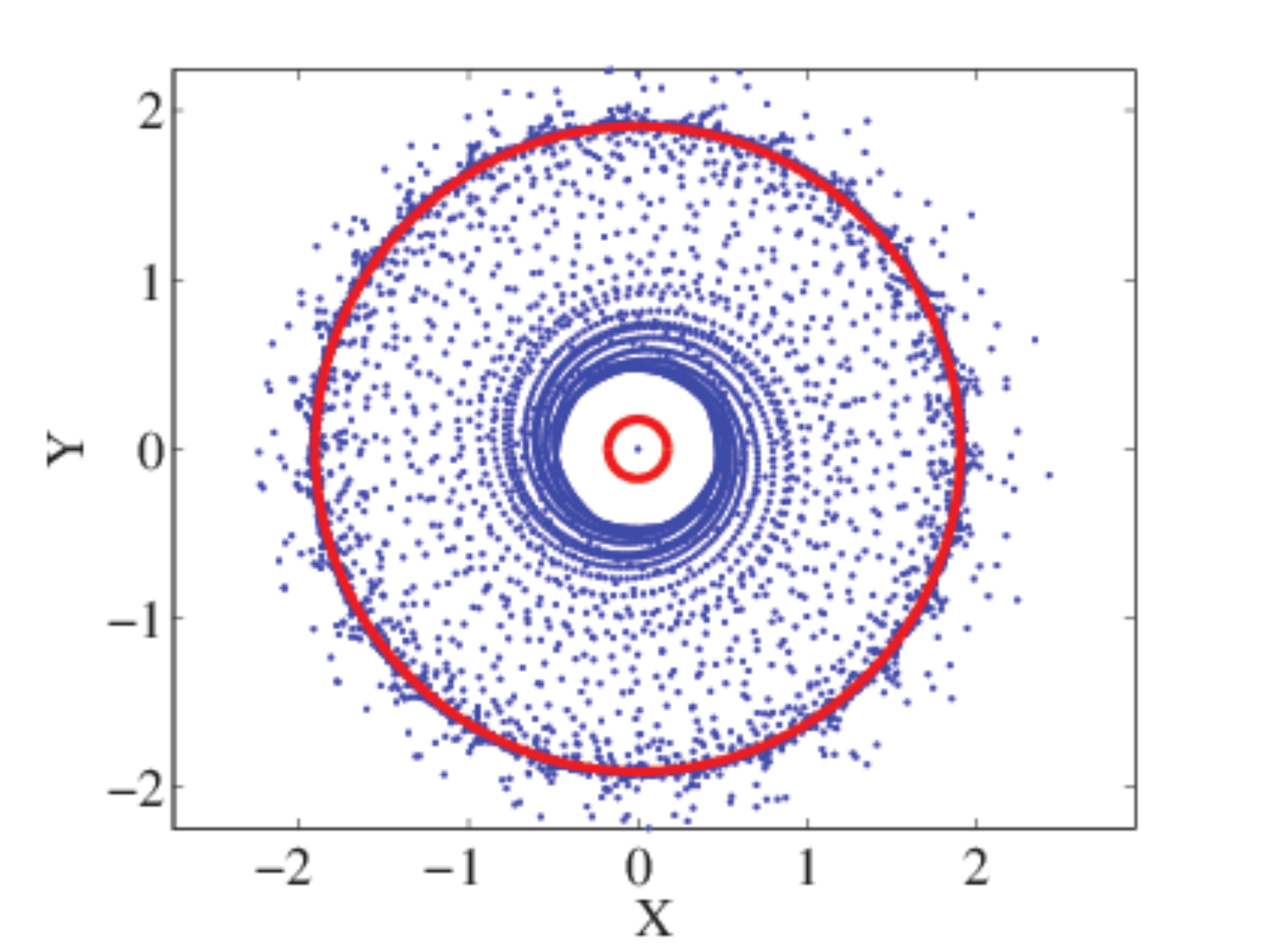}}
\subfigure[$\,\,\rho=100$]{\includegraphics[width=0.32\textwidth]{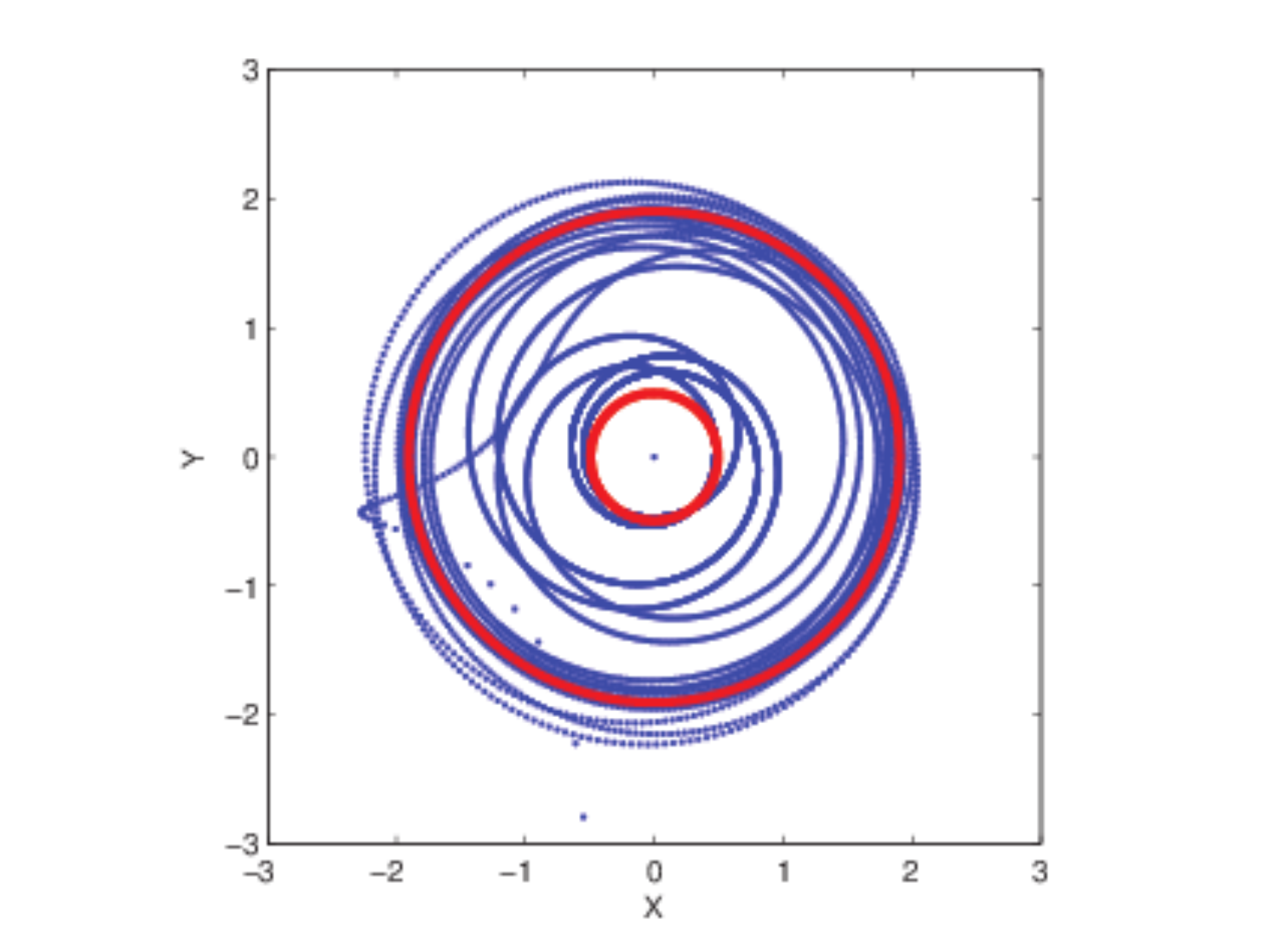}}
\caption{Plot of $Y(t)$ versus $X(t)$ for various values of $\rho$.  In Panels (c)--(d) the radii of the global modes   $\mathpzc{R}_{1,2}:=|u_{01,2}(0)|R_\mathrm{c1,2}$ in the $XY$ space are also shown.}
\label{fig:cgl_spacetime1}
\end{figure}
Panels (a)--(b) correspond to a weakly nonlinear scenario wherein the disturbance $u(x,t)$ decays to zero: the initial condition is repelled from the unstable global mode at $R_\mathrm{c1}=0.355$ (mode 1) and spirals towards the linearly stable fixed point $u=0$.  In Panel (c) there is a sharp change: the initial condition is again repelled from the unstable global mode at $R_\mathrm{c1}=0.355$ but attracted to the \textit{unstable} mode at $R_\mathrm{c2}=2.77$ (mode 2).  Mode 2 is clearly visible in Panel~(c) as a circle of radius $\radiustwo:=|u_{02}(0)|R_\mathrm{c2}=1.905$, where $u_{02}(x)$ is the spatial part of the second global mode.  Now, although mode 2 is a linearly unstable periodic orbit, it is unstable only in a one-sided sense: when the system occupies a state wherein $\sqrt{X^2+Y^2}>\radiustwo$, the system is drawn back towards mode 2, while for $\sqrt{X^2+Y^2}<\radiustwo$ the system is drawn towards the fixed point at zero.  However, the system cannot attain the fixed point at zero: it is forbidden from doing so by the presence of mode 1.  
More precisely, mode 1 is of radius $\radiusone:=|u_{01}|R_\mathrm{c2}=0.494$ and is therefore entirely contained inside the disc inscribed by mode 2.  Also, from Figure~\ref{fig:cgl_spacetime1}, mode 1 is linearly unstable in a two-sided sense.  Hence,  when the system is in the annular region $\radiusone<\sqrt{X^2+Y^2}<\radiustwo$, it is repelled back towards mode 2 (Panels (c)--(d)).  Thus, starting near mode 2, the system makes regular excursions from the periodic orbit but is forbidden from reaching the fixed point at $u=0$.  In conclusion, at large amplitudes such that $\sqrt{X^2+Y^2}>\radiustwo$, the system `snaps back' to the periodic mode at mode 2, while at smaller amplitudes, the system is drawn towards mode 1, repelled from mode 1, and drawn back again towards mode 2. 

The excursions away from mode 2 take place with quasi-regularity (Figure~\ref{fig:cgl_spectral}).  However, a close inspection of this figure shows that the excursions do not have a single regular period.  This is confirmed by examining the power spectrum $|\widehat{X}_\omega|$: there is a single sharp peak at $\omega=2.19\pm 0.01$ in close correspondence with global mode 2 with $\Omega=2.24$ (other power spectra based on $X$ and $Y$ yield the same conclusions).  There is a range of other 
\begin{figure}[htb]
\centering
\subfigure[]{\includegraphics[width=0.45\textwidth]{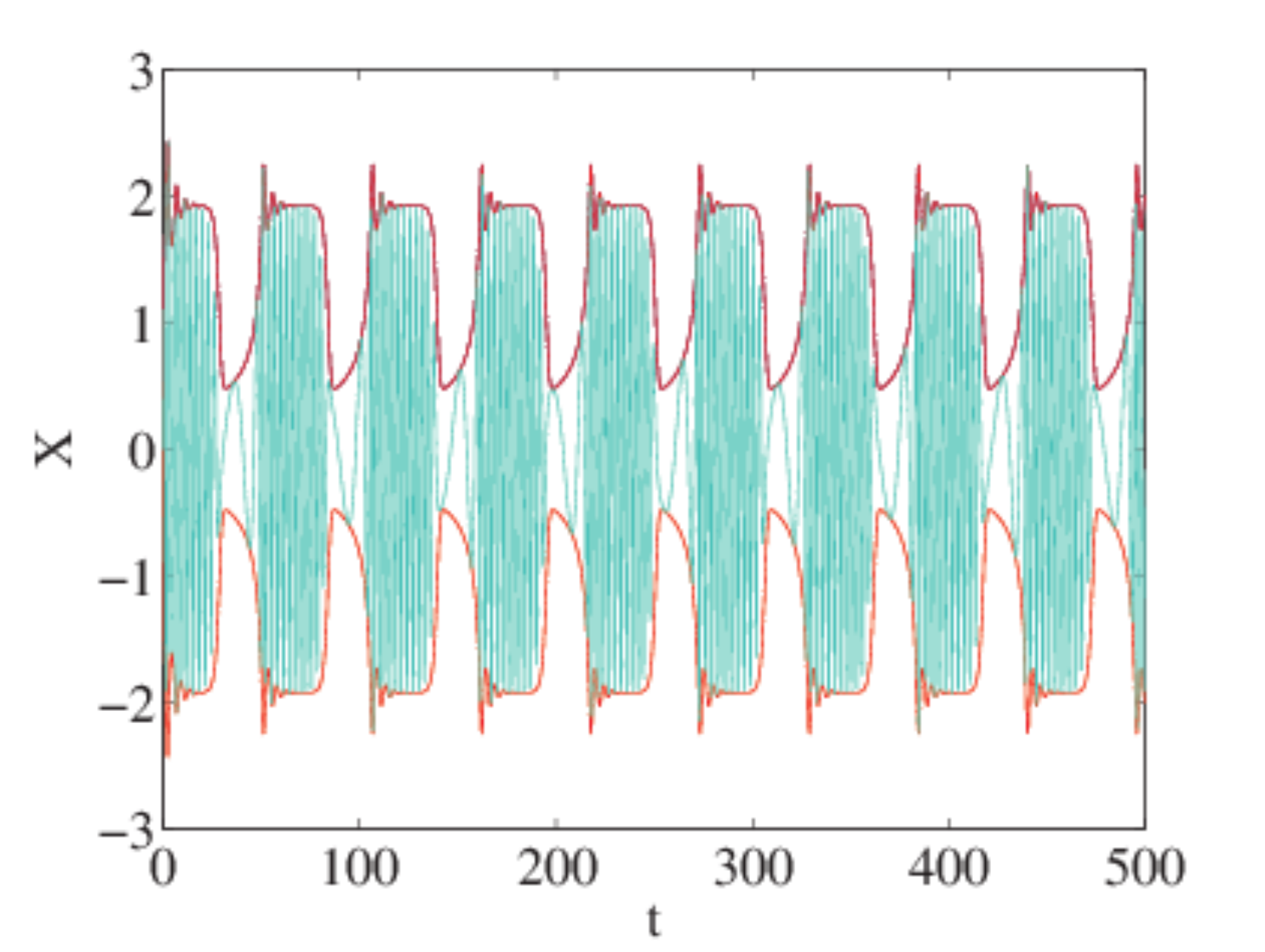}}
\subfigure[]{\includegraphics[width=0.45\textwidth]{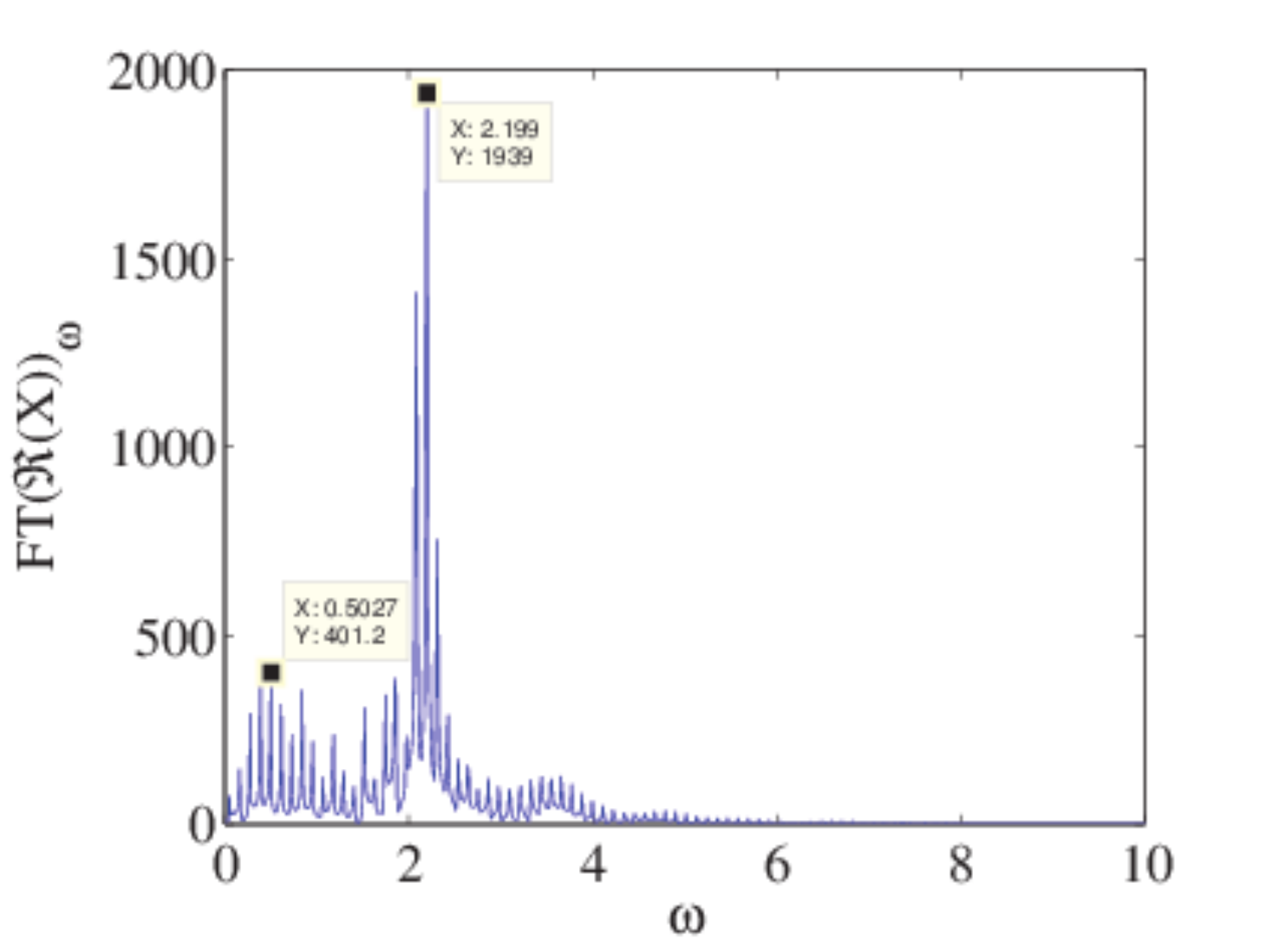}}
\caption{Numerical simulations at $\,\,\rho=3$.  (a) Plot of $\Re(t)$ versus $X(t)$ showing the envelope curves; (b) Fourier transform of the same showing a single sharp maximum corresponding to the global-mode frequency at mode 2.}
\label{fig:cgl_spectral}
\end{figure}
prominent frequencies $\omega=\{0.277,0.390,0.503,0.616,0.842\}$ corresponding to a certain narrow (but still finite) range of frequencies at which the excursions from mode 2 take place.  For these reasons, the motion near the second global mode is chaotic.

In spite of the chaoticity of the investigated trajectories, numerical evidence from Figure~\ref{fig:cgl_spacetime1} suggests that the chaotic motion is bounded, in the sense that $\sup_x|u(x,t)|<\infty$ for all time.  Moreover, investigation of Figure~\ref{fig:cgl_spectral} indicates that only a handful of frequencies are active, indicating the possibility that the dynamics are `slaved' of master modes arising from the linearized dynamics.   Rigorous analysis of this putative slaving phenomenon (using the theoretical tools advocated in Reference~\cite{doering1988low}) will therefore be a subject of a follow-up work on the theoretical side.

\subsection{Other parameter values: the case $U\neq 0$}

A second set of parameters is also taken, with $U=1$ but otherwise identical to Equation~\eqref{eq:params_u0}.  The corresponding value of $k^2$ is $k^2= 0.1976$.  A plot of $\Im(\Omega)$ as a function of $R$ is shown in Figure~\ref{fig:study_u1}(a).
\begin{figure}
	\centering
		\subfigure[]{\includegraphics[width=0.32\textwidth]{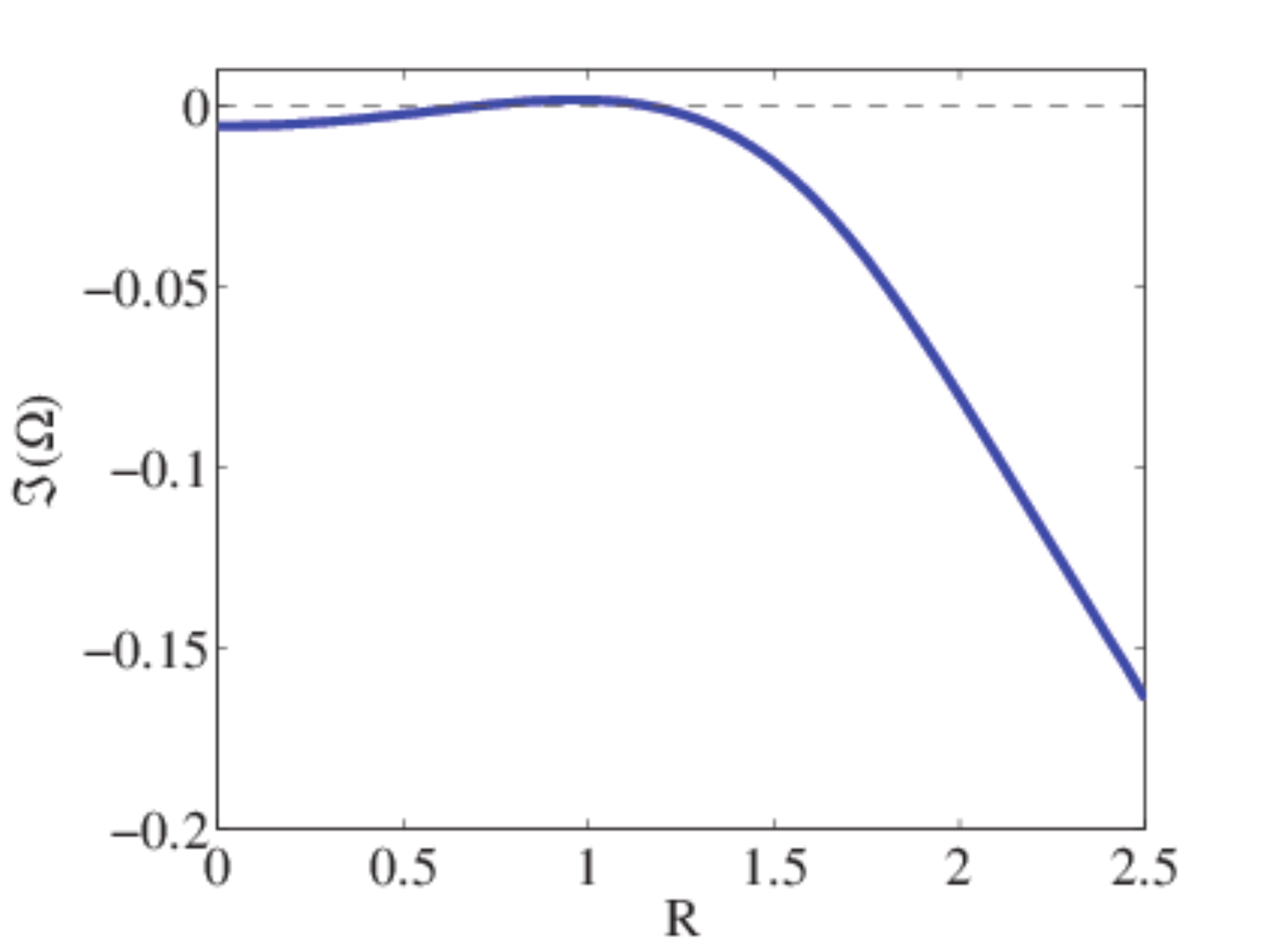}}
		\subfigure[]{\includegraphics[width=0.32\textwidth]{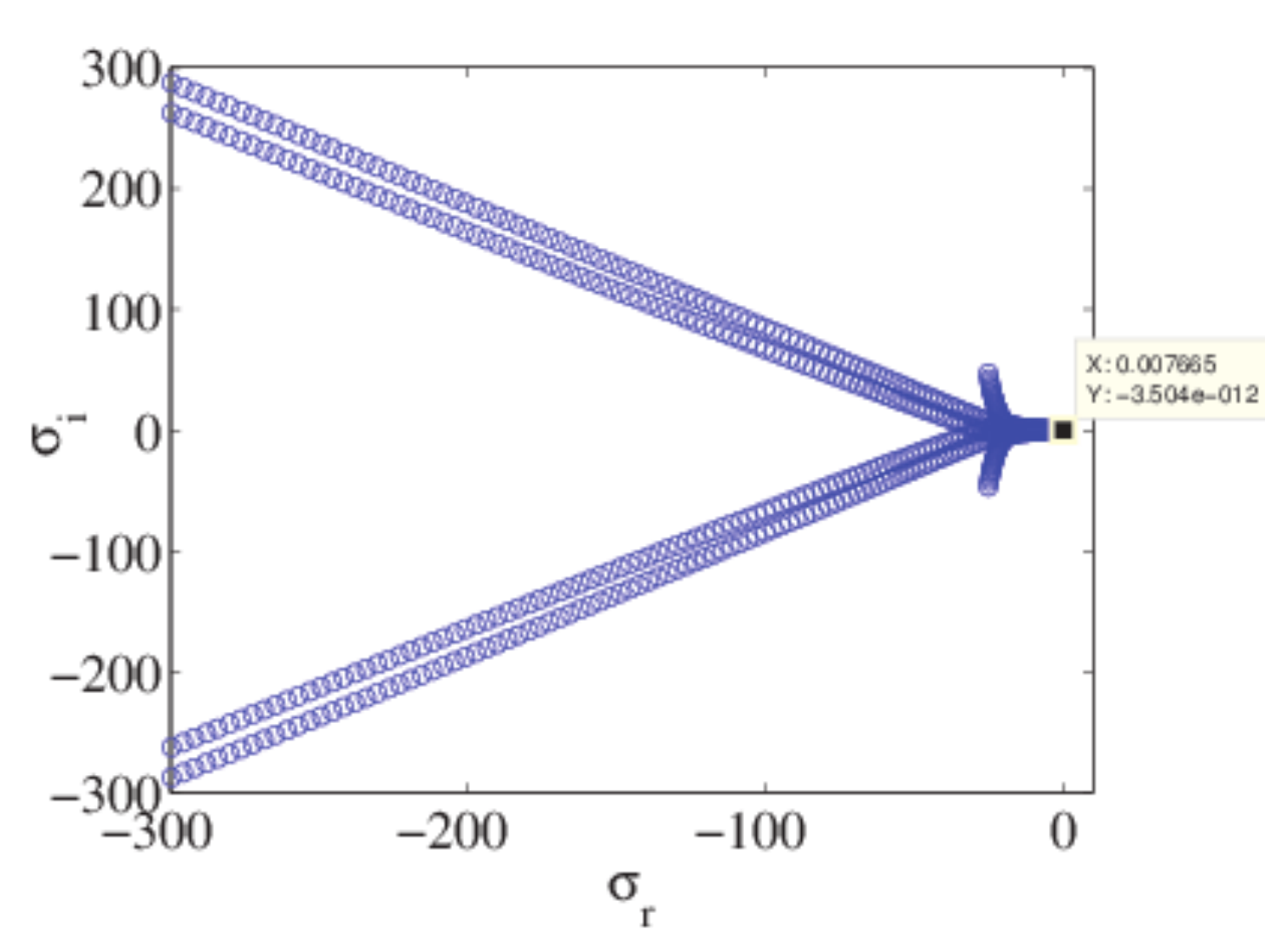}}
		\subfigure[]{\includegraphics[width=0.32\textwidth]{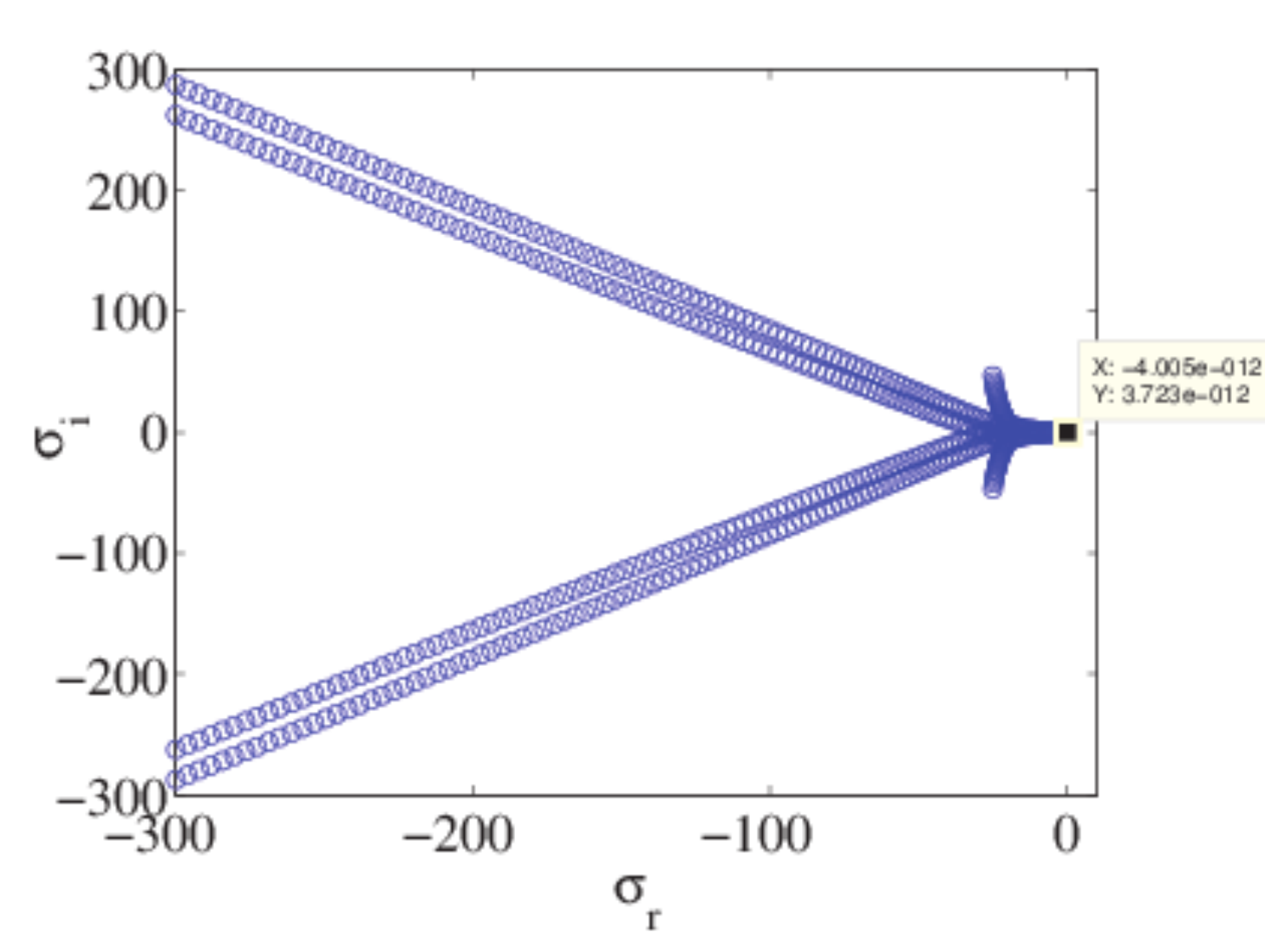}}
		\caption{The case $U=1$.  (a) Plot of $\Im(\Omega)$ as a function of $R$; (b) Floquet analysis at the global mode $R_\mathrm{c}=0.700$; (c) Floquet analysis at the global mode $R_\mathrm{c}=1.166$}
	\label{fig:study_u1}
\end{figure}
Two global modes are found at $R_{\mathrm{c}1}=0.7$ and $R_{\mathrm{c}2}=1.166$, with eigenfrequencies $\Omega=0.2734$ and $0.4064$ respectively.  A Floquet analysis is performed and the first mode is revealed to be unstable (consistent with the perturbation theory already derived) while the second mode is found to be neutrally stable (Figures~\ref{fig:study_u1}(b),(c)).
To examine whether the neutrally stable global mode is selected in a temporally-evolving scenario, we have carried out direct numerical simulations of Equation~\eqref{eq:cgl_nonlinear} using the same method as before, and with the initial condition given by Equation~\eqref{eq:ic_cgl}).  The results are shown in Figure~\ref{fig:dns_u1}.
\begin{figure}
	\centering
		\subfigure[]{\includegraphics[width=0.32\textwidth]{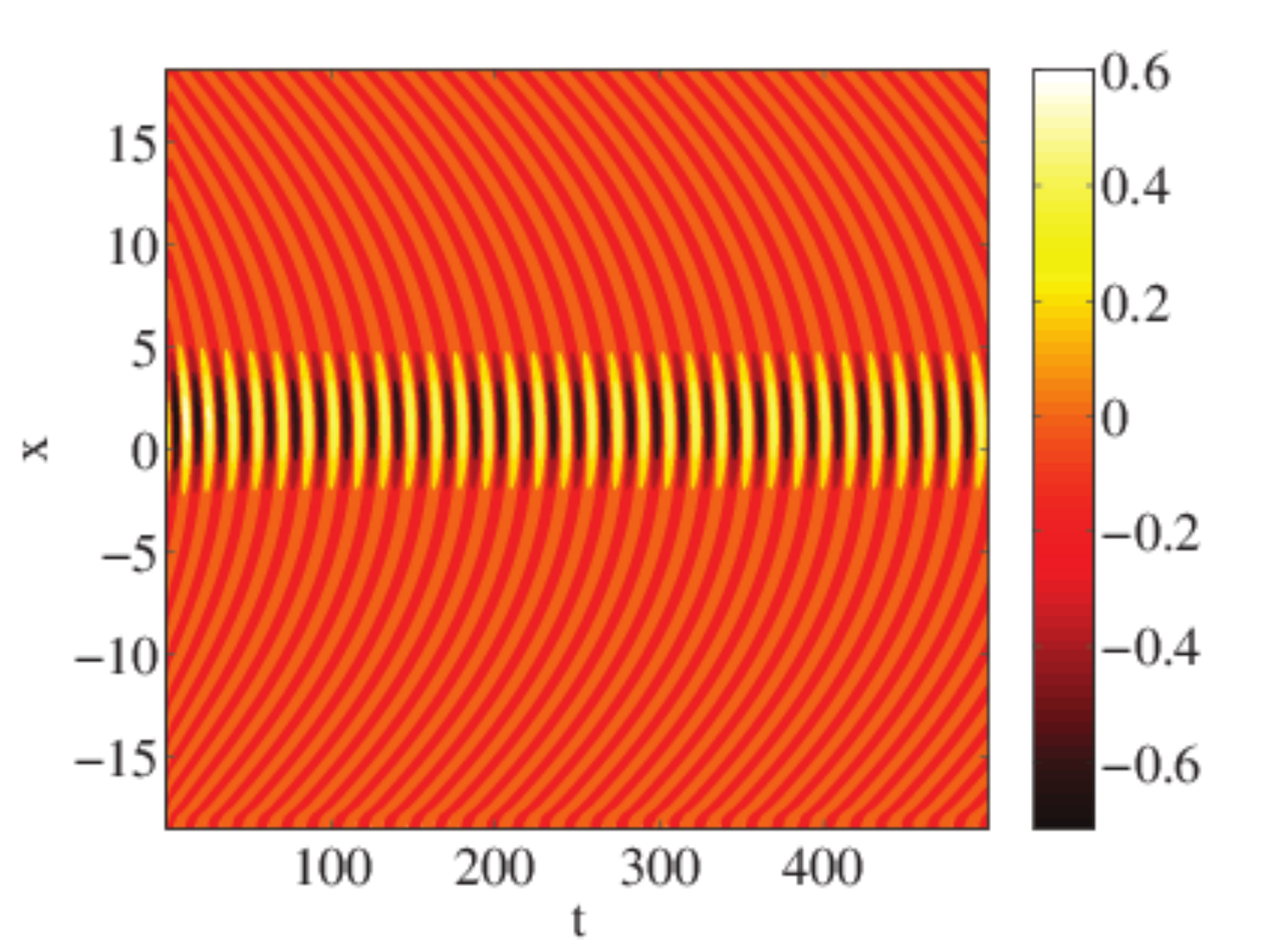}}
		\subfigure[]{\includegraphics[width=0.32\textwidth]{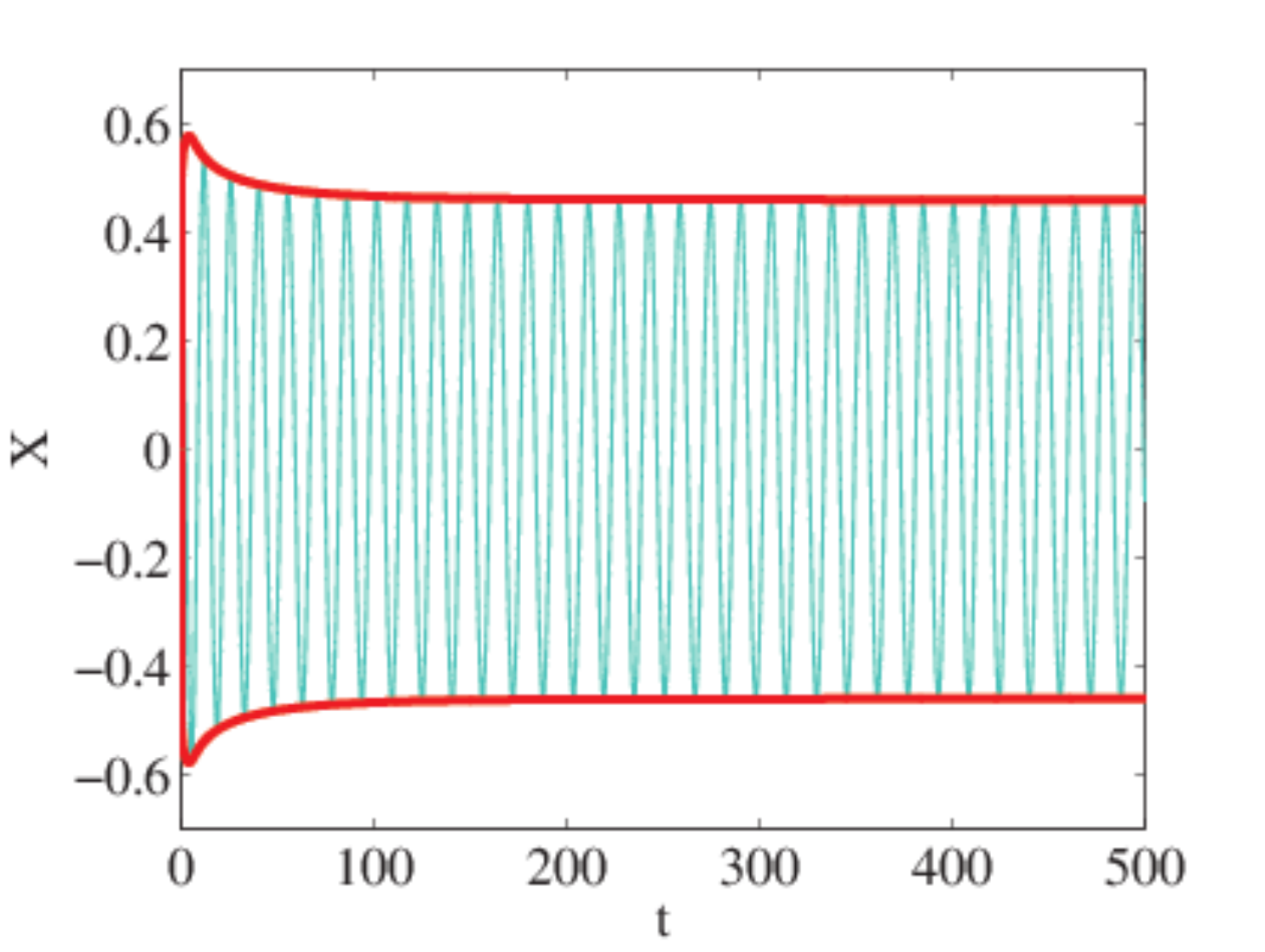}}
		\subfigure[]{\includegraphics[width=0.32\textwidth]{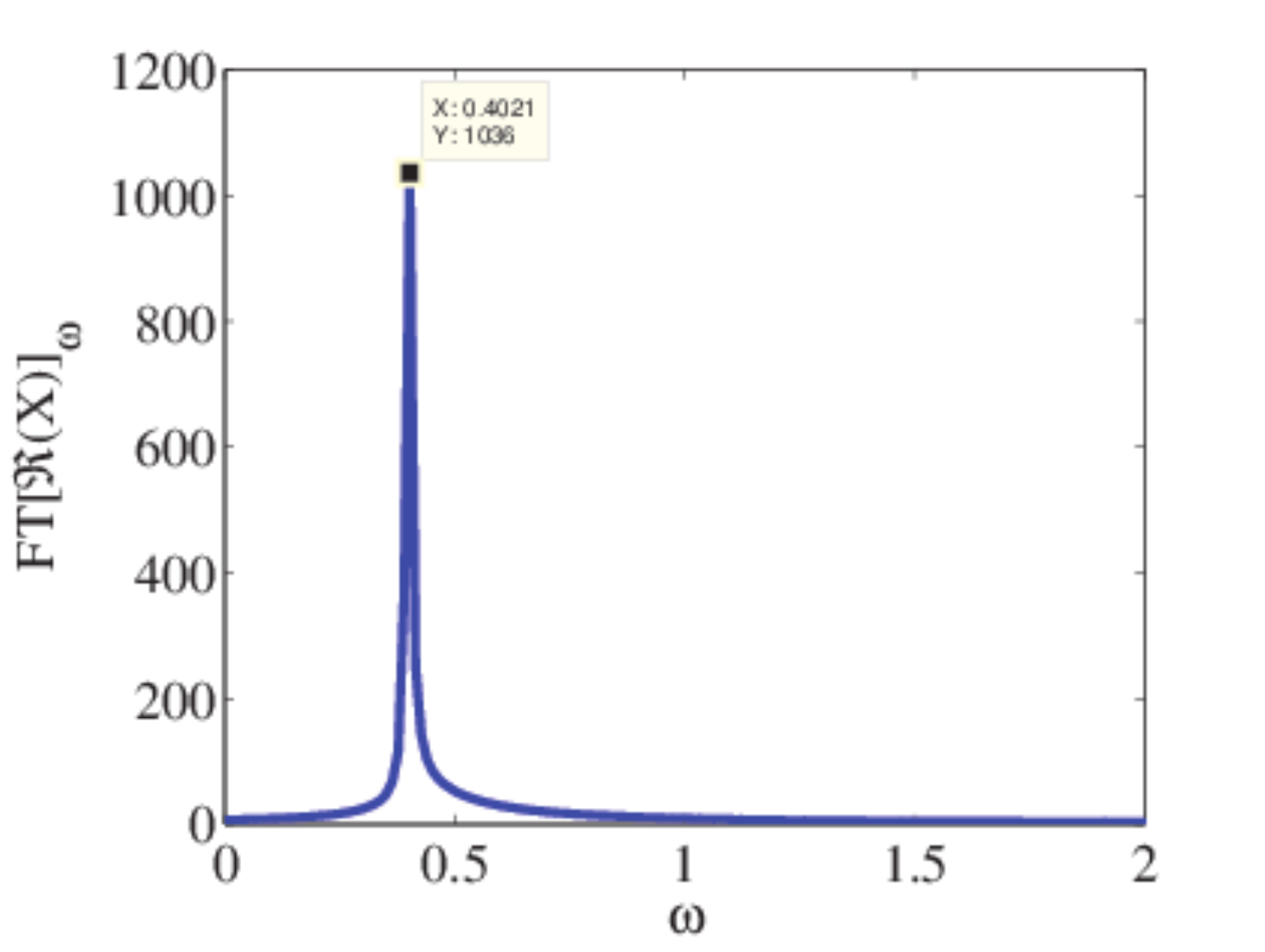}}
		\caption{DNS for the case $U=1$.  (a) Spacetime plot of $|u(x,t)|$; (b) Time series of $X=\Re(u(0,t)$; (c) Fourier transform of (b) showing the importance of the single large-amplitude global mode.}
	\label{fig:dns_u1}
\end{figure}
A spacetime plot in Figure~\ref{fig:dns_u1}(a) shows the selection of a persistent spatial structure, confirmed in Figure~\ref{fig:dns_u1}(b).  The frequency of oscillation coincides with that of the neutrally stable global mode (Figure~\ref{fig:dns_u1}(c)).  Compared to the case with $U=0$, the presence of advection has stabilized the large-amplitude global mode and eliminated the chaos from the system.


\subsection{Critical parameters for the onset of the subcritical nonlinear global mode}

In this subsection we take $\gamma=1-\imag$, $U=1$, and the parameters $\mu_0$ and $\mu_2$ both real and positive.  Based on Equation~\eqref{eq:cgl_eigs_formula_linear_cond}, the neutral curve for the onset of instability in the linear global mode becomes
\[
\mu_{0\mathrm{c}}^{\text{LGM}}=\frac{U^2}{4|\gamma|^2}+\sqrt{\tfrac{1}{2}|\mu_2\gamma|}\cos\left(\arg(\gamma)/2\right),
\]
and we work in a subcritical regime with $\mu_0\leq \mu_{0\mathrm{c}}^{\text{LGM}}(\mu_2)$.  The numerical iterative method outlined at the start of this section is applied to compute the nonlinear global modes.  These nonlinear global modes emerge when $\Im(\Omega)\geq 0$ for at least one $R$-value.  The emergence of the nonlinear global modes occurs precisely when the  correction $\omega_1$ changes sign (Section~\ref{sec:global_modes}, in particular Equation~\eqref{eq:perturbation_theory0}), corresponding to a subcritical Hopf bifurcation~\cite{Chomaz05}.
The critical parameter value $\mu_0^{\text{NGM}}$ such that $\max(\Im(\Omega))=0$ is plotted in Figure~\ref{fig:neutral},
with $\mu_0^{\text{NGM}}\leq \mu_0^{\text{LGM}}$, indicating the subcritical appearance of the nonlinear global modes in the parameter space. 
\begin{figure}[htb]
\centering
\includegraphics[width=0.6\textwidth]{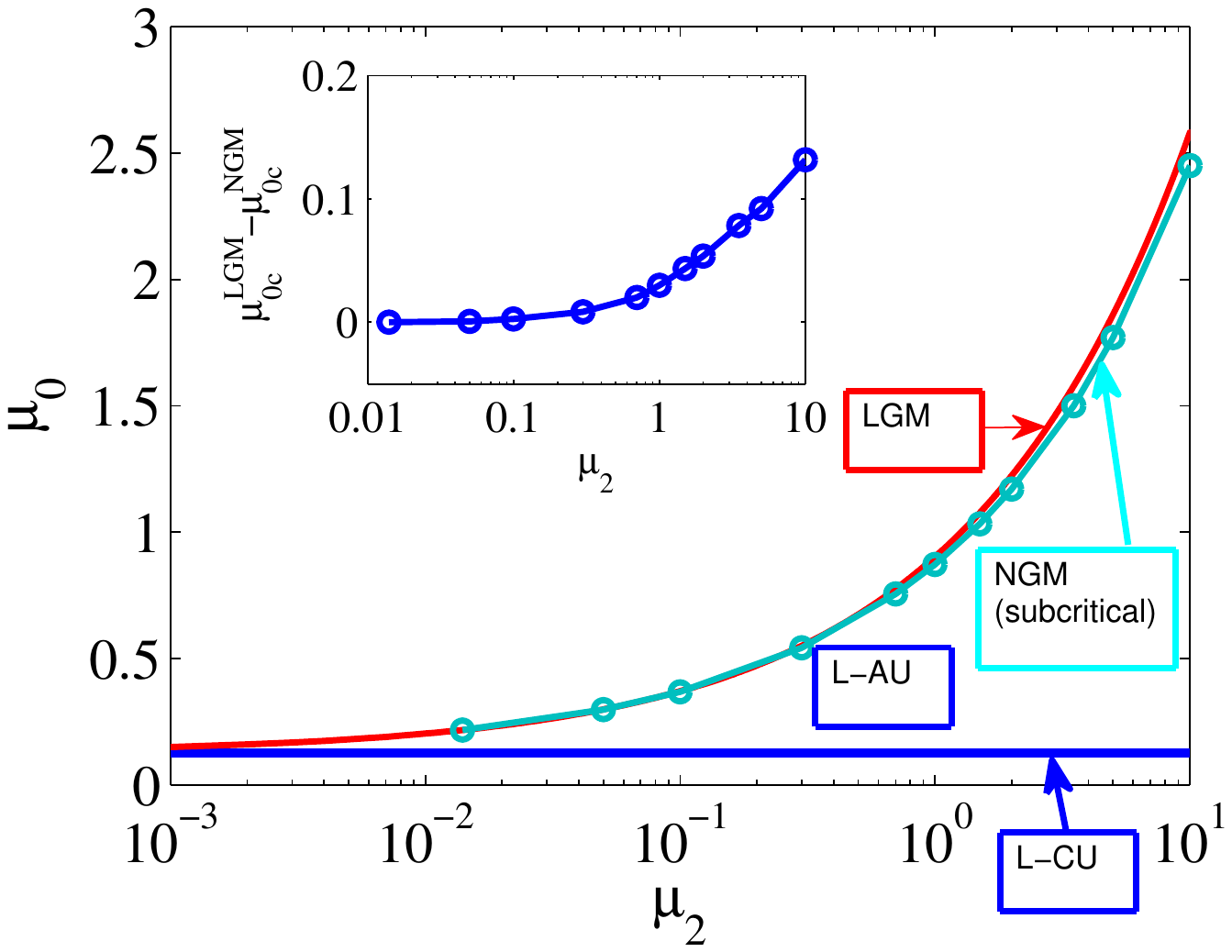}
\caption{The onset of the various disturbances in the $(\mu_2,\mu_0)$ parameter space.  Here, `LGM' means Linear Global Mode (the region above this curve corresponds to linear global instability), `NGM' means nonlinear Global mode, `L-AU' means linearly absolutely unstable, and `L-CU' means linearly convectively unstable.  The inset shows the difference $\mu_{0c}^{\text{LGM}}-\mu_{0c}^{\text{NGM}}$ as a function of $\mu_2$, which vanishes at the point where the nonlinear global mode emerges for the first time, corresponding to a subcritical Hopf bifurcation.}
\label{fig:neutral}
\end{figure}
The difference $\mu_{0c}^{\text{LGM}}-\mu_{0c}^{\text{NGM}}$ as a function of $\mu_2$ for the onset of the first nonlinear global mode is shown in the inset to Figure~\ref{fig:neutral}.  The difference between the two critical values increases as $\mu_2$ increases.  This corresponds to
a path in parameter space that starts in the neighbourhood of the point where $\omega_1=0$ (i.e. where the 
perturbative techniques introduced in Section~\ref{sec:global_modes} apply); the path then tends towards a region of parameter space 
far removed from criticality wherein the amplitude of the nonlinear global mode is finite.  However, in the main part of the figure, the two critical curves `shadow' one another, meaning that the nonlinear global mode is influenced by the corresponding linearized dynamics.  In contrast to the asymptotic global modes studied by Pier and co-workers~\cite{pier1996fully,pier1998}, the onset of linear absolute instability is necessary but insufficient for the onset of nonlinear global modes: one must also be `close' to the neutral curve of the corresponding linear global mode.  Also, the necessary condition previously derived for the onset of nonlinear global modes and linear transient growth (Equation~\eqref{eq:onset_necessary}) is not a sufficient condition for the onset of nonlinear global modes: again, one must also be  `close' to the neutral curve of the corresponding linear global mode  to observe the subcritical nonlinear global mode.

\section{Conclusions}
\label{sec:conc}

Summarizing, we have introduced a new two-level model admitting nonnormality and nonlinearity and highlighting the way in which subcritical transient growth can trigger an unstable global mode, which then leads to unbounded growth -- a rather counterintuitive result in a system that is linearly stable with respect to eigenvalue analysis.  We have argued that this model can be imported directly into an optical experiment, and the hope would be to  such experiments realised as a vivid illustration of the transient-growth mechanism.

Complementing the developments in the two-level system, we have investigated the nonlinear inhomogeneous complex Ginzburg--Landau equation as a spatially-extended analogue of the simple two-level system, and have demonstrated using non-perturbative methods the existence of multiple global modes in such systems.  Theory shows that for the subcritical case, the global mode of smaller amplitude is always unstable, while the numerics show that the global mode of larger amplitude can either be stable or unstable, with the instability controlled in particular by the advection.  The large-amplitude unstable global mode however demonstrates a kind of `one-sided' instability whereby trajectories in an abstract phase space make periodic excursions away from the global mode but return close to the global mode, thereby giving rise to a chaotic dynamics.

%

\subsection*{Acknowledgements}

L\'O N would like to acknowledge relevant discussions with many helpful experts, in particular M. Pietrzyk,  P. D. M. Spelt, Darryl Holm, and B. Pier.

This material is based upon works supported by the European Research 
Council (ERC) under the research project ERC-2011-AdG 290562-MULTIWAVE.

\appendix

\section{Application of the Complex Ginzburg--Landau Equation to paraxial optical beams}
\label{app:paraxial}

In optics, laser beams tend to be well collimated, such that their propagation is almost unidirectional.   The derivation of a model suitable for this system is described briefly.  In free space, and for an electric field oscillating at a single frequency, Maxwell's equations for a scalar electric field reduce to Helmholtz's equation, $\left(\nabla^2+k^2\right)E(\vecx)=0$, where $k=(\omega/c)>0$, $\omega$ is the wave frequency, and $c$ is the speed of light in a vacuum.  Here, information about the polarization of the light is not considered.  The Green's function (outgoing waves) for the Helmoholtz equation in three dimensions is $\mathe^{\imag kr}/r$.  Using this Green's function and Kirchoff's integral formula, the light received at a field point $\vecx$ from an aperture of surface area $S$ is
\begin{equation}
E(\vecx)=\frac{1}{4\pi}\int_S \left[\frac{\mathe^{\imag kr}}{r}\left(\frac{\partial}{\partial n}E(\vecx_0)\right)-E(\vecx_0)\frac{\partial}{\partial n}\left(\frac{\mathe^{\imag kr}}{r}\right)\right]\mathd S(\vecx_0),
\label{eq:kirchoff_huygens}
\end{equation}
where $\partial/\partial n=\widehat{\mathbf{n}}\cdot\nabla$, and where $\widehat{\mathbf{n}}$ is the outward-pointing unit normal of the surface $S$.  Also, $r=|\vecx-\vecx_0|$ is the relative separation between the field point $\vecx$ and the source point $\vecx_0$.  It is assumed that light propagating from the aperture behaves more or less like an outward-moving wave front, such that $\imag k u=-\partial u/\partial n$ on $S$.  Then, Equation~\eqref{eq:kirchoff_huygens} simplifies:
\begin{equation}
E(\vecx)=-\frac{1}{4\pi}\int_S\left[\imag k\left(1+\cos\theta\right)-\frac{\cos\theta}{r}\right]\frac{\mathe^{\imag kr}}{r}E(\vecx_0)\,\mathd S(\vecx_0).
\label{eq:kirchoff_huygens1}
\end{equation}
Moreover, the far-field assumption $kr\gg 1$ is made, such that Equation~\eqref{eq:kirchoff_huygens1} becomes
\begin{equation}
E(\vecx)=\frac{1}{4\pi\imag }\int_S k\left(1+\cos\theta\right)\frac{\mathe^{\imag kr}}{r}E(\vecx_0)\,\mathd S(\vecx_0)
\label{eq:kirchoff_huygens2}
\end{equation}
where $\theta$ is the angle between the normal vector $\widehat{\mathbf{n}}$ and the vector $\vecx-\vecx_0$.
Now, the paraxial approximation is invoked, namely that the beam is well collimated far from the source, such that $E(\vecx)$ is important only close to the optical axis (here chosen arbitrarily to be the $z$-axis).  Thus, we take $(z-z_0)^2\gg (x-x_0)^2+(y-y_0)^2$.  Hence, $\cos\theta\approx 1$, and the aperture $S$ resembles a planar surface with outward-pointing unit normal $\widehat{\mathbf{n}}=\widehat{\mathbf{z}}$.  Equation~\eqref{eq:kirchoff_huygens2} then reads
\[
E(\vecx)=\frac{k}{2\pi\imag }\iint_S \frac{\mathe^{\imag kr}}{r}E(\vecx_0)\,\mathd x_0\mathd y_0.
\]
Finally, the relative distance $r$ is approximated as $z-z_0+(1/2)[(x-x_0)^2+(y-y_0)^2]/(z-z_0)$ and the integral $S$ is taken over $\mathbb{R}^2$ (with the extra condition that $E(\vecx_0)$ is supported only on $S$).  One obtains
\begin{subequations}
\begin{equation}
E(x,y,z)=\frac{k}{2\pi\imag }\frac{\mathe^{\imag k(z-z_0)}}{z-z_0}\iint_{\mathbb{R}^2}E(x_0,y_0,z_0)\exp\left[k\frac{(x-x_0)^2+(y-y_0)^2}{2(z-z_0)}\right] \,\mathd x_0\mathd y_0.
\end{equation}
In this equation, the quantity
\begin{equation}
G_{3D}(x,x_0,y,y_0,z,z_0)=\frac{k}{2\pi\imag}\frac{\mathe^{\imag k(z-z_0)}}{z-z_0}\exp\left[k\frac{(x-x_0)^2+(y-y_0)^2}{2(z-z_0)}\right],\qquad z\neq z_0
\end{equation}%
\label{eq:greens_paraxial1}%
\end{subequations}%
can be viewed as a Green's function that propagates the light from the source point $z_0$ to a field point at $z$, in the paraxial approximation.  Moreover, the Green's function satisfies the paraxial wave equation
\begin{equation}
\left(\nabla_\perp^2+k^2\right)G_{3D}+2\imag k\frac{\partial G_{3D}}{\partial z}=0,\qquad z\neq z_0,\qquad \nabla_\perp^2=\partial_x^2+\partial_z^2.
\label{eq:greens_paraxial2}
\end{equation}
Thus, $E$-fields that satisfy the paraxial condition also solve Equation~\eqref{eq:greens_paraxial2}.  For such fields, one writes $E(x,y,z)=u(x,y)\mathe^{\imag kz}$, such that
\begin{equation}
\nabla_\perp^2 u+2\imag k\frac{\partial u}{\partial z}=0.
\label{eq:paraxial}
\end{equation}
By inspection of Equation~\eqref{eq:greens_paraxial1}, it is clear that the solution $u(x,y,z)$ to Equation~\eqref{eq:paraxial} is propagated from $z_0$ to $z$ as follows:
\[
u(x,y,z)=\iint_{\mathbb{R}^2} \mathcal{G}_{3D}(x,x_0,y,y_0,z,z_0)u(x_0,y_0,z_0)\,\mathd x_0,\mathd y_0,
\]
where
\[
\mathcal{G}_{3D}(x,x_0,y,y_0,z,z_0)=\frac{k}{2\pi\imag}\exp\left[k\frac{(x-x_0)^2+(y-y_0)^2}{2(z-z_0)}\right],\qquad z\neq z_0.
\]

In two dimensions with coordinates $(x,z)$, similar arguments hold: the Green's function for the full Helmholtz equation is $(\imag/4)H_0^{(1)}(k\rho)$, where $H_0^{(1)}$ is the Hankel function of the first kind, and where $\rho=\sqrt{(x-x_0)^2+(z-z_0)^2}$.  Following the same procedures as before, the paraxial wave equation~\eqref{eq:paraxial} is recovered, and the Green's function
\[
\mathcal{G}_{2D}(x,x_0,z,z_0)=\sqrt{\frac{k}{2\pi\imag(z-z_0)}}\mathe^{\imag k(x-x_0)^2/2(z-z_0)}
\]
is obtained, such that $u(x,z)$ is obtained from the same field at $z_0$ by convolution,
\[
u(x,z)=\int_{-\infty}^\infty \mathcal{G}_{2D}(x,x_0,z,z_0)u(x_0,z_0)\,\mathd x.
\]

In optical media, the propagation constant $k$ is no longer real or constant.  This can be seen as follows.  First, we note that $k=\omega/v$, where $v$ is the speed of light in the medium, and $v=c/\widetilde{n}$ is the (complex) index of refraction, $\widetilde{n}=n+\imag \kappa$.  Thus, $k=(\omega/c)(n+\imag \kappa)$, and the right-propagating electric field contains a factor $\mathe^{\imag k z}=\mathe^{-\imag (\omega/c)nz}\mathe^{-(\omega/c)\kappa z}$.  Thus, $\kappa>0$ corresponds to absorption, whereby the beam's energy is absorbed by the medium.
In a Gaussian duct, both the index of refraction $n$ and the absorption coefficient $\kappa$ vary in the direction perpendicular to propagation.  We consider the situation wherein $n=n_0-(n_2/2)x^2$ and where $n_0,n_2>0$, meaning that the index of refraction is maximal at the duct centreline $x=0$.  In this scenario, optical rays will be trapped near the duct centreline~\cite{SiegmanBook}.  In the opposite scenario ($n_2<0$), the variation of the index of refraction will cause incident rays to diverge from the duct centreline.  Also, we consider $\kappa=\kappa_0+(\kappa_2/2)x^2$, where $\kappa_0>0$.  Thus,
\[
k=(\omega/c)(n+\imag\kappa)=(\omega/c)(n_0+\imag\kappa_0)+\frac{\omega x^2}{2c}\left(-n_2+\imag \kappa_2\right).
\]
We also write this as $k=K_0+(1/2)K_2x^2$, with 
\[
K_0=(\omega/c)(n_0+\imag \kappa_0),\qquad 
K_2=(\omega/c)(-n_2+\imag \kappa_2),
\]
and we work in the limit where $|K_0|\gg |K_2x^2|$, such that $K^2\approx K_0^2+K_0 K_2 x^2$.  
Thus, the paraxial wave equation~\eqref{eq:paraxial} becomes
\[
\frac{\partial^2 u}{\partial x^2}+2\imag K_0\frac{\partial u}{\partial z}+K_0 K_2x^2 u=0,\qquad z>0,
\]
where we work with two-dimensional beams.  Re-arranging, one obtains
\begin{equation}
\frac{\partial u}{\partial z}=\frac{\imag }{2K_0}\frac{\partial^2 u}{\partial x^2}+\tfrac{1}{2}\imag K_2 x^2 u,\qquad z>0,
\label{eq:cgl_optical}
\end{equation}
which is precisely the linearized inhomogeneous complex Ginzburg-Landau equation~\eqref{eq:qm_model} with $\gamma=\imag/(2K_0)$ and $\mu_2=-\imag K_2$, and $\mu_0=U=0$.    Such a system is referred to as a `Gaussian duct'~\cite{SiegmanBook}.

%
%

\end{document}